\documentclass[a4paper,10pt,numbers=noenddot,twoside=on,headinclude=true,footinclude=false,headsepline=true]{scrartcl}

\usepackage[utf8]{inputenc}
\usepackage[T1]{fontenc}
\usepackage{textcomp}
\usepackage[english]{babel}
\usepackage{color}
\usepackage{amssymb}
\usepackage{amsmath}
\usepackage{amsfonts}
\usepackage{amsopn}
\usepackage{amsbsy}
\usepackage[overload,ntheorem]{empheq}
\usepackage{enumerate}
\usepackage[plainpages=false,pdfpagelabels,breaklinks=true,pdfborderstyle={/S/U/W 1.2}]{hyperref}
\usepackage{graphics}
\usepackage{units}
\usepackage[cal=boondoxo,bb=boondox]{mathalfa}
\usepackage{dsfont}

\numberwithin{equation}{section} 
\numberwithin{table}{section} 
\numberwithin{figure}{section}

\usepackage[amsmath,thmmarks,thref,hyperref]{ntheorem}

\theoremstyle{plain}
\newtheorem{theorem}{Theorem}[section]
\newtheorem{definition}[theorem]{Definition}
\newtheorem{lemma}[theorem]{Lemma}
\newtheorem{corollary}[theorem]{Corollary}
\newtheorem{proposition}[theorem]{Proposition}
\newtheorem{assumption}[theorem]{Assumption}

\theorembodyfont{\upshape}
\newtheorem{remark}[theorem]{Remark}

\theoremstyle{nonumberplain}

\theoremsymbol{\ensuremath{_\Box}}
\newtheorem{proof}{Proof}

\usepackage[style=alphabetic,
	citestyle=alphabetic,
	firstinits=true,
	backend=biber,
	hyperref=auto,
	sorting=nyt,
	maxcitenames=3,
	maxbibnames=99,
	minnames=3,
	arxiv=pdf,
	url=false,
	isbn=false,
	dateabbrev=true,
	datezeros=true,
	bibencoding=utf8]{biblatex}
\newbibmacro{string+doi}[1]{%
	\iffieldundef{doi}{#1}{\href{http://dx.doi.org/\thefield{doi}}{#1}}}
\DeclareFieldFormat
	{title}{\usebibmacro{string+doi}{\mkbibemph{#1}}\isdot}
\DeclareFieldFormat
	[article,inbook,incollection,inproceedings,patent,thesis,unpublished]
	{title}{\usebibmacro{string+doi}{\mkbibemph{#1}}\isdot}
\DeclareFieldFormat
	[article,inbook,incollection,inproceedings,patent,thesis,unpublished]
	{pages}{#1}
\DeclareFieldFormat
	[article]
	{journaltitle}{#1\isdot}
\DeclareFieldFormat
	[article]
	{volume}{\textbf{#1}}
\DefineBibliographyStrings{english}{pages={}}
\DeclareBibliographyDriver{article}{%
	\usebibmacro{bibindex}%
	\usebibmacro{begentry}%
	\usebibmacro{author/translator+others}%
	\setunit{\labelnamepunct}\newblock
	\usebibmacro{title}%
	\usebibmacro{journal+issuetitle}%
	\setunit*{\addcomma\space}
	\printfield{pages}
	\setunit*{\addcomma\space}
	\printfield{year}
	\usebibmacro{finentry}%
}

\newbibmacro*{journal+issuetitle}{%
	\usebibmacro{journal}%
	\setunit*{\addspace}%
	\iffieldundef{series}
	{}
	{\newunit
		\printfield{series}%
		\setunit{\addspace}}%
	\iffieldundef{volume}
		{}
		{\printfield{volume}%
		\setunit{\addspace}}%
	\setunit*{\addcolon\space}%
	\usebibmacro{issue}%
	\newunit}

\usepackage[scaled]{beramono}
\usepackage{helvet}
\usepackage[charter]{mathdesign} 
\SetMathAlphabet{\mathcal}{normal}{OMS}{cmsy}{m}{n} 
\SetMathAlphabet{\mathcal}{bold}{OMS}{cmsy}{m}{n} 
\providecommand{\ie}{i.~e.~}
\providecommand{\eg}{e.~g.~}
\providecommand{\cf}{cf.~}

\providecommand{\R}{\mathbb{R}}

\providecommand{\C}{\mathbb{C}}
\renewcommand{\C}{\mathbb{C}}

\providecommand{\ii}{\mathrm{i}}
\providecommand{\e}{\mathrm{e}}
\renewcommand{\Re}{\mathrm{Re} \,}
\renewcommand{\Im}{\mathrm{Im} \,}
\providecommand{\Hil}{\mathcal{H}}

\providecommand{\eps}{\varepsilon}

\providecommand{\ker}{\mathrm{ker} \, }
\providecommand{\ran}{\mathrm{ran} \, }

\providecommand{\ker}{\mathrm{ker} \,}

\providecommand{\dd}{\mathrm{d}}
\providecommand{\id}{\mathds{1}}

\providecommand{\Fourier}{\mathcal{F}}

\providecommand{\abs}[1]{\left \lvert #1 \right \rvert}
\providecommand{\sabs}[1]{\lvert #1 \vert}

\providecommand{\norm}[1]{\left \lVert #1 \right \rVert}
\providecommand{\snorm}[1]{\lVert #1 \rVert}
\providecommand{\bnorm}[1]{\bigl \lVert #1 \bigr \rVert}

\providecommand{\Norm}[1]{{\left \vert \kern-0.25ex \left \vert \kern-0.25ex \left \vert #1 \right \vert \kern-0.25ex \right \vert \kern-0.25ex \right \vert}}
\providecommand{\sNorm}[1]{{\vert \kern-0.25ex \vert \kern-0.25ex \vert #1 \vert \kern-0.25ex \vert \kern-0.25ex \vert}}
\providecommand{\bNorm}[1]{{\bigl \vert \kern-0.25ex \bigl \vert \kern-0.25ex \bigl \vert #1 \bigr \vert \kern-0.25ex \bigr \vert \kern-0.25ex \bigr \vert}}
\providecommand{\BNorm}[1]{{\Bigl \vert \kern-0.25ex \Bigl \vert \kern-0.25ex \Bigl \vert #1 \Bigr \vert \kern-0.25ex \Bigr \vert \kern-0.25ex \Bigr \vert}}
\providecommand{\scpro}[2]{\left \langle #1 , #2 \right \rangle}
\providecommand{\sscpro}[2]{\langle #1 , #2 \rangle}
\providecommand{\bscpro}[2]{\bigl \langle #1 , #2 \bigr \rangle}
\providecommand{\Bscpro}[2]{\Bigl \langle #1 , #2 \Bigr \rangle}

\providecommand{\sopro}[2]{\vert #1 \rangle \langle #2 \vert}

\providecommand{\Rot}{\mathrm{Rot} \,}
\providecommand{\Div}{\mathrm{Div} \,}

\providecommand{\Jphys}{\mathcal{J}_+}
\providecommand{\Maux}{M^{\mathrm{aux}}}

\usepackage{slashbox}
\usepackage{caption}
\usepackage{subcaption}

\usepackage{tabularx}

\bibliography{/Users/max/Library/texmf/tex/latex/max/bibliography}

\subtitle{How to Make Quantum-Wave Analogies Rigorous}
\title{The Schrödinger Formalism of Electro- \\ magnetism and Other Classical Waves}
\author{Giuseppe De Nittis${}^1$ \& Max Lein${}^2$}

\begin{document}

\maketitle
\vspace{-9mm}
\begin{center}
	${}^1$ Facultad de Matemáticas \& Instituto de Física, 
	Pontificia Universidad Católica de Chile \linebreak
	Avenida Vicuña Mackenna 4860, 
	Santiago, 
	Chile \linebreak
	{\footnotesize \href{mailto:gidenittis@mat.uc.cl}{\texttt{gidenittis@mat.uc.cl}}}
	\medskip
	\\
	${}^2$ Advanced Institute of Materials Research,  
	Tohoku University \linebreak
	2-1-1 Katahira, Aoba-ku, 
	Sendai, 980-8577, 
	Japan \linebreak
	{\footnotesize \href{mailto:maximilian.lein.d2@tohoku.ac.jp}{\texttt{maximilian.lein.d2@tohoku.ac.jp}}}
\end{center}
\begin{abstract}
	This paper systematically develops the Schrödinger formalism that is valid also for \emph{gyrotropic} media where the material weights $W = 
	\left (
	\begin{smallmatrix}
		\eps & \chi \\
		\chi^* & \mu \\
	\end{smallmatrix}
	\right ) \neq \overline{W}$ are complex. 
	This is a non-trivial extension of the Schrödinger formalism for \emph{non-}gyrotropic media (where $W = \overline{W}$) that has been known since at least the 1960s \cite{Wilcox:scattering_theory_classical_physics:1966,Kato:scattering_theory_two_Hilbert_spaces:1967}. 
	Here, Maxwell's equations are rewritten in the form $\ii \partial_t \Psi = M \Psi$ where the selfadjoint (hermitian) \emph{Maxwell operator} $M = W^{-1} \, \Rot \big \vert_{\omega \geq 0} = M^*$ takes the place of the Hamiltonian and $\Psi$ is a complex wave representing the physical field $(\mathbf{E},\mathbf{H}) = 2 \Re \Psi$. Writing Maxwell's equations in Schrödinger form gives us access to the rich toolbox of techniques initially developed for quantum mechanics and allows us to apply them to classical waves. 
	To show its utility, we explain how to identify conserved quantities in this formalism. Moreover, we sketch how to extend our ideas to other classical waves. 
\end{abstract}

\newpage
\tableofcontents

\section{Introduction} 
\label{intro}
The idea to realize quantum phenomena with classical waves, light in particular, harkens back to the beginnings of modern quantum mechanics in the 1920s. The heuristic basis for such \emph{quantum-light analogies} is the fact that both theories describe waves. Quantum-wave analogies can be read both ways: on the one hand they serve as a source for inspiration and guide physicists in the search for novel effects. On the other experiments with classical waves are often easier to realize, allow for more latitude in the design and can image the waves directly (in stark contrast to \emph{quantum} wave functions which are not observable as a matter of principle). To mention but one example is an experiment performed by Schneider et al \cite{Schneider_et_al:spin_wave_packet_dynamics_experiment:2008} who have managed to capture the dynamics of a spin wave packet in an yttrium iron garnet film. 
Such measurements allow physicists to probe the transition from “quantum” to “classical dynamics” and offer new insights into the inner workings of various quantum phenomena. 

Recent advances in manufacturing \cite{Fischer_et_al:direct_laser_writing:2013,Kadic_Bueckmann_Schittny_Wegener:metamaterials_beyond_electromagnetism:2013,Zieger_et_al:direct_laser_writing:2017} coupled with new theoretical insights gave new impulses in the search for novel media for classical waves. One particularly notable work is that of Raghu and Haldane \cite{Raghu_Haldane:quantum_Hall_effect_photonic_crystals:2008} who predicted the existence of unidirectional, backscattering-free edge modes in two-dimensional gyrotropic photonic crystals in analogy to the Quantum Hall Effect. Their work jumpstarted the search for topological phenomena in classical waves, which culminated in a number of spectacular experiments for electromagnetic \cite{Wang_et_al:unidirectional_backscattering_photonic_crystal:2009,Rechtsman_Zeuner_et_al:photonic_topological_insulators:2013,Lu_et_al:experimental_observation_Weyl_points:2015}, mechanical \cite{Kane_Lubensky:topological_boundary_modes_isostatic_lattices:2014,Paulose_Chen_Vitelli:topological_modes_dislocations_mechanical_metamaterials:2015,Suesstrunk_Huber:mechanical_topological_insulator:2015} and acoustic waves \cite{Xiao_et_al:geometric_phase_acoustic_systems:2015}. Furthermore, others \cite{Shindou_et_al:chiral_magnonic_edge_modes:2013,Jin_Lu_et_al:topological_magnetoplasmon:2016} have been proposed but not yet observed experimentally. 

Maxwell's equations are not mere partial differential equations, they contain more mathematical structure that can be exploited; they can be studied within the Lagrangian or Hamiltonian formalism (see \eg \cite[Chapter~2.3 and Chapter~13.1]{Spohn:dynamics_charged_particles:2004}, \cite[Chapter~1.6]{Marsden_Ratiu:intro_mechanics_symmetry:1999} or \cite[Chapter~12.7]{Jackson:electrodynamics:1998}). But when one wants to determine how and \emph{to what extent} specific quantum-wave analogies can be made rigorous, one natural starting point is to write the dynamical equation \emph{in the form of a Schrödinger equation}, 
\begin{align*}
	\ii \partial_t \psi(t) = M \psi(t) - \ii J(t) 
	, 
	&&
	\psi(0) = \phi \in \Hil
	, 
\end{align*}
where $M = M^*$ is a selfadjoint operator on a \emph{complex} Hilbert space $\Hil$ and $J(t)$ a current. 

The general idea is by no means new: the founding fathers of quantum mechanics were well-aware that the \emph{in vacuo} Maxwell equations are \emph{the} relativistically covariant equations for a massless spin-$1$ particle \cite[pp.~151 and 198]{Wigner:representations_Lorentz_group:1939}. Similarly, the Schrödinger formalism has been employed to develop scattering theory for electromagnetic \cite{Birman_Solomyak:L2_theory_Maxwell_operator:1987} and other classical waves \cite{Wilcox:scattering_theory_classical_physics:1966,Kato:scattering_theory_two_Hilbert_spaces:1967,Schulenberger_Wilcox:completeness_wave_operators:1971,Reed_Simon:scattering_theory_wave_equations:1977} among other things \cite{Birman_Solomyak:L2_theory_Maxwell_operator:1987,Figotin_Klein:localization_classical_waves_II:1997,Kuchment:math_photonic_crystals:2001,DeNittis_Lein:sapt_photonic_crystals:2013}. However, these works crucially assume that the material weights $W = \left (
\begin{smallmatrix}
	\eps & \chi \\
	\chi^* & \mu \\
\end{smallmatrix}
\right ) = \overline{W}$, which phenomenologically describe the properties of the medium, be \emph{real} — and this \emph{excludes gyrotropic media} where  $W \neq \overline{W}$. During our investigation of the ray optics limit \cite[Section~2.2]{DeNittis_Lein:ray_optics_photonic_crystals:2014}, we discovered that \emph{the Maxwell equations for media with complex weights that are commonly used in the physics literature are incomplete.} What has been missing is the restriction of these equations to (complex) waves composed only of non-negative or non-positive frequencies to ensure that they describe the propagation of \emph{real-valued} fields (compare also with the discussion in \cite[Section~6]{DeNittis_Lein:symmetries_Maxwell:2014}). However, in \cite[Section~2.2]{DeNittis_Lein:ray_optics_photonic_crystals:2014} we only gave the solution and promised a first principles derivation from Maxwell's equations for linear, dispersive media later. With this article we will pay our debt to the reader. 

One of the problems which arises is to balance the real-valuedness of physical fields with the need to work with complex Hilbert spaces in order to be able to reach into the rich mathematical tool box developed for quantum systems. While the idea to write an electromagnetic wave as the real part of a complex wave is part and parcel of every course on electromagnetism, usually no explicit reference to the Hilbert spaces these waves belong to is made, something that needs to be specified in the Schrödinger formalism; this is especially relevant when the waves propagate in media rather than vacuum. The two earlier attempts to adapt the Cartan-Altland-Zirnbauer classification scheme \cite{Altland_Zirnbauer:superconductors_symmetries:1997,Schnyder_Ryu_Furusaki_Ludwig:classification_topological_insulators:2008,Chiu_Teo_Schnyder_Ryu:classification_topological_insulators:2016} to photonic crystals by the authors \cite{DeNittis_Lein:symmetries_Maxwell:2014} and to coupled mechanical oscillators by Süsstrunk and Huber \cite{Suesstrunk_Huber:classification_mechanical_metamaterials:2016} neglected to take the reality of physical fields properly into account. Explaining why this is so and how to develop a “fully native” topological classification for classical waves will be one of the main points of this paper. More specifically, our main goals are: 
\begin{enumerate}[(1)]
	\item We derive the \emph{Schrödinger formalism} for electromagnetic waves propagating in linear, dispersionless, lossless media that takes the real-valuedness of the physical fields into account. Recasting the dynamical equation allows one to readily adapt and apply techniques initially developed for quantum mechanics to classical electromagnetism. At least for gyrotropic media, this is new. 
	\item To demonstrate the utility of the Schrödinger formalism, we show how it helps to identify physically meaningful conserved quantities. Moreover, we contrast and compare that to works that use the Lagrangian mechanics and Noether's theorem. 
	\item And we show that the above scheme can be implemented also for other classical waves, including coupled mechanical oscillators and transverse acoustic waves. 
\end{enumerate}

\paragraph{Outline of the paper} 
First, we derive Maxwell's equations for linear, possibly gyrotropic media in Section~\ref{Maxwell_equations} from linear dispersive constitutional relations. Then we introduce the Schrödinger formalism of Maxwell's equations (Section~\ref{Schroedinger}), the main matter of this work. In Section~\ref{comparison} we contrast and compare the Schrödinger formalism of electromagnetism with previous approaches in the literature and quantum mechanics. To show the utility of the Schrödinger formalism, we include a discussion of conserved quantities for electromagnetism in media in Section~\ref{conserved}. Moreover, to illustrate that the ideas of Section~\ref{Schroedinger} apply to other classical wave equations, we sketch how to formulate the Schrödinger formalism for certain linear acoustic equations (Section~\ref{other_waves}). 

\paragraph{Acknowledgements} 
\label{intro:acknowledgements}
G.~D.~research is supported by the grant \emph{Iniciaci\'{o}n en Investigaci\'{o}n 2015} - $\text{N}^{\text{o}}$ 11150143 funded by FONDECYT. 
M.~L.~has been supported by JSPS through a WAKATE~B grant (grant number 16K17761) and a Fusion grant from the WPI-AIMR. M.~L.\ is indebted to Kostya Bliokh for the discussions which initiated this work and appreciates his comments on the first version of the manuscript. The authors would also like to thank Ilya Dodin for very useful references, and generous explanations about the applicability of our framework to problems from plasma physics. 
\section{Deriving Maxwell's equations for linear, non-dispersive media} 
\label{Maxwell_equations}
The purpose of this section is to systematically derive Maxwell's equations for gyrotropic media where the material weights are complex. The difficulty here is that there are actually \emph{two distinct} Maxwell's equations, one for positive and one for negative frequencies. This ambiguity is not a mere footnote, but a manifestation of the real-valuedness of the physical fields $(\mathbf{E},\mathbf{H})$. To derive these equations systematically, we implement the following strategy:%
\begin{enumerate}[(1)]
	\item We start with Maxwell's equations for linear, \emph{dispersive} media in the time domain, equation~\eqref{Maxwell_equations:eqn:full_equations}. 
	\item Taking the Fourier transform in time, these are then rewritten in the frequency domain.
	\item Assuming that the material weights are approximately constant for waves from a narrow frequency window, we replace the frequency-dependent weights with the frequency independent ones. 
	\item We then Fourier transform back to obtain \emph{two} equations~\eqref{Maxwell_equations:eqn:approximate_Maxwell_equations}, one for positive and one for negative frequencies. 
\end{enumerate}
These equations form the natural starting point for the Schrödinger formalism we will present later in Section~\ref{Schroedinger}; we will postpone a discussion of necessary technical assumptions and other mathematical details until then. Given this goal, we will also take the opportunity here and show how a Hilbert space structure emerges naturally in this context; the scalar product~\eqref{Maxwell_equations:eqn:scalar_product_Hil_C} is obtained from the field energy.

\subsection{Fundamental equations and notions} 
\label{Maxwell_equations:fundamental_equations}
During our investigation of the ray optics limit \cite{DeNittis_Lein:ray_optics_photonic_crystals:2014}, we learnt that it was helpful to compare not just (i) states and (ii) the dynamical equation for the classical waves with the corresponding quantum system, but also include (iii) observables and (iv)~additional information on the physics (such as typical states, regimes, length and time scales) in our considerations. This paper will initially focus on (i) and (ii) and revisit (iii) and (iv) in Sections~\ref{conserved}. 

The propagation of electromagnetic waves in media is described by 4 equations that we collectively refer to as \emph{Maxwell's equations in media}: 
\begin{subequations}\label{Maxwell_equations:eqn:full_equations}
	\begin{align}
		\frac{\partial}{\partial t} \left (
		\begin{matrix}
			\mathbf{D}(t) \\
			\mathbf{B}(t) \\
		\end{matrix}
		\right ) &= \left (
		\begin{matrix}
			+ \nabla \times \mathbf{H}(t) \\
			- \nabla \times \mathbf{E}(t) \\
		\end{matrix}
		\right ) - \left (
		\begin{matrix}
			j^D(t) \\
			0 \\
		\end{matrix}
		\right )
		&&
		\mbox{(dynamical equation)}
		\label{Maxwell_equations:eqn:full_equations:dynamics}
		\\
		\left (
		\begin{matrix}
			\nabla \cdot \mathbf{D}(t) \\
			\nabla \cdot \mathbf{B}(t) \\
		\end{matrix}
		\right ) &= \left (
		\begin{matrix}
			\rho^D(t) \\
			0 \\
		\end{matrix}
		\right )
		&&
		\mbox{(constraint equation)}
		\label{Maxwell_equations:eqn:full_equations:constraint}
		\\
		\bigl ( \mathbf{D}(t) , \mathbf{B}(t) \bigr ) &= \bigl ( \mathcal{W}(\mathbf{E},\mathbf{H}) \bigr )(t) 
		&&
		\mbox{(constitutive relations)}
		\label{Maxwell_equations:eqn:full_equations:constitutive_relations}
		\\
		\partial_t \rho^D + \nabla \cdot j^D &= 0 
		&&
		\mbox{(charge conservation)}
		\label{Maxwell_equations:eqn:full_equations:charge_conservation}
	\end{align}
\end{subequations}
These equations involve the \emph{electromagnetic field} $(\mathbf{E},\mathbf{H})$ and the \emph{auxiliary fields} $(\mathbf{D},\mathbf{B})$ consisting of the electric displacement $\mathbf{D}$ and the magnetic induction $\mathbf{B}$; by definition $(\mathbf{D},\mathbf{B})$ are such that \eqref{Maxwell_equations:eqn:full_equations:dynamics} is satisfied, and the dependence of $(\mathbf{D},\mathbf{B})$ on the electromagnetic field is codified in the \emph{constitutive relations}~\eqref{Maxwell_equations:eqn:full_equations:constitutive_relations} which provide an effective description of the interaction between the electromagnetic field and the microscopic charges in the medium \cite[Chapter~6]{Jackson:electrodynamics:1998}. Hereon after, we shall assume that the fields $\bigl ( \mathbf{E}(t) , \mathbf{H}(t) \bigr )$ have square integrable amplitudes, \ie the waves are elements of the Banach space of physical fields 
\begin{align}
	\Hil_{\mathrm{phys}} = L^2(\R^3,\R^6) 
	= \Bigl \{ (\mathbf{E},\mathbf{H}) : \R^3 \longrightarrow \R^6 \; \; \big \vert \; \; \int_{\R^3} \dd x \bigl ( \sabs{\mathbf{E}(x)}^2 + \sabs{\mathbf{H}(x)}^2 \bigr ) < \infty \Bigr \}
	. 
\end{align}
For the media of interest here this is equivalent to considering waves with finite field energy. 

Sources, that is charge densities $\pmb{\rho} = (\rho^D,0)$ and current densities $\mathbf{J} = (j^D,0)$, can drive electromagnetic fields, and appear as inhomogeneities in the differential equations; they have to satisfy \emph{local charge conservation}~\eqref{Maxwell_equations:eqn:full_equations:charge_conservation}. Mathematically, we could admit magnetic charge densities $\rho^B$ and current densities $j^B$, but those are unphysical. Because electromagnetic media do not necessarily carry waves of all frequencies and directions, \eg photonic crystals and dispersive media may sport photonic band gaps, only the contribution to the current density that is proportional to excitable states is relevant; we will go into more detail at the end of Section~\ref{Maxwell_equations:hilbert_space} and condense it down to Assumption~\ref{Maxwell_equations:assumption:current_density} for the current density. 

For linear \emph{dispersive} media the auxiliary fields at time $t$ (see \eg \cite[Chapter~7.10]{Jackson:electrodynamics:1998}, \cite[Section~2]{Butcher_Cotter:nonlinear_optics:1991} or \cite[Chapter~8.1]{Van_Bladel:electromagnetism:2007})
\begin{align}
	\left (
	\begin{matrix}
		\mathbf{D}(t,x) \\
		\mathbf{B}(t,x) \\
	\end{matrix}
	\right ) &= \bigl ( \mathcal{W}(\mathbf{E},\mathbf{H}) \bigr )(t) 
	= \bigl ( W \ast (\mathbf{E},\mathbf{H}) \bigr )(t)
	= \int_{-\infty}^t \dd s \, W(t - s,x) \, \left (
	\begin{matrix}
		\mathbf{E}(s,x) \\
		\mathbf{H}(s,x) \\
	\end{matrix}
	\right )
	, 
	\label{Maxwell_equations:eqn:linear_dispersive_constitive_relations}
\end{align}
are computed from the past history $(-\infty,t] \ni t \mapsto \bigl ( \mathbf{E}(t) , \mathbf{H}(t) \bigr )$ of the electromagnetic field by convolving $\bigl ( \mathbf{E}(t) , \mathbf{H}(t) \bigr )$ with the $6 \times 6$ matrix-valued function 
\begin{align*}
	W(t,x) &= \left (
	\begin{matrix}
		\eps(t,x) & \chi^{EH}(t,x) \\
		\chi^{HE}(t,x) & \mu(t,x) \\
	\end{matrix}
	\right )
\end{align*}
in the time variable; it is customarily split into $3 \times 3$ blocks, namely the electric permittivity $\eps$, the magnetic permeability $\mu$, and the magnetoelectric couplings $\chi^{EH}$ and $\chi^{HE}$. $\mathcal{W}$ describes the response of the medium to the impinging electromagnetic waves, and necessarily needs to satisfy the following two conditions: 
\begin{assumption}\label{Maxwell_equations:assumption:reality_causality}
	\begin{enumerate}[(a)]
		\item $W = \overline{W}$ is real (electromagnetic fields are \emph{real}). 
		\item $W(t) = 0$ for all $t < 0$ (\emph{causality}, events in the presents are only influenced by the past).
	\end{enumerate}
\end{assumption}
Typically, $W(t,x) = W_{\mathrm{nd}}(x) \, \delta(t) + W_{\mathrm{dis}}(t,x)$ is split into a non-dispersive part $W_{\mathrm{nd}}$ that describes the contribution which reacts instantaneously, and a \emph{bona fide} dispersive part $W_{\mathrm{dis}}$ that depends on the past field configuration, 
\begin{align*}
	\left (
	\begin{matrix}
		\mathbf{D}(t,x) \\
		\mathbf{B}(t,x) \\
	\end{matrix}
	\right ) &= \bigl ( \mathcal{W}(\mathbf{E},\mathbf{H}) \bigr )(t)
	= W_{\mathrm{nd}}(x) \, \left (
	\begin{matrix}
		\mathbf{E}(t,x) \\
		\mathbf{H}(t,x) \\
	\end{matrix}
	\right ) + \bigl ( W_{\mathrm{dis}} \ast (\mathbf{E},\mathbf{H}) \bigr )(t)
	. 
\end{align*}
The fact that the auxiliary fields $\bigl ( \mathbf{D}(t) , \mathbf{B}(t) \bigr )$ not only depend on the electromagnetic field $\bigl ( \mathbf{E}(t) , \mathbf{H}(t) \bigr )$ at time $t$ but on the field at times in the past means \emph{the medium has a memory}. Note that non-dispersive media have no memory, because the auxiliary fields $\bigl ( \mathbf{D}(t) , \mathbf{B}(t) \bigr )$ are computed from the electromagnetic field $\bigl ( \mathbf{E}(t) , \mathbf{H}(t) \bigr )$ at time $t$. 

Later on in Section~\ref{Maxwell_equations:heuristic_monochromatic_approximation} we will explain how to neglect dispersion by replacing the convolution integral in \eqref{Maxwell_equations:eqn:linear_dispersive_constitive_relations} with 
\begin{align}
	\left (
	\begin{matrix}
		\mathbf{D}(t,x) \\
		\mathbf{B}(t,x) \\
	\end{matrix}
	\right ) &= W(x) \, \left (
	\begin{matrix}
		\mathbf{E}(t,x) \\
		\mathbf{H}(t,x) \\
	\end{matrix}
	\right )
	\label{Maxwell_equations:eqn:linear_nondispersive_constitive_relations}
\end{align}
where $W(x)$ is now independent of time. 

The constitutive relations of \emph{non}-linear media are typically expressed as power series in the fields where the higher-order terms are either iterated convolutions in time akin to \eqref{Maxwell_equations:eqn:linear_dispersive_constitive_relations} (see \eg \cite{Babin_Figotin:nonlinear_Maxwell_1:2001,Babin_Figotin:nonlinear_Maxwell:2003}) or instantaneous in time similar to \eqref{Maxwell_equations:eqn:linear_nondispersive_constitive_relations} \cite[§~80]{Landau_Lifshitz:electromagnetism_media:1963}. In principle the analysis made in this paper can be adapted to the non-linear setting. Babin and Figotin, for example, have shown that in a certain scaling limit, the non-linear Maxwell equations reduce to a non-linear Schrödinger equation \cite{Babin_Figotin:nonlinear_Maxwell_4:2005} and taken the reality of electromagnetic waves correctly into account; while strictly speaking their derivation was done for \emph{non}-gyrotropic media, in view of our results, their arguments should extend straightforwardly to media with complex material weights. 
\medskip

\noindent
\emph{Electromagnetic observables}, constituent~(iii) of the physical theory of electromagnetism, are \emph{functionals of the fields} $\mathcal{F} : L^2(\R^3,\R^6) \longrightarrow \R$, they map electromagnetic fields onto real numbers. In stark contrast to quantum mechanics, the components of $\mathbf{E}$ and $\mathbf{H}$ themselves are accessible to experiment. While observables do not have to be linear or quadratic in the fields, the most frequently considered ones are. We postpone a more in-depth discussion of electromagnetic observables and examples to Section~\ref{conserved}. 

\subsection{Maxwell's equations in the frequency domain} 
\label{Maxwell_equations:Fourier}
To see how the response of the medium depends on the light's frequency, we apply the (inverse) Fourier transform in time
\begin{align}
	\widehat{\Psi}(\omega) = \bigl ( \Fourier^{\, -1} \Psi \bigr )(\omega) 
	= \frac{1}{\sqrt{2\pi}} \int_{\R} \dd t \, \e^{+ \ii t \omega} \, \Psi(t) 
	. 
	\label{Maxwell_equations:eqn:Fourier_transform_time}
\end{align}
to Maxwell's equations in order to decompose $\bigl ( \mathbf{E}(t) , \mathbf{H}(t) \bigr )$ into its frequency components: 
\begin{subequations}\label{Maxwell_equations:eqn:Maxwell_frequency_domain}
	\begin{align}
		\omega \, \widehat{W}(\omega) \, \bigl ( \widehat{\mathbf{E}}(\omega) , \widehat{\mathbf{H}}(\omega) \bigr ) &= \Rot \bigl ( \widehat{\mathbf{E}}(\omega) , \widehat{\mathbf{H}}(\omega) \bigr ) - \ii \, \widehat{\mathbf{J}}(\omega) 
		\label{Maxwell_equations:eqn:Maxwell_frequency_domain:dynamics}
		\\
		\Div \bigl ( \widehat{W}(\omega) \, \bigl ( \widehat{\mathbf{E}}(\omega) , \widehat{\mathbf{H}}(\omega) \bigr ) \bigr ) &= \widehat{\pmb{\rho}}(\omega) 
		\label{Maxwell_equations:eqn:Maxwell_frequency_domain:constraint}
		\\
		\bigl ( \widehat{\mathbf{D}}(\omega) , \widehat{\mathbf{B}}(\omega) \bigr ) &= \widehat{W}(\omega) \, \bigl ( \widehat{\mathbf{E}}(\omega) , \widehat{\mathbf{H}}(\omega) \bigr )
		\label{Maxwell_equations:eqn:Maxwell_frequency_domain:constitutive_relations}
		\\
		\omega \, \widehat{\pmb{\rho}}(\omega) + \ii \, \Div \widehat{\mathbf{J}}(\omega) &= 0 
		\label{Maxwell_equations:eqn:Maxwell_frequency_domain:charge_conservation}
	\end{align}
\end{subequations}
To simplify notation, we have abbreviated the Fourier transformed electromagnetic field by $\widehat{\Psi} = \bigl ( \widehat{\psi^E} , \widehat{\psi^H} \bigr )$, the charge density by $\widehat{\pmb{\rho}} = \bigl ( \widehat{\rho^D} , \widehat{\rho^B} \bigr )$ and the current density by $\widehat{\mathbf{J}} = \bigl ( \widehat{j^D} , \widehat{j^B} \bigr )$, and introduced the \emph{free Maxwell} or rotation operator 
\begin{align*}
	\Rot = \left (
	\begin{matrix}
		0 & + \ii \nabla^{\times} \\
		- \ii \nabla^{\times} & 0 \\
	\end{matrix}
	\right ) 
	. 
\end{align*}
Moreover, $\Div(\mathbf{D},\mathbf{B}) = \bigl ( \nabla \cdot \mathbf{D} , \nabla \cdot \mathbf{B} \bigr )$ consists of two copies of the usual divergence, and for convenience we have absorbed a factor of $\sqrt{2\pi}$ into the definition of $\widehat{W} = \sqrt{2\pi} \, \Fourier^{\, -1} W$ that stems from $\Fourier^{\, -1} (f \ast g) = \sqrt{2\pi} \, \Fourier^{\, -1} f \, \Fourier^{\, -1} g$. The fact that $\mathcal{W}$ is real, $W = \overline{W}$, yields a relation of $\widehat{W}$ at $\pm \omega$, 
\begin{align}
	\widehat{W}(-\omega) = \overline{\widehat{W}(+\omega)} 
	. 
	\label{Maxwell_equations:eqn:symmetry_hatW}
\end{align}
The Fourier transforms of other real quantities such as $\widehat{\Psi}(-\omega) = \overline{\widehat{\Psi}(\omega)}$, $\widehat{\pmb{\rho}}(-\omega) = \overline{\widehat{\pmb{\rho}}(\omega)}$ and $\widehat{\mathbf{J}}(-\omega) = \overline{\widehat{\mathbf{J}}(+\omega)}$ also satisfy relation~\eqref{Maxwell_equations:eqn:symmetry_hatW}. Real electromagnetic fields $\bigl ( \overline{\mathbf{E}} , \overline{\mathbf{H}} \bigr ) = (\mathbf{E},\mathbf{H})$ are necessarily composed of complex waves with frequencies $\pm \omega$, 
\begin{align*}
	\bigl ( \mathbf{E}(t) , \mathbf{H}(t) \bigr ) &= \frac{1}{\sqrt{2\pi}} \int_{\R} \dd \omega \, \e^{+ \ii t \omega} \, \bigl ( \widehat{\mathbf{E}}(\omega) , \widehat{\mathbf{H}}(\omega) \bigr )
	.  
\end{align*}
Fortunately, the information contained in the positive and negative frequency contributions  
\begin{align}
	\Psi_{\pm}(t) &= \frac{1}{\sqrt{2\pi}} \int_0^{+\infty} \dd \omega \, \e^{\pm \ii t \omega} \, \bigl ( \widehat{\mathbf{E}}(\pm \omega) , \widehat{\mathbf{H}}(\pm \omega) \bigr )
	\label{Maxwell_equations:eqn:complex_positive_negative_frequency_waves}
\end{align}
is redundant because the real-valuedness of $(\mathbf{E},\mathbf{H})$ implies their \emph{relative phase is locked}, 
\begin{align}
	\Psi_- = \overline{\Psi_+} 
	, 
	\label{Maxwell_equations:eqn:phase_locking_condition}
\end{align}
and it suffices to focus on \eg $\Psi_+$ only. Similarly, we denote the positive and negative frequency contributions to the current and charge densities with $J_{\pm}$ and $\rho_{\pm}$, and also here positive and negative frequency parts satisfy the phase locking condition~\eqref{Maxwell_equations:eqn:phase_locking_condition}. Note that even though $\Psi_+ \neq \Psi_- = \overline{\Psi_+}$, their real parts necessarily agree, 
\begin{align}
	2 \Re \Psi_+ &= \Psi_+ + \overline{\Psi_+} 
	= (\mathbf{E},\mathbf{H})
	= \overline{\Psi_-} + \Psi_- 
	= 2 \Re \Psi_-
	. 
	\label{Maxwell_equations:eqn:real_states_determined_by_positive_frequency_part}
\end{align}
Viewing Maxwell's equations in the frequency domain allows us clearly identify different frequency regimes where the medium has certain characteristic properties. In addition to reality and causality, repeated as (a) and (b) below, physics imposes three more constraints on the material weights (see \eg \cite[Section~6.8]{Jackson:electrodynamics:1998}, \cite{Silveirinha:Z2_topological_index_continuous_photonic_systems:2016} or \cite[Section~2]{Butcher_Cotter:nonlinear_optics:1991}): 
\begin{assumption}[Material weights]\label{Maxwell_equations:assumption:5_minimal_assumptions}
	\begin{enumerate}[(a)]
		\item The physical fields are \emph{real-valued}, \ie the material weights satisfy $\widehat{W}(-\omega) = \overline{\widehat{W}(+\omega)}$. 
		\item \emph{Causality} holds, \ie $W(t) = 0$ for all $t < 0$, and therefore $\omega \mapsto \widehat{W}(\omega)$ extends to a function which is \emph{analytic} on the upper half complex plane $\C^+$. 
		\item At \emph{large frequencies} $\omega \to \pm \infty$, any medium behaves \emph{like vacuum}, \ie $\lim_{\omega \to \pm \infty} \widehat{W}(\omega) = \id$ in suitable units. 
		\item For \emph{small frequencies}, the material weights are real and the bianotropic tensors $\widehat{\chi}^{EH}$ and $\widehat{\chi}^{HE}$ which quantify the magnetoelectric coupling vanish, 
		\begin{align*}
			\lim_{\omega \to 0} \widehat{W}(\omega) &= \lim_{\omega \to 0} \left (
			\begin{matrix}
				\widehat{\eps}(\omega) & \widehat{\chi}^{EH}(\omega) \\
				\widehat{\chi}^{HE}(\omega) & \widehat{\mu}(\omega) \\
			\end{matrix}
			\right )
			= \left (
			\begin{matrix}
				\widehat{\eps}(0) & 0 \\
				0 & \widehat{\mu}(0) \\
			\end{matrix}
			\right )
			= \left (
			\begin{matrix}
				\overline{\widehat{\eps}(0)} & 0 \\
				0 & \overline{\widehat{\mu}(0)} \\
			\end{matrix}
			\right )
			. 
		\end{align*}
		\item The \emph{field energy density}
		\begin{align}
			\Bscpro{\bigl ( \widehat{\mathbf{E}}(\omega) , \widehat{\mathbf{H}}(\omega) \bigr )}{\tfrac{\dd}{\dd \omega} \bigl ( \omega \, \widehat{W}(\omega) \bigr ) \, \bigl ( \widehat{\mathbf{E}}(\omega) , \widehat{\mathbf{H}}(\omega) \bigr )} > 0 
			\label{Maxwell_equations:eqn:field_energy_density_frequency}
		\end{align}
		is \emph{positive definite} for all $\omega$, \ie there exists a constant $c > 0$ so that 
		\begin{align*}
			\frac{\dd}{\dd \omega} \bigl ( \omega \, \widehat{W}(\omega) \bigr ) \geq c \, \id 
			. 
		\end{align*}
	\end{enumerate}
\end{assumption}
%

\subsection{The Hilbert space of electromagnetic fields with finite field energy} 
\label{Maxwell_equations:hilbert_space}
The field energy density~\eqref{Maxwell_equations:eqn:field_energy_density_frequency} not only singles out electromagnetic fields with finite field energy, but it also gives rise to a scalar product — and therefore, a Hilbert space. This is important as a mathematical definition of linear partial differential equations requires us to fix a vector space on which these equations act. The purpose of this subsection is to give a suitable re-interpretation of equations~\eqref{Maxwell_equations:eqn:complex_positive_negative_frequency_waves}–\eqref{Maxwell_equations:eqn:real_states_determined_by_positive_frequency_part}, and while parts of it seem technical and perhaps unnecessary, we will need all of this to arrive at physically relevant equations for gyrotropic, non-dispersive media. 

The \emph{total field energy} of an electromagnetic wave is computed by integrating up all the frequency components, 
\begin{align}
	\mathcal{E}(\mathbf{E},\mathbf{H}) &= \frac{1}{2} \int_{\R} \dd \omega \, \Bscpro{\bigl ( \widehat{\mathbf{E}}(\omega) , \widehat{\mathbf{H}}(\omega) \bigr )}{\tfrac{\dd}{\dd \omega} \bigl ( \omega \, \widehat{W}(\omega) \bigr ) \, \bigl ( \widehat{\mathbf{E}}(\omega) , \widehat{\mathbf{H}}(\omega) \bigr )} 
	= \bnorm{(\mathbf{E},\mathbf{H})}_{\C}^2
	, 
	\label{Maxwell_equations:eqn:integrated_field_energy_frequency}
\end{align}
which is defined in terms of the scalar product for $\C^6$-valued waves on $\R^3$, 
\begin{align*}
	\scpro{\Phi}{\Psi} = \int_{\R^d} \dd x
	 \, \Phi(x) \cdot \Psi(x)
	= \int_{\R^d} \dd x \, \sum_{j = 1}^6 \overline{\Phi_j(x)} \, \Psi_j(x)
	. 
\end{align*}
\emph{Complex} linear combinations of such (time- and space-dependent) fields $\bigl ( \mathbf{E}(t,x) , \mathbf{H}(t,x) \bigr )$ which solve Maxwell's equations~\eqref{Maxwell_equations:eqn:full_equations} in the absence of sources form the vector space 
\begin{align}
	\Hil_{\C} :& \negmedspace = \Bigl \{ \Psi \in L^{\infty} \bigl ( \R , L^2(\R^3,\C^6) \bigr ) \; \; \big \vert \; \; \mathcal{E}(\Psi) < \infty 
	,
	\Bigr .
	\notag \\
	&\qquad \qquad \qquad \qquad \qquad \qquad \quad \; 
	\Bigl .
	\mbox{$\Psi(t)$ solves \eqref{Maxwell_equations:eqn:full_equations:dynamics} with $\mathbf{J}(t) = 0$}
	\Bigr \} 
	. 
	\label{Maxwell_equations:eqn:complixification_Hilbert_space}
\end{align}
where “distances” are naturally measured in terms of the \emph{(field) energy norm} $\snorm{\Psi}_{\C} = \sqrt{\mathcal{E}(\Psi)}$ and the \emph{energy scalar product} 
\begin{align}
	\scpro{\Phi}{\Psi}_{\C} = \frac{1}{2} \int_{\R} \dd \omega \, \Bscpro{\widehat{\Phi}(\omega)}{\tfrac{\dd}{\dd \omega} \bigl ( \omega \, \widehat{W}(\omega) \bigr ) \, \widehat{\Psi}(\omega)}
	\label{Maxwell_equations:eqn:scalar_product_Hil_C}
\end{align}
gives meaning to the notion of orthogonality. The index $\C$ emphasizes that $\Hil_{\C}$ is a \emph{complex} Hilbert space as the natural basis with respect to which to expand are \emph{complex exponentials} $\e^{- \ii t \omega}$. 

Going from the Maxwell equations~\eqref{Maxwell_equations:eqn:full_equations} in the time domain to Maxwell's equations~\eqref{Maxwell_equations:eqn:Maxwell_frequency_domain} in the frequency domain has a natural interpretation here in terms of Hilbert spaces: the Fourier transform in time decomposes 
\begin{align}
	\widehat{\Hil}_{\C} = \Fourier^{\, -1} \Hil_{\C} = \int_{\R}^{\oplus} \dd \omega \, \Hil(\omega)
\end{align}
into a direct integral of Hilbert spaces. The mathematically inclined reader may look up the somewhat technical definition in \cite[Part~II, Chapter~2, Section~3]{Dixmier:von_Neumann_algebras:1981} which explains how the $\Hil(\omega)$ are glued together, but in a nutshell it means the following: the Fourier transform in time decomposes electromagnetic fields into frequency components $\bigl ( \widehat{\mathbf{E}}(\omega) , \widehat{\mathbf{H}}(\omega) \bigr )$ which themselves are elements of Hilbert spaces 
\begin{align}
	\Hil(\omega) :& \negmedspace = \Bigl \{ \widehat{\Psi}(\omega) : \R^3 \longrightarrow \C^6 \; \; \big \vert \; \; \norm{\Psi}_{\C} < \infty 
	,
	\Bigr .
	\notag \\
	&\qquad \qquad \qquad \qquad \qquad \quad 
	\Bigl .
	\mbox{$\Psi(t)$ solves \eqref{Maxwell_equations:eqn:full_equations:dynamics} with $\widehat{\mathbf{J}}(\omega) = 0$}
	\Bigr \} 
\end{align}
with energy scalar product at frequency $\omega$, 
\begin{align}
	\bscpro{\widehat{\Phi}(\omega)}{\widehat{\Psi}(\omega)}_{\omega} &= 
	\Bscpro{\widehat{\Phi}(\omega) \, }{ \, \tfrac{\dd}{\dd \omega} \bigl ( \omega \, \widehat{W}(\omega) \bigr ) \, \widehat{\Psi}(\omega)}
	\notag \\
	&= \int_{\R^3} \dd x \, \widehat{\Phi}(\omega,x) \cdot \tfrac{\dd}{\dd \omega} \bigl ( \omega \, \widehat{W}(\omega,x) \bigr ) \, \widehat{\Psi}(\omega,x)
	, 
\end{align}
and associated norm $\bnorm{\widehat{\Psi}(\omega)}_{\omega} = \bscpro{\widehat{\Psi}(\omega)}{\widehat{\Psi}(\omega)}_{\omega}^{\nicefrac{1}{2}}$. If we insert operators $\widehat{A}(\omega)$ such as the derivative $- \ii \nabla$ in the second argument, they act only to the wave $\widehat{\Psi}(\omega)$ and not the energy density, 
\begin{align*}
	\bscpro{\widehat{\Phi}(\omega)}{\widehat{A}(\omega) \widehat{\Psi}(\omega)}_{\omega} &= \int_{\R^3} \dd x \, \widehat{\Phi}(\omega,x) \cdot \tfrac{\dd}{\dd \omega} \bigl ( \omega \, \widehat{W}(\omega,x) \bigr ) \, \bigl ( A(\omega) \, \widehat{\Psi} \bigr )(\omega,x) 
	. 
\end{align*}
This frequency-dependent scalar product has been widely used in the literature (for instance in \cite[equation~(3)]{Silverinha:Z2_topological_insulator_continuous_photonic_materials:2016} or in \cite[equation~(38)]{Raghu_Haldane:quantum_Hall_effect_photonic_crystals:2008}). Note that waves $\widehat{\Psi}(\omega)$ and $\widehat{\Phi}(\omega')$ from different frequencies $\omega \neq \omega'$ are automatically orthogonal to one another \emph{by definition} of the direct integral of Hilbert spaces. 

These Hilbert spaces $\Hil(\omega)$ really depend on the frequency not just through the scalar product, but may “shrink” or “grow” as vector spaces: for example, dispersive material weights may have singularities or produce spectral gaps (such as \cite[equations~(17) or (18)]{Silverinha:Z2_topological_insulator_continuous_photonic_materials:2016}), so it can happen that $\Hil(\pm \omega) = \{ 0 \}$ actually reduces to the trivial vector space. Note that the reality condition $W(t,x) = \overline{W(t,x)}$ imposed on the weights and fields translate to 
\begin{align}
	\bigl ( \widehat{\mathbf{E}}(-\omega) \, , \, \widehat{\mathbf{H}}(-\omega) \bigr ) = \bigl ( \overline{\widehat{\mathbf{E}}(\omega)} \, , \, \overline{\widehat{\mathbf{H}}(\omega)} \bigr )
	\label{Maxwell_equations:eqn:reality_condition_EM_field}
\end{align}
and manifests itself in a symmetry between the Hilbert spaces at $\pm \omega$: complex conjugation $(C \Psi)(x) = \overline{\Psi(x)}$ can be seen as an antiunitary between $\Hil(\omega)$ and $\Hil(-\omega)$ as the weights $\widehat{W}(-\omega) = \overline{\widehat{W}(+\omega)} = C \, \widehat{W}(+\omega) \, C$ are related by $C$. Therefore, poles and spectral gaps necessarily come in pairs of opposite frequency. 

Introducing $\Hil(\omega)$ and $\Hil_{\C}$ clarifies an earlier comment we made regarding the current density: since the medium can only carry certain states at a given frequency, namely those contained in $\Hil(\omega)$, only those can be excited by sources. Put another way, we need to impose 
\begin{assumption}[Current density only excites states supported by the medium]\label{Maxwell_equations:assumption:current_density}
	~\\
	Suppose the current density $\mathbf{J} \in L^1 \bigl ( \R , L^2(\R^3,\R^6) \bigr )$ is such that $\widehat{\mathbf{J}}(\omega) \in \widehat{W}(\omega) \, \Hil(\omega)$ holds for almost all frequencies. 
\end{assumption}
%

\subsection[The subspace of real electromagnetic fields and reduction to complex $\omega \geq 0$ waves]{The subspace of real electromagnetic fields and reduction to complex waves of non-negative frequencies} 
\label{Maxwell_equations:reduction}
The real-valuedness of $(\mathbf{E},\mathbf{H})$ manifests itself in the phase locking condition~\eqref{Maxwell_equations:eqn:reality_condition_EM_field} after Fourier transform, and it follows just like in our discussion around equation~\eqref{Maxwell_equations:eqn:phase_locking_condition} that the information contained in the negative frequency part is redundant. This singles out the \emph{real} subspace 
\begin{align}
	\Hil_{\R} = \Bigl \{ \Psi \in \Hil_{\C} \; \; \big \vert \; \; \widehat{\Psi}(- \omega) = \overline{\widehat{\Psi}(+\omega)} \mbox{ for almost all $\omega \in \R$} \Bigr \} 
	\subset \Hil_{\C} 
	\label{Maxwell_equations:eqn:definition_Hil_R}
\end{align}
of the \emph{complex} Hilbert space $\Hil_{\C}$.

However, while physical solutions $(\mathbf{E},\mathbf{H}) = 2 \Re \Psi$ are real, it is often more convenient to represent them as the real part of a complex wave $\Psi$. However, the complex wave is not uniquely determined by the real wave, and there are many choices we could conceivably make. Two have already been suggested in equation~\eqref{Maxwell_equations:eqn:real_states_determined_by_positive_frequency_part}, positive and negative frequency contributions $\Psi_{\pm}$ defined through~\eqref{Maxwell_equations:eqn:complex_positive_negative_frequency_waves}. For this pair we can see explicitly that the information contained in $\Psi_- = \overline{\Psi_+}$ and $\Psi_+$ is redundant, the symmetry constraint~\eqref{Maxwell_equations:eqn:real_states_determined_by_positive_frequency_part} shows how to reconstruct one from the other. 

Customarily, one picks \emph{complex waves composed of non-negative frequencies}, $\Psi = \Psi_+$. This has several advantages, all of which will be explored further in Section~\ref{Schroedinger}: first of all, restricting ourselves to complex waves in this way is convenient, because $\omega \geq 0$ is a \emph{spectral condition} that, mathematically speaking, can be easily and \emph{consistently} imposed. And a single set of Maxwell equations~\eqref{Maxwell_equations:eqn:approximate_Maxwell_equations} for $\omega \geq 0$ suffices to describe the waves and their dynamics. Moreover, this set of equations is naturally compatible with the \emph{Helmholtz decomposition} into transversal and longitudinal waves. 
Another advantage is that the transversal, bounded pseudo eigenfunctions\footnote{Simply put, a pseudo eigenfunction is an eigenfunction that lies outside of the vector space of the solutions one considers. Plane waves have infinite field energy, and therefore lie outside of the space of electromagnetic waves with finite field energy; nevertheless, they solve the eigenvalue equation for homogeneous media and arbitrary solutions may be expanded in terms of plane waves. } of this first-order equation are \emph{necessarily complex}; in the simplest case where the weights are independent of position (\ie the medium is homogeneous) the pseudo eigenfunctions are plane waves of positive frequency which gives rise to \emph{two} linearly independent \emph{real} solutions, 
\begin{align}
	A_{\Re} \cos (k \cdot x - \omega t) + A_{\Im} \sin (k \cdot x - \omega t) = \Re \Bigl ( \bigl ( A_{\Re} - \ii A_{\Im} \bigr ) \, \e^{+ \ii (k \cdot x - \omega t)} \Bigr )
	. 
	\label{Maxwell_equations:eqn:real_complex_plane_waves}
\end{align}
Here, real and imaginary part of the complex amplitudes, $A_{\Re}$ and $A_{\Im}$, become the prefactors of $\cos$ and $\sin$. The above arguments hold verbatim if plane waves are replaced by another suitable basis, \eg Bloch waves in case the weights $W_{\pm}$ describe a periodic electromagnetic medium. In fact, in view of \eqref{Maxwell_equations:eqn:symmetry_hatW} it is \emph{not possible} to find \emph{real} pseudo eigenfunctions for the first-order Maxwell equations: non-zero real waves are \emph{always} linear combinations of complex waves of positive and negative frequencies. 

Put mathematically, restricting to $\omega \geq 0$ means we consider the Hilbert space
\begin{align}
	\widehat{\Hil} = \Fourier^{\, -1} \, \Hil 
	= \int_{[0,\infty)}^{\oplus} \dd \omega \, \Hil(\omega) 
	\label{Maxwell_equations:eqn:dispersive_non_negative_frequency_Hilbert_space}
\end{align}
instead, and if we endow it with the scalar product 
\begin{align*}
	\scpro{\Phi}{\Psi} = \int_0^{\infty} \dd \omega \, \bscpro{\widehat{\Phi}(\omega)}{\widehat{\Psi}(\omega)}_{\omega} 
	, 
\end{align*}
then the inclusion of the factor of $\nicefrac{1}{2}$ in the definition of $\scpro{\, \cdot \,}{\, \cdot \,}_{\C}$ and the fact that the field energy is equally distributed in the $\omega > 0$ and $\omega < 0$ part yields $\mathcal{E}(\mathbf{E},\mathbf{H}) = \snorm{(\mathbf{E},\mathbf{H})}^2$. 

To summarize what we have done here: we have identified real electromagnetic fields with finite field energy $\bigl ( \mathbf{E}(t) , \mathbf{H}(t) \bigr )$ with a \emph{complex} wave $\Psi(t)$ composed of non-negative frequencies. In fact, $\Psi$ is an element of the \emph{complex} Hilbert space $\Hil$ that comes furnished with the energy scalar product and the energy norm. Working with complex (as opposed to real) Hilbert spaces becomes necessary if we want to even define the notion of field energy and adapt methods from quantum mechanics to classical electromagnetism. Any real electromagnetic field has a \emph{unique} representative in $\Hil$ while minimizing the amount of unphysical fields to complex gradient fields; the Hilbert space $\Hil_{\C}$, which we have introduced initially, can be seen as the complexification of $\Hil$ and contains unphysical transversal fields with non-zero imaginary part. 

\subsection{Neglecting dispersion: Maxwell equations for approximately monochromatic waves} 
\label{Maxwell_equations:heuristic_monochromatic_approximation}
In a great many situations dispersion can be safely neglected if the electromagnetic waves are composed of frequencies $\omega \approx \omega_0$, provided that 
\begin{align}
	\frac{\dd}{\dd \omega} \widehat{W}(\omega_0) \approx 0 
	\label{Maxwell_equations:eqn:neglecting_dispersion}
\end{align}
and therefore $\widehat{W}(\omega) \approx \widehat{W}(\omega_0)$ holds in some sense. Then the operator 
\begin{align*}
	\frac{\dd}{\dd \omega} \bigl ( \omega \, \widehat{W}(\omega) \bigr ) \approx \widehat{W}(\omega_0) \approx \widehat{W}(\omega)
\end{align*}
which enters the energy density~\eqref{Maxwell_equations:eqn:field_energy_density_frequency} can be Taylor expanded in $\omega - \omega_0$ and is approximately given by the material weights themselves. Hence, for states $\widehat{\Psi}(\omega)$ from within this narrow frequency region we have 
\begin{align}
	\scpro{\widehat{\Psi}(\omega)}{\widehat{\Psi}(\omega)}_{\omega} &= \Bscpro{\widehat{\Psi}(\omega)}{\tfrac{\dd}{\dd \omega} \bigl ( \omega \, \widehat{W}(\omega) \bigr ) \, \widehat{\Psi}(\omega)}
	\approx \scpro{\widehat{\Psi}(\omega)}{\widehat{W}(\omega_0) \, \widehat{\Psi}(\omega)}
	\label{Maxwell_equations:eqn:scalar_product_approximation}
	\\
	&\approx \Bscpro{\widehat{\Psi}(\omega)}{\tfrac{\dd}{\dd \omega} \bigl ( \omega \, \widehat{W}(\omega) \bigr ) \big \vert_{\omega = \omega_0} \, \widehat{\Psi}(\omega)}
	= \scpro{\widehat{\Psi}(\omega)}{\widehat{\Psi}(\omega)}_{\omega_0}
	, 
	\notag
\end{align}
and the left-hand side is finite if and only if the right-hand side is. 

To obtain a mathematical model we will use the same weights $\widehat{W}(\pm \omega_0)$ for \emph{all} frequencies \emph{of a given sign}, even those away from $\omega_0$ where the approximation~\eqref{Maxwell_equations:eqn:neglecting_dispersion} is no longer accurate. Of course, the range of validity is still limited to frequencies $\omega \approx \omega_0$ and only makes physically meaningful predictions for waves from within that narrow frequency range. This model is the one that is widely used in the literature, \emph{the only important difference here compared to most of the literature is that we take the difference in weights for positive and negative frequencies properly into account.} Put another way, in the absence of dispersion we need to distinguish between \emph{non-gyrotropic media} with \emph{real} weights, 
\begin{align*}
	W_+ = \widehat{W}(\omega_0) = \overline{\widehat{W}(\omega_0)} 
	= \widehat{W}(- \omega_0) = W_-
	, 
\end{align*}
and \emph{gyrotropic media} whose weights a non-vanishing imaginary part,  
\begin{align*}
	W_+ = \widehat{W}(\omega_0) \neq \overline{\widehat{W}(\omega_0)} 
	= \widehat{W}(- \omega_0) = W_-
	. 
\end{align*}
Therefore, for \emph{non}-gyrotropic media the Hilbert space $\Hil_{\C}$ has a \emph{single}, frequency-indepen\-dent scalar product. In fact, for lossless, positive index media we will show later on that $\Hil_{\C}$ defined in \eqref{Maxwell_equations:eqn:complixification_Hilbert_space} coincides with $L^2_{W_+}(\R^3,\C^6)$, the Hilbert space of complex fields with integrable amplitude and weighted scalar product $\bscpro{\Phi}{W_+ \, \Psi}$; we will properly introduce this Hilbert space in Section~\ref{Schroedinger:abstract:classical_waves}. The Hilbert space for \emph{gyrotropic} media splits in two, $\Hil_{\C} \subset \Hil_+ \oplus \Hil_-$, with two \emph{different} scalar products $\bscpro{\Phi}{W_{\pm} \, \Psi}$. 

Given that in general $W_+ \neq W_-$, after undoing the Fourier transform we obtain not one, but \emph{two} sets of Maxwell equations on \emph{two different} Hilbert spaces $\Hil_{\pm}$ — one for the \emph{positive frequency component} $\Psi_+$ and one for the \emph{negative frequency component} $\Psi_-$:\footnote{$\Psi_{\pm}^{(D,B)}$ which enters \eqref{Maxwell_equations:eqn:approximate_Maxwell_equations:constitutive_relations} are the positive and negative frequency contributions to $( \widehat{\mathbf{D}} , \widehat{\mathbf{B}})$ defined via \eqref{Maxwell_equations:eqn:complex_positive_negative_frequency_waves}.} 
\begin{subequations}\label{Maxwell_equations:eqn:approximate_Maxwell_equations}
	\begin{empheq}[left={\pm \omega \geq 0 \qquad \quad \empheqlbrace\,}]{align}
		W_{\pm} \, \ii \partial_t \Psi_{\pm}(t) &= \Rot \Psi_{\pm}(t) - \ii J_{\pm}(t) 
		, 
		\label{Maxwell_equations:eqn:approximate_Maxwell_equations:dynamics}
		\\
		\Div W_{\pm} \Psi_{\pm}(t) &= \rho_{\pm}(t) 
		, 
		\label{Maxwell_equations:eqn:approximate_Maxwell_equations:constraint}
		\\
		\Psi_{\pm}^{(D,B)}(t) &= W_{\pm} \, \Psi_{\pm}(t)
		, 
		\label{Maxwell_equations:eqn:approximate_Maxwell_equations:constitutive_relations}
		\\
		\Div J_{\pm}(t) + \partial_t \rho_{\pm}(t) &= 0 
		. 
		\label{Maxwell_equations:eqn:approximate_Maxwell_equations:charge_conservation}
	\end{empheq}
\end{subequations}
Note that we have yet to explain how to decompose the real, physical initial condition $\bigl ( \mathbf{E}(t_0) , \mathbf{H}(t_0) \bigr ) = (\mathbf{E}_0 , \mathbf{H}_0)$ into its non-negative and non-positive frequency components $\Psi_{\pm}(t_0)$ that enter as \emph{complex} initial conditions; we will address this later in Section~\ref{Schroedinger:auxiliary_operators:projections}. 

Clearly, for \emph{non}-gyrotropic media where $W_+ = W_-$, the two sets of equations coincide and $\Psi_{\pm}$ are elements of the same Hilbert space. In fact, we may choose to work with $(\mathbf{E},\mathbf{H})$ directly. 

Therefore, if $\widehat{W}(\pm \omega) \approx \widehat{W}(\pm \omega_0) = W_{\pm}$ describe a lossless, positive index medium, this heuristic argument shows that the field energy is proportional to $W_{\pm}$ and integrates up to 
\begin{align}
	\mathcal{E} \bigl ( \mathbf{E}(t) , \mathbf{H}(t) \bigr ) &= \frac{1}{2} \, \Bigl ( \bscpro{\Psi_+(t) \,}{\, W_+ \, \Psi_+(t)} + \bscpro{\Psi_-(t) \,}{\, W_- \, \Psi_-(t)} \Bigr ) 
	\notag \\
	&= \bscpro{\Psi_+(t) \,}{\, W_+ \, \Psi_+(t)} 
	, 
	\label{Maxwell_equations:eqn:field_energy}
\end{align}
because the field energy content of positive and negative frequency contribution is necessarily the same due to the symmetry $W_- = \overline{W_+}$ and $\Psi_- = \overline{\Psi_+}$; we will show that this expression is consistent with equation~\eqref{Maxwell_equations:eqn:integrated_field_energy_frequency} later in Section~\ref{Schroedinger:Hil_vs_Hil_phys}. The total field energy is also a conserved quantity: that can be checked by deriving equation~\eqref{Maxwell_equations:eqn:field_energy} with respect to time and plugging in \eqref{Maxwell_equations:eqn:approximate_Maxwell_equations:dynamics}. 

Fortunately, once we represent real electromagnetic fields as a complex wave composed solely of non-negative frequencies (as explained in Section~\ref{Maxwell_equations:heuristic_monochromatic_approximation}), we can treat both, non-gyrotropic and gyrotropic media in exactly the same manner. 
\medskip

\noindent
From the perspective of mathematics, an immediate interesting question is to \emph{quantify} the error we make when neglecting dispersion, \ie when replacing the frequency-dependent $\widehat{W}(\pm \omega)$ with $W_{\pm} = \widehat{W}(\pm \omega_0)$. Put succinctly, we would like to estimate the difference  $\bigl ( \mathbf{E}_{\mathrm{dis}}(t) , \mathbf{H}_{\mathrm{dis}}(t) \bigr ) - \bigl ( \mathbf{E}_{\mathrm{nd}}(t) , \mathbf{H}_{\mathrm{nd}}(t) \bigr )$ between the solution to the dispersive Maxwell equations and the approximate equations where dispersion has been neglected. This is more intricate than it seems at first glance: the initial condition for the full equation with dispersion is a \emph{past trajectory} $(-\infty,t_0] \ni t \mapsto \bigl ( \mathbf{E}(t) , \mathbf{H}(t) \bigr )$ — dispersive media have a “memory” — whereas the solutions to the approximate Maxwell equations~\eqref{Maxwell_equations:eqn:approximate_Maxwell_equations} depend only on the \emph{instantaneous} field configuration $\bigl ( \mathbf{E}(t_0) , \mathbf{H}(t_0) \bigr ) = (\mathbf{E}_0 , \mathbf{H}_0)$. So the matter is not as simple as evolving the same initial conditions with two equations, because there are \emph{many} past trajectories with 
\begin{align*}
	\bigl ( \mathbf{E}(t) , \mathbf{H}(t) \bigr ) \big \vert_{t = t_0} = (\mathbf{E}_0 , \mathbf{H}_0)
	, 
\end{align*}
and there is no unique or “obvious” choice of past trajectory on the basis of which to compare \eqref{Maxwell_equations:eqn:Maxwell_frequency_domain} with \eqref{Maxwell_equations:eqn:approximate_Maxwell_equations}. One physically sensible solution would be to place the same source in a medium described by \eqref{Maxwell_equations:eqn:Maxwell_frequency_domain} and the dispersion-free medium~\eqref{Maxwell_equations:eqn:approximate_Maxwell_equations}, and then estimate the difference of the two solutions. 
\section{Schrödinger formalism for linear, non-dispersive media} 
\label{Schroedinger}
Up until now, the formalism covers \emph{any} linear medium, Assumption~\ref{Maxwell_equations:assumption:5_minimal_assumptions} only lists \emph{necessary} assumptions which have to be imposed on physical grounds. Away from singularities of $\widehat{W}(\omega)$ we showed how to neglect dispersion and reduce Maxwell's equations~\eqref{Maxwell_equations:eqn:full_equations} to \eqref{Maxwell_equations:eqn:approximate_Maxwell_equations} governing complex waves of non-negative frequencies. Having gotten rid of redundant information does not just lead to an aesthetically more pleasing formalism, one of the benefits being that non-gyrotropic, gyrotropic and lossy materials can be described with the \emph{exact same set of equations}, but also avoids introducing unphysical transversal fields.

However, to obtain a Schrödinger formalism, we need to impose additional constraints, namely we will only consider electromagnetic waves $(\mathbf{E},\mathbf{H})$ with \emph{finite field energy} propagating in \emph{lossless positive index materials} (such as dielectrics): 
\begin{assumption}[Material weights]\label{Schroedinger:assumption:material_weights}
	The \emph{material weights} are a $6 \times 6$ matrix-valued function 
	\begin{align}
		W(x) = \left (
		\begin{matrix}
			\eps(x) & \chi(x) \\
			\chi(x)^* & \mu(x) \\
		\end{matrix}
		\right )
		\label{Schroedinger:eqn:material_weights}
	\end{align}
	that has the following properties: 
	\begin{enumerate}[(a)]
		\item The medium is \emph{lossless}, \ie $W(x) = W(x)^*$ takes values in the \emph{hermitian} matrices. 
		\item The medium is \emph{not a negative index material}, \ie there exist positive constants $C \geq c > 0$ so that $C \, \scpro{\Psi}{\Psi} \geq \scpro{\Psi}{W \Psi} \geq c \scpro{\Psi}{\Psi}$ holds. Put another way, the eigenvalues $w_1(x) , \ldots , w_6(x)$ of the hermitian matrix $W(x)$ are all positive and bounded away from $0$ and $\infty$ uniformly in $x$. 
	\end{enumerate}
\end{assumption}
\emph{To unburden the notation, we will systematically drop the index $+$ from all non-negative frequency objects: from now on, $W_+$ is replaced with $W$, $\Psi_+$ with $\Psi$, $J_+$ with $J$ and so forth.}

Thanks to the second assumption on the weights, the inverse of the weights $W^{-1}$ is bounded and takes values in the hermitian matrices again. Therefore, multiplying of both sides of \eqref{Maxwell_equations:eqn:approximate_Maxwell_equations:dynamics} with $\ii \, W^{-1}$ for $\omega \geq 0$ leads to 
\begin{align}
	\ii \, \partial_t \Psi(t) = M \, \Psi(t) - \ii \, J(t), 
	&&
	\Psi(0) = \Phi
	, 
	\label{Schroedinger:eqn:pm_Maxwell_Schroedinger_equation}
\end{align}
as an equivalent form of the dynamical law~\eqref{Maxwell_equations:eqn:approximate_Maxwell_equations:dynamics}. Here, the non-negative frequency part $\Psi$ of the electromagnetic field plays the role of the wave function and the non-negative frequency Maxwell operator 
\begin{align}
	M = W^{-1} \, \Rot \big \vert_{\omega \geq 0} 
	\label{Schroedinger:eqn:pm_Maxwell}
\end{align}
takes the place of the Hamiltonian. 

Seemingly, a simple algebraic operation is all it takes to go from \eqref{Maxwell_equations:eqn:approximate_Maxwell_equations:dynamics} to \eqref{Schroedinger:eqn:pm_Maxwell_Schroedinger_equation}, obviating the need for further discussion. However, specifying a Hilbert space is more subtle than in quantum mechanics for several reasons: we have to explain how to \emph{restrict \eqref{Maxwell_equations:eqn:approximate_Maxwell_equations} to non-negative frequencies $\omega \geq 0$}, make a judicious \emph{choice of scalar product} so that $M$ is selfadjoint (or hermitian in physics parlance) and verify that the \emph{divergence constraint~\eqref{Maxwell_equations:eqn:approximate_Maxwell_equations:constraint}} is satisfied. We emphasize that this is not a mere mathematical footnote, but its choice has a clear physical interpretation. Before we do, though, we take a step back and outline the Schrödinger formalism in the abstract.

\subsection{Abstract setting} 
\label{Schroedinger:abstract}
Not just electromagnetism but many other classical wave equations \cite{Dodin:geometric_view_noneikonal_waves:2014,Ruiz_Dodin:lagrangian_geometrical_optics_vector_waves_spin_particles:2015}, including transverse acoustic waves \cite{Pierce:wave_equation_fluids:1990,Alicki:Dirac_equation_MHD_waves:1992}, mechanical waves \cite{Suesstrunk_Huber:mechanical_topological_insulator:2015}, magneto plasmons \cite{Jin_Lu_et_al:topological_magnetoplasmon:2016} and linearized equations from plasma physics \cite{Ruiz:geometric_theory_waves_and_plasma_physics:2017}, admit a Schrödinger formalism, because they share certain characteristics: (1)~The equations are \emph{first-order in time}, (2)~the operators generating the dynamics have a \emph{product structure}, and (3)~in contrast to quantum mechanics the \emph{wave fields take values in $\R^n$} (instead of $\C^n$). Therefore, the ideas we develop here can be applied quite broadly to other wave equations. To fix a frame of reference, let us quickly recap the Schrödinger picture of quantum mechanics.

\subsubsection{Quantum mechanics} 
\label{Schroedinger:abstract:quantum_mechanics}
A quantum system is specified by a Hamilton operator, \eg 
\begin{align*}
	H = \frac{1}{2m} (- \ii \hbar \nabla)^2 + V(x) 
\end{align*}
describes a non-relativistic quantum particle on the continuum $\R^d$ subjected to a conservative force $\mathbf{F}(x) = - \nabla V(x)$. Typically, in the physics literature the Hilbert space $\Hil$ — \emph{whose elements represent states} — is not made explicit, since once a Hamiltonian is selected, there is a default choice with respect to which $H$ is selfadjoint (hermitian). For a non-relativistic quantum particle on $\R^d$ with internal degrees of freedom $\C^n$ the Hilbert space is constructed from the vector space of square integrable functions\footnote{The mathematically inclined reader will undoubtedly have noticed that we have used the same symbol for the vector space of square integrable functions and that composed of \emph{equivalence classes} of square integrable functions differing only on a set of measure $0$ (following \cite[Chapter~2]{Lieb_Loss:analysis:2001}).  Distinguishing the two is a  necessary technical, but for our purposes inconsequential detail of the proper definition of weighted $L^p$-spaces. }, 
\begin{align}
	L^2(\R^d,\C^n) = \Bigl \{ \Psi : \R^d \longrightarrow \C^n \; \; \big \vert \; \; \int_{\R^d} \dd x \, \sabs{\Psi(x)}^2 < \infty \Bigr \} 
	, 
	\label{Schroedinger:eqn:unweighted_L2}
\end{align}
that is endowed with the scalar product 
\begin{align}
	\scpro{\Phi}{\Psi} = \int_{\R^d} \dd x \, \Phi(x) \cdot \Psi(x)
	\label{Schroedinger:eqn:unweighted_scalar_product}
\end{align}
and the induced norm $\norm{\Psi} = \scpro{\Psi}{\Psi}^{\nicefrac{1}{2}}$. Complex conjugation is implicitly contained in the dot product $\Phi(x) \cdot \Psi(x) = \sum_{j = 1}^n \overline{\Phi_j(x)} \; \Psi_j(x)$. Pure states, rank-$1$ projections $P = \sopro{\Psi}{\Psi}$, can be represented as wave functions $\Psi$ normalized to $1 = \norm{\Psi}^2$; the normalization is necessary to link the wave function to probabilities of experimental outcomes. 

Note that (complex) quantum wave functions themselves are \emph{not} physically observable as wave functions are only defined up to a phase (pure states can be identified with \emph{rays} in the Hilbert space). Indeed, only the associated probabilities and expectation values are accessible. This is in stark contrast to (real-valued) classical waves that are directly observable. 

The Hamiltonian $H = H^*$ is selfadjoint (hermitian) with respect to the scalar product of choice, meaning that in addition to technical questions regarding domains we have 
\begin{align*}
	\sscpro{\Phi}{H \Psi} = \sscpro{H \, \Phi}{\Psi}
	. 
\end{align*}
Hamiltonians such as the one above are typically selfadjoint with respect to the “default” scalar product~\eqref{Schroedinger:eqn:unweighted_scalar_product}, which is why the choice is often not made explicit. The selfadjointness then implies that the evolution group $\e^{- \ii \frac{t}{\hbar} H}$ associated to the Schrödinger equation
\begin{align}
	\ii \hbar \partial_t \Psi(t) = H \Psi(t) 
	, 
	&&
	\Psi(0) = \Phi
	, 
	\label{Schroedinger:eqn:quantum_Schroedinger_equation}
\end{align}
exists, is \emph{unitary} and solves \eqref{Schroedinger:eqn:quantum_Schroedinger_equation} in the form $\Psi(t) = \e^{- \ii \frac{t}{\hbar} H} \Phi$. Unitarity has a second consequence, namely the \emph{existence of a conserved quantity} — probabilities in this case. 

Put succinctly, a quantum system is fixed by specifying a Hilbert space $\Hil$ whose elements represent states and a selfadjoint operator $H$, the Hamiltonian, that enters the dynamical law, the Schrödinger equation. 

\subsubsection{Classical waves} 
\label{Schroedinger:abstract:classical_waves}
The reason we have spent a few paragraphs recalling basic facts on quantum mechanics is that we can now concisely list the \emph{three} ingredients of the Schrödinger formalism of classical waves: 
\begin{enumerate}[(S1)]
	\item \textbf{States}, real vector fields, are represented as elements of a \emph{complex} Hilbert space $\Hil$ which itself is a subspace of the Hilbert space 
	\begin{align}
		L^2_W(\R^d,\C^n) &= \Bigl \{ \Psi : \R^d \longrightarrow \C^n \; \; \big \vert \; \; \int_{\R^d} \dd x \, \Psi(x) \cdot W(x) \, \Psi(x) < \infty \Bigr \} 
		\label{Schroedinger:eqn:weighted_L2_space}
		\\
		&= \Bigl \{ \Psi : \R^d \longrightarrow \C^n \; \; \big \vert \; \; \int_{\R^d} \dd x \, \sabs{\Psi(x)}^2 < \infty \Bigr \} 
		\notag 
	\end{align}
	consisting of waves with square integrable amplitudes of \emph{non-negative frequencies} and is endowed with the \emph{weighted} scalar product, 
	\begin{align}
		\sscpro{\Phi}{\Psi}_W = \bscpro{\Phi}{W \, \Psi} 
		= \int_{\R^d} \dd x \, \Phi(x) \cdot W(x) \, \Psi(x) 
		, 
		\label{Schroedinger:eqn:weighted_scalar_product}
	\end{align}
	where the operator $W$ describes the properties of the medium and has to possess properties akin to those enumerated in Assumption~\ref{Schroedinger:assumption:material_weights}: 
	\begin{enumerate}[(a)]
		\item $W$ is selfadjoint (hermitian) with respect to the standard scalar product~\eqref{Schroedinger:eqn:unweighted_scalar_product}. 
		\item $W$ is bounded away from $0$ and $\infty$, \ie there exist $C > c > 0$ so that $c \, \scpro{\Psi}{\Psi} \leq \scpro{\Psi}{W \, \Psi} \leq C \, \scpro{\Psi}{\Psi}$ holds true. 
	\end{enumerate}
	These conditions arise naturally in the context of the Schrödinger formalism. First, the boundedness assumption of $W$ is necessary to guarantee that fields $\Psi$ with square integrable amplitudes, $\scpro{\Psi}{\Psi} < \infty$, lie in $\Hil$. The remaining properties ensure that $\scpro{\Phi}{\Psi}_W = \scpro{\Phi}{W \, \Psi}$ satisfies the axioms of scalar products: requiring 
	\begin{align*}
		\scpro{\Phi}{\Psi}_W &= \scpro{\Phi}{W \, \Psi} \overset{!}{=} \overline{\scpro{\Psi}{W \, \Phi}} = \overline{\scpro{\Psi}{\Phi}_W}
	\end{align*}
	means we need to impose $W = W^*$, \ie part of condition~(a). Another defining property of scalar products is non-negativity, $\scpro{\Psi}{\Psi}_W \geq 0$, and that $\scpro{\Psi}{\Psi}_W = 0$ should imply $\Psi = 0$; therefore, $W$ has to be strictly positive, $\scpro{\Psi}{W \, \Psi} \geq c \, \scpro{\Psi}{\Psi}$. 
	\item These weights also enter the \textbf{Maxwell-type operator} 
	\begin{align}
		M = W^{-1} \, D
		\label{Schroedinger:eqn:Maxwell_type_operator}
	\end{align}
	that is the product of the weights $W$ and a second “free” operator $D$, which is also selfadjoint (hermitian) with respect to the \emph{unweighted} scalar product~\eqref{Schroedinger:eqn:unweighted_scalar_product}. This operator (endowed with the domain of $D$) takes the place of the Hamiltonian in the \textbf{Schrödinger equation} (potentially with a current $J(t)$ as source term), 
	\begin{align}
		\ii \, \partial_t \Psi(t) = M \Psi(t) - \ii \, J(t) 
		, 
		&&
		\Psi(0) = \Phi
		. 
		\label{Schroedinger:eqn:prototypical_Schroedinger_equation_wave}
	\end{align}
	Due to the product structure $M$ is selfadjoint with respect to the \emph{weighted} scalar product $\scpro{\, \cdot \,}{\, \cdot \,}_W$ defined in equation~\eqref{Schroedinger:eqn:weighted_scalar_product} above, something that can be seen from a straightforward computation. Again, dispensing with technical questions of domains, $M^{\ast_W} = M$ follows from 
	\begin{align}
		\bscpro{\Phi}{M \Psi}_W &= \bscpro{\Phi}{W \, W^{-1} \, \Rot \Psi}
		= \bscpro{\Rot \Phi}{\Psi}
		\notag \\
		&= \bscpro{W^{-1} \, \Rot \Phi}{W \Psi}
		= \bscpro{M \Phi}{\Psi}_W
		. 
		\label{Schroedinger:eqn:symmetry_Maxwell_operator}
	\end{align}
	Therefore its evolution group $\e^{- \ii t M}$ exists and is unitary with respect to $\scpro{\, \cdot \,}{\, \cdot \,}_W$. In the absence of sources, $\scpro{\Psi(t)}{\Psi(t)}_W = \scpro{\Psi(0)}{\Psi(0)}_W$ and similar quantities are \emph{conserved}; in case of Maxwell's equations the former quantity is the total field energy. 
	\item The biggest conceptual difference to quantum mechanics is that \textbf{physical fields $u$ are real}, and the equation $u(t) = 2 \Re \Psi(t)$ connects the solution $u(t)$ of the classical wave equation in question to the complex wave $\Psi(t)$ that solves \eqref{Schroedinger:eqn:prototypical_Schroedinger_equation_wave}. 
	To avoid including unphysical field configurations, we had to restrict ourselves to complex fields of non-negative frequencies. 
\end{enumerate}
%

\subsection{Imposing the frequency constraints} 
\label{Schroedinger:auxiliary_operators}
When the material weights are complex, then the two equations equations~\eqref{Maxwell_equations:eqn:approximate_Maxwell_equations} for $\pm \omega \geq 0$ are distinct and one natural question emerges: how do I split a real wave into its positive and negative frequency part? Once we are given a \emph{solution} $\bigl ( \mathbf{E}(t) , \mathbf{H}(t) \bigr )$ that is easy, we just need to Fourier transform the solution. However, usually we are \emph{given} an initial condition $\bigl ( \mathbf{E}(t_0) , \mathbf{H}(t_0) \bigr )$ and are tasked to \emph{find} a solution. Consequently, we need to break down this initial condition prior to being able to determine the associated solution. 

So as a first step in identifying (S1)–(S3) for electromagnetism in matter, we start by studying the operator \eqref{Schroedinger:eqn:pm_Maxwell} \emph{without} imposing frequency constraints and with the help of that auxiliary operator relate the real, physical fields to the complex fields representing them. The conditions on $W$ spelled out in Assumption~\ref{Schroedinger:assumption:material_weights} allow us define the complex Hilbert space $L^2_W(\R^3,\C^6)$ as in \eqref{Schroedinger:eqn:weighted_L2_space} associated with the non-negative frequency weights $W$, endowed with the \emph{weighted} scalar product~\eqref{Schroedinger:eqn:weighted_scalar_product}. Later on, we will link the square of the induced norm $\norm{\Psi}_W^2 = \scpro{\Psi}{\Psi}_W$ to the total \emph{field energy}, so this object has a neat physical interpretation. Hence, the second equality in \eqref{Schroedinger:eqn:weighted_L2_space} tells us that under the conditions imposed on the material weights $W$, electromagnetic waves have finite field energy if and only if their amplitudes are square integrable.

\subsubsection{The Helmholtz decomposition into transversal and longitudinal waves} 
\label{Schroedinger:auxiliary_operators:helmholtz}
Another useful and physically significant concept is that of a \emph{Helmholtz decomposition adapted to the medium}; as we shall show at the end of this section, the dynamics of the transversal and longitudinal components of electromagnetic fields is very different. Concretely, we would like a Helmholtz decomposition associated to the either one of the \emph{auxiliary Maxwell operators}
\begin{align}
	\Maux_{\pm} = W_{\pm}^{-1} \, \Rot
	, 
	\label{Schroedinger:eqn:auxiliary_Maxwell_operators}
\end{align}
one for the positive frequency weights $W_+ = W$ and the other for the negative frequency weights $W_- = \overline{W}$. Their actions on complex fields coincides with \eqref{Schroedinger:eqn:pm_Maxwell} but lacks the restriction to non-negative or non-positive frequencies. Put another way, $L^2_{W_{\pm}}(\R^3,\C^6)$ contains unphysical states. Note that the phase locking condition~\eqref{Maxwell_equations:eqn:phase_locking_condition} manifests itself as 
\begin{align*}
	C \, \Maux_{\pm} \, C = - \Maux_{\mp}
	, 
\end{align*}
and we are free to choose only one of these auxiliary operators. We shall pick $\Maux = \Maux_+$ and drop the index $+$ consistently to unburden the notation. 

Originally, the standard Helmholtz decomposition (\cf \eg \cite[Appendix~A.4]{DeNittis_Lein:adiabatic_periodic_Maxwell_PsiDO:2013} and references therein) refers to the splitting of a vector field $B$ on $\R^3$ into a rotational field and a gradient field, 
\begin{align*}
	B = B_{\perp} + B_{\parallel} = \nabla \times A + \nabla \varphi 
	. 
\end{align*}
An alternative characterization of transversal and longitudinal components is $\nabla \cdot B_{\perp} = 0$ and $\nabla \times B_{\parallel} = 0$, and it is not coincidental that the latter appears as a constraint for the \emph{in vacuo} Maxwell equations in the absence of charges. 

Also in matter does this decomposition have an elegant geometric interpretation if written in terms of Hilbert spaces, namely 
\begin{align}
	L^2_W(\R^3,\C^6) = \mathcal{J} \oplus \mathcal{G}
\end{align}
is split into the orthogonal sum of \emph{longitudinal gradient fields}
\begin{align*}
	\mathcal{G} &= \Bigl \{ \Psi = \bigl ( \nabla \varphi^E , \nabla \varphi^H \bigr ) \in L^2_W(\R^3,\C^6) \; \; \big \vert \; \; \bigl ( \varphi^E , \varphi^H \bigr ) \in L^2_{\mathrm{loc}}(\R^3,\C^2) \Bigr \} 
	\\
	&= \ran \, \bigl ( \nabla , \nabla \bigr ) 
	= \ker \, \Rot
	= \ker \, \Maux
\end{align*}
and \emph{transversal fields} 
\begin{align*}
	\mathcal{J} &= \mathcal{G}^{\perp_W}
	= \ker \, \bigl ( \Div \, W \bigr )
	= \ran \Maux
	,
\end{align*}
which \emph{by definition} are those that are $\sscpro{\, \cdot \,}{\, \cdot \,}_W$-orthogonal to gradient fields. The notation $\oplus$ here means that any $\Psi = \Psi_{\perp} + \Psi_{\parallel} \in L^2_W(\R^3,\C^6)$ can be split \emph{uniquely} into the sum of two mutually $\scpro{\, \cdot \,}{\, \cdot \,}_W$-\emph{orthogonal} waves $\Psi_{\perp} \in \mathcal{J}$ and $\Psi_{\parallel} \in \mathcal{G}$. 

The fact that elements of $\mathcal{J}$ satisfy the constraint condition~\eqref{Maxwell_equations:eqn:approximate_Maxwell_equations:constraint} for $\rho = 0$ can be easily seen by taking the scalar product of $\Phi = \bigl ( \nabla \varphi^E , \nabla \varphi^H \bigr )$ with $\Psi \in \mathcal{J}$ and performing partial integration. 

The utility of the Helmholtz decomposition becomes apparent when solving the dynamical equation~\eqref{Maxwell_equations:eqn:approximate_Maxwell_equations:dynamics} under the constraint~\eqref{Maxwell_equations:eqn:approximate_Maxwell_equations:charge_conservation}: writing the field $\Psi = \Psi_{\perp} + \Psi_{\parallel}$ and the current $J(t) = J_{\perp}(t) + J_{\parallel}(t)$ as the sum of transversal and longitudinal contributions, then $\nabla \times \nabla \varphi = 0$ implies the longitudinal gradient field is obtained by simply integrating up the current, 
\begin{align}
	\Psi_{\parallel}(t) &= \Psi_{\parallel}(0) + \ii \int_0^t \dd s \, J_{\parallel}(s) 
	. 
	\label{Schroedinger:eqn:solution_longitudinal_part}
\end{align}
Thanks to charge conservation~\eqref{Maxwell_equations:eqn:approximate_Maxwell_equations:charge_conservation} and $\Psi_{\perp}(t) \in \ker \bigl ( \Div \, W \bigr )$, the solution automatically satisfies the constraint~\eqref{Maxwell_equations:eqn:approximate_Maxwell_equations:constraint}. 

While the Helmholtz decomposition of $\Psi = \Psi_{\perp} + \Psi_{\parallel}$ is uniquely determined by $W$ alone, \emph{a priori} it is not clear how to \emph{compute} $\Psi_{\perp}$ and $\Psi_{\parallel}$ from $\Psi$. Secondly, we \emph{actually} wish to obtain these two complex waves from the real, physical fields $(\mathbf{E},\mathbf{H})$ in such a way that $\Psi$ is composed only of waves of \emph{non-negative} frequencies. Excluding negative frequencies is necessary, because otherwise we would evolve the negative frequency component $\Psi_-$ with the wrong dynamical equation in case positive and negative frequency weights differ, $W_+ = W \neq \overline{W} = W_-$; that is the whole reason why \eqref{Maxwell_equations:eqn:approximate_Maxwell_equations:dynamics} comes as \emph{two} sets of equations, one for $\omega \geq 0$ and another for $\omega \leq 0$. 

\subsubsection{Systematically representing real fields as complex waves} 
\label{Schroedinger:auxiliary_operators:projections}
The key idea is to define $\mathcal{G} = \ker \Maux$ and the positive frequency subspace $\Jphys \subset \mathcal{J}$ as \emph{spectral subspaces} associated to the auxiliary Maxwell operator $\Maux$: endowed with the domain of the free Maxwell operator $\Rot$, we know from standard arguments that $\Maux$ is selfadjoint (\cf \cite{Birman_Solomyak:L2_theory_Maxwell_operator:1987} and Appendix~\ref{appendix:generalized_Maxwell_type_operators}) and therefore admits a functional calculus that ascribes meaning to the expressions 
\begin{subequations}\label{Schroedinger:eqn:spectral_projections_positive_and_0}
	\begin{align}
		P_+ &= 1_{(0,\infty)}(\Maux) 
		, 
		\\
		P_0 &= 1_{\{ 0 \}}(\Maux)
		, 
	\end{align}
\end{subequations}
for the spectral projections onto $(0,\infty)$ and the eigenspace for the eigenvalue $0$ that are defined in terms of the characteristic functions 
\begin{align*}
	1_{(0,\infty)}(\omega) &= 
	\begin{cases}
		1 & \omega > 0 \\
		0 & \mbox{else} \\
	\end{cases}
	, 
	\\
	1_{\{ 0 \}}(\omega) &= 
	\begin{cases}
		1 & \omega = 0 \\
		0 & \mbox{else} \\
	\end{cases}
	,
\end{align*}
onto the sets of the positive numbers and $0$, respectively. These two operators are orthogonal projections associated to different parts of the spectrum of $\Maux$, namely $P_+$ maps onto the positive frequencies and $P_0$ onto the eigenspace $\mathcal{G} = \ker \Maux = \ran P_0$ to the eigenvalue $0$. Fields in $\ran P_+ = \Jphys$ are automatically transversal, because $P_+ \, P_0 = 0$ implies any $\Psi_{\perp} \in \Jphys$ is $\scpro{\, \cdot \,}{\, \cdot \,}_W$-orthogonal to all gradient fields and therefore $\Jphys \subset \mathcal{J} = \mathcal{G}^{\perp_W}$. While the definition of $P_+$ may seem very abstract, it can in many situations be computed quite explicitly. For waves from a finite frequency range propagating in periodic media, $P_+ (\mathbf{E},\mathbf{H})$ can be expressed in terms of finitely many Bloch waves. 

Then the map\footnote{
Without the factor of $\nicefrac{1}{2}$ this operator would sometimes be called the Hardy projection \cite{Noether:class_singular_integral_equations:1921,Coifman_Weiss:extension_Hardy_spaces:1977}; it enters in the modern extension of the Noether, Gohberg and Krein \cite{Gohberg_Krein:fundamental_aspects_defect_numbers_root_numbers:1957} and has applications in the analysis of topological insulators and index theorems \cite{Schulz_Baldes:Z2_indices_factorization_properties_odd_symmetric_Fredholm_operators:2015,Prodan_Schulz_Baldes:complex_topological_insulators:2016}. } which extracts the non-negative frequency contribution from $(\mathbf{E},\mathbf{H})$ is
\begin{align*}
	Q = P_+ + \tfrac{1}{2} \, P_0 
	. 
\end{align*}
Real fields $(\mathbf{E},\mathbf{H}) \in \Hil_{\mathrm{phys}} = L^2(\R^3,\R^6)$ are represented as the complex wave $\Psi = Q (\mathbf{E},\mathbf{H})$ on the \emph{complex} Banach space $\Hil = \mathcal{J}_+ \oplus \mathcal{G} = Q \bigl [ L_W^2(\R^3,\C^6) \bigr ]$; once we endow $\Hil$ with the weighted energy scalar product $\scpro{\, \cdot \,}{\, \cdot \,}_W$, it becomes a Hilbert space.

\paragraph{Non-gyrotropic media} 
\label{par:non_gyrotropic_media}
It is these two orthogonal projections and the map $Q$ which will implement the frequency restrictions implied in equations~\eqref{Maxwell_equations:eqn:approximate_Maxwell_equations} and allow us to compute $\Psi_{\perp} = P_+ (\mathbf{E},\mathbf{H})$ and $\Psi_{\parallel} = \tfrac{1}{2} \, P_0 (\mathbf{E},\mathbf{H})$. Let us first start with the non-gryotropic case where $W = W_+ = \overline{W} = W_-$ are real. The motivation to include the factor of $\nicefrac{1}{2}$ in the definition of $Q$ is to avoid counting longitudinal gradient fields “twice”: we can use functional calculus and the symmetry relation $C \, \Maux \, C = - \Maux$ to deduce 
\begin{align}
	2 \Re \, Q \, \Re &= \frac{1}{2} \Bigl ( 1_{(0,\infty)}(\Maux) + C \, 1_{(0,\infty)}(\Maux) \, C + 1_{(0,\infty)}(\Maux) \, C + C \, 1_{(0,\infty)}(\Maux)
	\Bigr . \notag \\
	&\qquad \; \; \; \Bigl .
	+ \tfrac{1}{2} \, 1_{\{ 0 \}}(\Maux) + \tfrac{1}{2} \, C \, 1_{\{ 0 \}}(\Maux) \, C + \tfrac{1}{2} \, C \, 1_{\{ 0 \}}(\Maux) + \tfrac{1}{2} \, 1_{\{ 0 \}}(\Maux) \, C \Bigr )
	\notag \\
	&= \Bigl ( 1_{(0,\infty)}(\Maux) + 1_{(-\infty,0)}(\Maux) + 1_{\{ 0 \}}(\Maux) \Bigr ) \, \Re
	= \Re
	\notag 
\end{align}
as maps on $L^2_W(\R^3,\C^6)$. Therefore, $Q$ maps real fields onto their complex representatives while given a complex field, taking $2 \Re$ restores the real field.
\begin{lemma}\label{Schroedinger:lem:equivalence_real_complex_states}
	Suppose $W = \overline{W}$ is real and satisfies Assumption~\ref{Schroedinger:assumption:material_weights}. Then any real field $(\mathbf{E},\mathbf{H}) \in L^2(\R^3,\R^6)$ can be uniquely represented by $\Psi = Q (\mathbf{E},\mathbf{H}) \in \Hil$ and the real field is recovered by taking $2 \Re$, 
	\begin{align}
		2 \Re \, Q (\mathbf{E},\mathbf{H}) = (\mathbf{E},\mathbf{H}) \in L^2(\R^3,\R^6)
		. 
		\label{Schroedinger:eqn:2Re_left_inverse_Q}
	\end{align}
	Put mathematically, $Q : L^2(\R^3,\R^6)$ is one-to-one and $2 \Re \, Q = \id : L^2(\R^3,\R^6) \longrightarrow L^2(\R^3,\R^6)$ reduces to the identity map. 
\end{lemma}
The fact that $Q$ is one-to-one is the content of Lemma~\ref{appendix:1_to_1_correspondence:lem:1_to_1_correspondence}~(1) while the computation above shows the second part. 

\paragraph{Generalization to gyrotropic media} 
\label{par:gyrotropic_media}
The simple arguments used in the computation confirming that $2 \Re$ is the inverse of $Q$ fail if $W_+ \neq W_- = \overline{W_+}$ has a non-vanishing imaginary part, and it is not clear whether also here $2 \Re \, Q \, \Re = \Re$ holds true. For such gyrotropic media, two distinct auxiliary Maxwell operators $\Maux_{\pm} = W_{\pm}^{-1} \, \Rot$ enter the above computation and it no longer clear whether $P_+$ is orthogonal to 
\begin{align*}
	P_- = 1_{(-\infty,0)} \bigl ( \Maux_- \bigr )
	, 
\end{align*}
meaning whether $P_+ \, P_- = 0 = P_- \, P_+$ still holds, since $P_+$ and $P_-$ are spectral projections associated to two \emph{different} selfadjoint operators. Similarly, functional calculus is of no help when trying to show that 
\begin{align*}
	1_{\{ 0 \}} \bigl ( \Maux_- \bigr ) \, 1_{(0,\infty)} \bigl ( \Maux_+ \bigr )
	= 0 
	= 1_{\{ 0 \}} \bigl ( \Maux_+ \bigr ) \, 1_{(0,\infty)} \bigl ( \Maux_- \bigr ) 
\end{align*}
holds. If that were true, though, then the “projection defect” 
\begin{align*}
	\bigl ( 2 \Re \, Q \, \Re \bigr )^2 - 2 \Re \, Q \, \Re &=  \bigl ( P_+ \, P_- + P_- \, P_+ \bigr ) \, \Re 
	+ \\
	&\quad 
	- \frac{1}{4} \Bigl ( 1_{\{ 0 \}} \bigl ( \Maux_- \bigr ) \, P_+ + 1_{\{ 0 \}} \bigl ( \Maux_+ \bigr ) \, P_- \Bigr ) \, \Re 
\end{align*}
would vanish, and combined with the injectivity of $2 \Re \, Q : L^2(\R^3,\R^6) \longrightarrow L^2(\R^3,\R^6)$ we deduce that $2 \Re \, Q \Re$ is an injective projection, consequently the identity, and $2 \Re \, Q \, \Re = \Re$ follows. 

Nevertheless, we can show that $Q$ and $2 \Re \, Q$ restricted to real fields are one-to-one (Lemma~\ref{appendix:1_to_1_correspondence:lem:1_to_1_correspondence}~(1) and (2)), and therefore provide us with a unique association between real fields and complex fields composed solely of positive frequencies. This generalizes Lemma~\ref{Schroedinger:lem:equivalence_real_complex_states} to gyrotropic media. 
\begin{proposition}[Identification of real and complex fields]\label{Schroedinger:prop:identification_real_complex_vector_spaces}
	Suppose the material weights $W$ satisfy Assumption~\ref{Schroedinger:assumption:material_weights}. 
	\begin{enumerate}[(1)]
		\item Any physical (\ie real-valued) electromagnetic field $(\mathbf{E},\mathbf{H}) \in L^2(\R^3,\R^6)$ with square integrable amplitudes can be uniquely represented by a complex wave $\Psi = Q (\mathbf{E},\mathbf{H}) \in \Hil$ with $Q = P_+ + \tfrac{1}{2} P_0$. 
		\item Any complex $\Psi = Q (\mathbf{E}',\mathbf{H}') \in \Hil$ field constructed from a real field $(\mathbf{E}',\mathbf{H}') \in L^2(\R^3,\R^6)$ uniquely represents $(\mathbf{E},\mathbf{H}) = 2 \Re \Psi \in L^2(\R^3,\R^6)$. 
	\end{enumerate}
\end{proposition}
The mathematically inclined reader may find the proof in Appendix~\ref{appendix:1_to_1_correspondence}, which heavily relies on \cite[Lemma~2.5]{DeNittis_Lein:ray_optics_photonic_crystals:2014}. 

While the above statement still yields a one-to-one correspondence between real and complex waves, the above statement is weaker than Lemma~\ref{Schroedinger:lem:equivalence_real_complex_states}. That is because we do not know whether $(\mathbf{E},\mathbf{H}) = (\mathbf{E}',\mathbf{H}')$ in item~(2) agree, \ie whether $2 \Re \, Q \, \Re = \Re$ holds in case the material weights are complex. We attempted unsuccessfully to modify $Q = P_+ + \tfrac{1}{2} P_0$ the maps $2 \Re$ is its left-inverse while imposing a number of conditions: 
\begin{enumerate}[(1)]
	\item $Q$ maps real fields onto complex fields composed solely of non-negative frequencies, $Q \bigl [ L^2(\R^3,\R^6) \bigr ] = 1_{[0,\infty)}(\Maux) \bigl [ L^2_W(\R^3,\C^6) \bigr ]$. 
	\item $Q : L^2(\R^3,\R^6) \longrightarrow \Hil$ is one-to-one.
	\item $2 \Re \, Q : L^2(\R^3,\R^6) \longrightarrow L^2(\R^3,\R^6)$ is one-to-one.
	\item When the weights are real, $Q$ reduces to $Q = P_+ + \tfrac{1}{2} P_0$. 
\end{enumerate}
The first condition ensures that the relevant Maxwell operator is still $M$ while the second and the third are the key ingredients for showing $2 \Re \, Q \, \Re = \Re$. As explained above, it may be easier to show $\bigl ( 2 \Re \, Q \, \Re \bigr )^2 = 2 \Re \, Q \, \Re$ instead. The last one is a consistency condition for the non-gyrotropic case where we already know the correct association. 

\subsection{Schrödinger formalism for classical electromagnetism} 
\label{Schroedinger:formalism}
Now we are in a position to (S1) identify the Hilbert space, (S2)~show that the Schrödinger equation along with the Maxwell operator are equivalent to Maxwell's equation~\eqref{Maxwell_equations:eqn:approximate_Maxwell_equations}, and (S3)~take the real-valuedness of physical fields properly into account.

\subsubsection{States} 
\label{Schroedinger:formalism:states}
The physical fields $(\mathbf{E},\mathbf{H})$ and the complex fields $\Psi = Q \, (\mathbf{E},\mathbf{H})$ that represent them are connected through the projection $Q = 1_{(0,\infty)}(\Maux) + \tfrac{1}{2} \, 1_{\{ 0 \}}(\Maux)$, and therefore the Hilbert spaces is the subspace of non-negative ($\omega \geq 0$) frequency states,
\begin{align*}
	\Hil = Q \, \bigl ( L^2_W(\R^3,\C^6) \bigr ) 
	= P_+ \, \bigl ( L^2_W(\R^3,\C^6) \bigr ) \oplus P_0 \, \bigl ( L^2_W(\R^3,\C^6) \bigr )
	= \Jphys \oplus \mathcal{G}
	, 
\end{align*}
which inherits the scalar product $\scpro{\, \cdot \,}{\, \cdot \,}_W$ from $L^2_W(\R^3,\C^6)$; any vector $\Psi = \Psi_{\perp} + \Psi_{\parallel}$ can again be split into transversal and longitudinal components, 
\begin{subequations}\label{Schroedinger:eqn:complex_wave_computed_from_real_wave}
	\begin{align}
		\Psi &= Q \, (\mathbf{E},\mathbf{H})
		, 
		\\
		\Psi_{\perp} &= P_+ \, (\mathbf{E},\mathbf{H})
		, 
		\\
		\Psi_{\parallel} &= P_0 \, (\mathbf{E},\mathbf{H})
		. 
	\end{align}
\end{subequations}
Note that applying $Q$ to the physical fields makes sense, because real fields $(\mathbf{E},\mathbf{H})$  with square integrable amplitudes are elements of $\Hil_{\mathrm{phys}} = L^2(\R^3,\R^6)$ (defined in analogy to \eqref{Schroedinger:eqn:unweighted_L2}), \ie those can also be regarded as elements of $L^2_W(\R^3,\C^6)$. 

\subsubsection{The Maxwell operator and the Schrödinger equation} 
\label{Schroedinger:formalism:dynamics}
The relevant Maxwell operator 
\begin{align*}
	M = 1_{[0,\infty)}(\Maux) \; \Maux \; 1_{[0,\infty)}(\Maux) \, \big \vert_{\Hil} 
	= \Maux \big \vert_{\Hil} = W^{-1} \, \Rot \big \vert_{\Jphys} \oplus 0 \vert_{\mathcal{G}}
\end{align*}
is just the restriction of the auxiliary Maxwell operator to $\omega \geq 0$; for a complete mathematical definition, we refer to Appendix~\ref{appendix:generalized_Maxwell_type_operators}. This operator inherits the selfadjointness of $\Maux$ and enter the Schrödinger equation
\begin{align}
	\ii \, \partial_t \Psi(t) = M \Psi(t) - \ii \, J(t)
	, 
	&&
	\Psi(0) = Q \, (\mathbf{E}_0,\mathbf{H}_0)
	, 
	\label{Schroedinger:eqn:pm_Schroedinger_equation_em}
\end{align}
in place of the Hamilton operator; the complex currents 
\begin{align*}
	J(t) &= J_{\perp}(t) + J_{\parallel}(t) 
	= Q \, W^{-1} \, \mathbf{J}(t) 
\end{align*}
are defined from the real current $\mathbf{J}(t)$ via $J_{\perp}(t) = P_+ \, W^{-1} \, \mathbf{J}(t)$ and $J_{\parallel}(t) = P_0 \, W^{-1} \, \mathbf{J}(t)$. 

Cutting out the negative part of the spectrum is not just necessary for a coherent physical interpretation, but that procedure also changes essential properties of $M$. Akin to Dirac operators and \emph{auxiliary} Maxwell operators have positive \emph{and} negative spectrum. The easiest way to see this is to recognize that the Maxwell equations are \emph{the} relativistic equations for a massless spin-$1$ particle \cite[pp.~197]{Wigner:representations_Lorentz_group:1939}. However, this is not to be interpreted as particles (positive frequency states) and antiparticles (negative frequency states), but rather an artifact of the \emph{complexification of the fields}. Indeed, the phase locking condition~\eqref{Maxwell_equations:eqn:phase_locking_condition} between positive and negative frequency wave functions, which stems from the reality of physical fields, tells us that positive and negative frequency fields are not independent degrees of freedom. Mathematical similarities do not always imply analogous physical interpretations. 

In quantum mechanics, the physical implication of the Hamilton operator's selfadjointness is conservation of total probability. Similarly, the selfadjointness of the Maxwell leads to conservation of total field energy in the absence of sources: according to the arguments preceding \eqref{Maxwell_equations:eqn:field_energy} the total field energy is nothing but the square of the weighted norm so that the $\scpro{\, \cdot \,}{\, \cdot \,}_W$-unitarity of the evolution group $\e^{- \ii t M}$ yields 
\begin{align*}
	\mathcal{E} \bigl ( \mathbf{E}(t) , \mathbf{H}(t) \bigr ) &= \Bscpro{Q \bigl ( \mathbf{E}(t) , \mathbf{H}(t) \bigr ) \,}{\, Q \bigl ( \mathbf{E}(t) , \mathbf{H}(t) \bigr )}_W 
	\\
	&= \Bscpro{\e^{- \ii t M} \, Q \bigl ( \mathbf{E}(0) , \mathbf{H}(0) \bigr ) \,}{\, \e^{- \ii t M} \, Q \bigl ( \mathbf{E}(0) , \mathbf{H}(0) \bigr )}_W 
	\\
	&= \Bscpro{Q \bigl ( \mathbf{E}(0) , \mathbf{H}(0) \bigr ) \,}{\, Q \bigl ( \mathbf{E}(0) , \mathbf{H}(0) \bigr )}_W 
	. 
\end{align*}
We will explain in Section~\ref{Schroedinger:Hil_vs_Hil_phys} below that this expression is consistent with the more general definition of field energy~\eqref{Maxwell_equations:eqn:integrated_field_energy_frequency} for dispersive media. 

Of course, this is a convenient \emph{choice} of scalar product and we could have gone with, say, the regular, unweighted scalar product. In fact, for a large share of the arguments the only restriction we need to impose on the scalar product is that the associated norms be equivalent; this is so that $\Hil$ is complete with respect to any of these norms. However, when using scalar products different from $\scpro{\, \cdot \,}{\, \cdot \,}_W$, we can no longer exploit many of the neat simplifications that are directly linked to the selfadjointness of $M$. For instance, Bloch functions of photonic crystals do not form an orthonormal basis with respect to the unweighted scalar product $\scpro{\, \cdot \,}{\, \cdot \,}$. 

Another useful consequence is the existence of a resolution of the identity via the so-called projection-valued measure associated to $M$, commonly (albeit inaccurately) referred to in the physics literature as the existence of a complete set of eigenfunctions. Concretely, this allows us to expand electromagnetic fields in terms of Bloch functions associated to periodic light conductors (also called \emph{photonic crystals}). 

\subsubsection{Equivalence of the Schrödinger formalism} 
\label{Schroedinger:formalism:reality}
Now we can put all the pieces together: in Sections~\ref{Schroedinger:auxiliary_operators:projections} and \ref{Schroedinger:formalism:states} we have shown that there is a one-to-one correspondence between real fields $(\mathbf{E},\mathbf{H}) = 2 \Re \Psi$ with square integrable amplitudes and complex fields of non-negative frequency $\Psi$, introduced the relevant Hilbert space and the Maxwell operator, and our efforts culminate in the main result of this paper: 
\begin{theorem}[Equivalence of Maxwell's equations and Schrödinger formalism]\label{Schroedinger:thm:equivalence_frameworks}~\\
	Suppose the material weights $W$ which describe the medium satisfy Assumption~\ref{Schroedinger:assumption:material_weights}. 
	The sources $\pmb{\rho}(t)$ and $\mathbf{J}(t)$ are assumed to only excite states supported by the medium (in the sense of Assumption~\ref{Maxwell_equations:assumption:current_density}) and their non-negative frequency contributions $\rho_+(t)$ and $J_+(t)$ satisfy charge conservation~\eqref{Maxwell_equations:eqn:approximate_Maxwell_equations:charge_conservation}. 
	And lastly, on the initial conditions $(\mathbf{E}_0,\mathbf{H}_0) \in L^2(\R^3,\R^6)$ we impose that the field energy shall be finite and that $\Psi_+(t_0) = Q (\mathbf{E}_0,\mathbf{H}_0)$ satisfies the constraint equation~\eqref{Maxwell_equations:eqn:approximate_Maxwell_equations:constraint}. Then under these technical conditions the following holds: 
	\begin{enumerate}[(1)]
		\item Equations~\eqref{Maxwell_equations:eqn:approximate_Maxwell_equations} and \eqref{Schroedinger:eqn:pm_Schroedinger_equation_em} are equivalent, \ie the solution 
		\begin{align}
			\Psi(t) = \e^{- \ii (t - t_0) M} \Psi(t_0) - \ii \int_{t_0}^t \dd s \, J(t)
			\label{Maxwell_equations:eqn:solution_Maxwell_equations_Schroedinger_form}
		\end{align}
		of the Schrödinger-type equation~\eqref{Schroedinger:eqn:pm_Schroedinger_equation_em} satisfies the non-negative frequency Maxwell equations~\eqref{Maxwell_equations:eqn:approximate_Maxwell_equations} with initial condition $\Psi(t_0) = Q (\mathbf{E}_0,\mathbf{H}_0)$ and current density $J(t) = Q \, W^{-1} \, \mathbf{J}(t)$. Conversely, the non-negative frequency solution $\Psi_+(t)$ to \eqref{Maxwell_equations:eqn:approximate_Maxwell_equations} with $\Psi_+(t_0) = Q (\mathbf{E}_0,\mathbf{H}_0)$ satisfies equation \eqref{Schroedinger:eqn:pm_Schroedinger_equation_em}. 
		\item The physical fields $\bigl ( \mathbf{E}(t) , \mathbf{H}(t) \bigr )$ can be uniquely reconstructed from its complex representative $\Psi(t) = \Psi_+(t) = Q \bigl ( \mathbf{E}(t) , \mathbf{H}(t) \bigr )$. 
	\end{enumerate}
\end{theorem}
For physicists we reckon our systematic derivation is persuasive. However, for the benefit of more mathematically minded people, we have included a proof in Appendix~\ref{appendix:1_to_1_correspondence}. The key ingredients are the one-to-one correspondence between physical fields and complex states representing them (Proposition~\ref{Schroedinger:prop:identification_real_complex_vector_spaces}) as well as showing that the frequency decomposition via the Fourier transform in time is the same as that obtained from the spectral decomposition of $M$. 

\subsection{Complexified equations for gyrotropic materials} 
\label{Schroedinger:complexified_equations}
For non-gyrotropic media where thanks to $W = \overline{W}$ the two sets of  equations~\eqref{Maxwell_equations:eqn:approximate_Maxwell_equations} coincide, and it is much more straight-forward to treat those equations. The easiest way to proceed here is to mathematically admit complex electromagnetic fields and consider 
\begin{align*}
	M_{\C} = W^{-1} \, \Rot
\end{align*}
on the complexified Hilbert space $\Hil_{\C} = L^2_W(\R^3,\C^6)$ \emph{without} making any frequency constraints. The dynamical equation~\eqref{Maxwell_equations:eqn:approximate_Maxwell_equations:dynamics} can be recast as 
\begin{align}
	\ii \, \partial_t \Psi &= M_{\C} \Psi - \ii \, W^{-1} \, \mathbf{J}(t) 
	, 
	&&
	\Psi(0) = (\mathbf{E}_0,\mathbf{H}_0)
	, 
	\label{Schroedinger:eqn:nongyrotropic_Maxwell_Schroedinger_equation}
\end{align}
where $\mathbf{J}(t)$ is the real current and $\Psi(t)$ is just a relabelling of $\bigl ( \mathbf{E}(t) , \mathbf{H}(t) \bigr )$ to make the similarity to quantum mechanics more explicit. Admitting complex fields of course \emph{doubles} the degrees of freedom, because mathematically we have to admit complex electromagnetic fields for technical reasons. These equations and their non-linear generalizations have been studied extensively in the literature (see \eg \cite{Birman_Solomyak:L2_theory_Maxwell_operator:1987,Figotin_Klein:localization_classical_waves_II:1997,Kuchment:math_photonic_crystals:2001,Babin_Figotin:nonlinear_Maxwell:2003,DeNittis_Lein:sapt_photonic_crystals:2013}). 

Using the Helmholtz decomposition and local charge conservation~\eqref{Maxwell_equations:eqn:approximate_Maxwell_equations:charge_conservation}, it is straightforward to verify that the solution $\Psi(t)$ also satisfies the constraint equation~\eqref{Maxwell_equations:eqn:approximate_Maxwell_equations:constraint}. 

The reality of the waves emerges from a symmetry that $M_{\C}$ necessarily possesses: for non-gyrotropic materials where $W = \overline{W} = C \, W \, C$ holds, complex conjugation \emph{anti}commutes with the Maxwell operator, 
\begin{align*}
	C \, M_{\C} \, C = - M_{\C}
	. 
\end{align*}
In the parlance of the Cartan-Altland-Zirnbauer scheme \cite{Altland_Zirnbauer:superconductors_symmetries:1997,Schnyder_Ryu_Furusaki_Ludwig:classification_topological_insulators:2008} $C$ is an “even particle-hole-type symmetry”, although we would like to reiterate what we stated in Section~\ref{Schroedinger:auxiliary_operators:projections}, namely that the terminology has been created with quantum mechanics in mind and its interpretation does not carry over to electromagnetism. Due to phase locking~\eqref{Maxwell_equations:eqn:phase_locking_condition} “particles” and “holes” are not independent degrees of freedom. 

Consequently, $C$ and therefore also the real part $\Re = \tfrac{1}{2} (\id + C)$ \emph{commute} with the evolution group, 
\begin{align*}
	\Re \, \e^{- \ii t M_{\C}} = \e^{- \ii t M_{\C}} \, \Re 
	. 
\end{align*}
This is of course just a fancy way of saying that initially real electromagnetic fields remain real. If one restrict oneself to real initial states, then these complexified equations indeed give a simple way to study Maxwell's equations using techniques from quantum mechanics. Real solutions of course retain the symmetry $\Psi_- = \overline{\Psi_+}$ where now $\Psi_- = P_- (\mathbf{E},\mathbf{H})$ is computed with the help of the spectral projection $P_- = 1_{(-\infty,0)}(M_{\C})$ onto the negative frequencies. 

It would be tempting to surmise that a version of this simplified complexification procedure also works in case $W \neq \overline{W}$. But our first-principles derivation of \eqref{Maxwell_equations:eqn:approximate_Maxwell_equations} in Section~\ref{Maxwell_equations} shows that this is unfortunately not the case. Postulating \eqref{Schroedinger:eqn:nongyrotropic_Maxwell_Schroedinger_equation} for complex weights leads to unphysical electromagnetic fields, because even if they are initially real, they acquire a non-zero imaginary part over time — in direct contradiction with one of the basic tenets of electromagnetism, that physical fields are real. The correct, albeit somewhat complicated complexification procedure is to write $M_{\C} = M_+ \oplus M_-$ as the direct sum of positive and negative frequency operators acting on $\Hil_{\C} = \Hil_+ \oplus \Hil_-$ (see \cite[Section~2.2]{DeNittis_Lein:ray_optics_photonic_crystals:2014} for details). But in many situations there is no advantage to work with the complexified equations; in fact for things such as the topological classification of electromagnetic media it is detrimental \cite{DeNittis_Lein:symmetries_electromagnetism:2017}. Indeed, Theorem~\ref{Schroedinger:thm:equivalence_frameworks} guarantees that \emph{all} of the physical phenomena are described by the non-negative frequency equations, because they are equivalent to Maxwell's equations. 

\subsection{Comparison with the dispersive Hilbert space $\Hil$ from Section~\ref{Maxwell_equations:hilbert_space}} 
\label{Schroedinger:Hil_vs_Hil_phys}
We close this section by comparing the Hilbert spaces $\Hil = \int_0^{\infty} \dd \omega \, \Hil(\omega)$ introduced in Sections~\ref{Maxwell_equations:hilbert_space} and $\Hil = Q \bigl [ L^2_W(\R^3,\C^6) \bigr ]$ from Section~\ref{Schroedinger:formalism}; it turns out that in the absence of dispersion these two Hilbert spaces are one and the same. Because their definitions initially seem very different from one another, so it is worthwhile to make the necessary identifications explicit. 

As we are only interested in real states, we may restrict our discussion to $\omega \geq 0$. Suppose the weights $\widehat{W}(\omega) = W$, $\omega \geq 0$, are frequency-independent. Then for fixed sign of the frequency all the scalar products coincide, 
\begin{align*}
	\scpro{\Phi}{\Psi}_{\omega} = \Bscpro{\Phi}{\tfrac{\dd }{\dd \omega} \bigl ( \omega \, \widehat{W}(\omega) \bigr ) \, \Psi} = \scpro{\Phi}{W \, \Psi} = \scpro{\Phi}{\Psi}_W 
	, 
\end{align*}
and we have to figure out how to identify $\Hil(\omega)$ using properties of the Maxwell operator $M$. The key ingredient is the resolution of the identity 
\begin{align*}
	\id_{\Hil} = \int_0^{\infty} \dd 1_{\omega}(M) 
\end{align*}
provided by the so-called projection-valued measure \cite[Chapter~VIII.3]{Reed_Simon:M_cap_Phi_1:1972}; simply put, this is the generalization of expanding the identity in terms of eigenfunctions of a selfadjoint (hermitian) operator in case the operator's spectrum does not just consist of eigenvalues. Here, $\dd 1_{\omega}(M)$ projects onto the (usually infinitesimal) sliver of the Hilbert space composed of states that rotate with frequency $\omega$ when time evolved, \ie it consists of (generalized) solutions to 
\begin{align*}
	M \varphi_{\omega} = \omega \, \varphi_{\omega}
	. 
\end{align*}
But this is \emph{exactly} the definition of $\Hil(\omega)$, and we see that $\Hil = \int_0^{\infty} \dd \omega \, \Hil(\omega)$ and $\Hil = Q \bigl [ L^2_W(\R^3,\C^6) \bigr ]$ coincide not just as vector spaces, but because they share their scalar product also as Hilbert spaces. Put another way, we have just proven 
\begin{lemma}[Equivalence of Hilbert spaces]\label{Schroedinger:lem:identification_dispersive_nondispersive_Hilbert_spaces}
	Suppose the material weights satisfy Assumption~\ref{Schroedinger:assumption:material_weights}. Then $\Hil$ defined in equation~\eqref{Maxwell_equations:eqn:dispersive_non_negative_frequency_Hilbert_space} for the constant frequency weights $\widehat{W}(\omega) = W$, $\omega \geq 0$, coincides with $\Hil = 1_{[0,\infty)}(\Maux) \, \bigl [ L_W^2(\R^3,\C^6) \bigr ]$ as Hilbert spaces with energy scalar product $\scpro{\, \cdot \,}{\, \cdot \,}_W$. 
\end{lemma}
Because we can canonically identify these two Hilbert spaces, we also justify calling $\scpro{\, \cdot \,}{\, \cdot \,}_W$ energy scalar product and $\norm{\cdot}_W^2$ the energy norm: for frequency-\emph{in}dependent material weights, $\frac{\dd}{\dd \omega} \bigl ( \omega \, \widehat{W}(\omega) \bigr ) = W$ holds and $\norm{\Psi}_W^2 = \scpro{\Psi}{W \, \Psi}$ computes the energy content of the real electromagnetic field $(\mathbf{E},\mathbf{H})$ represented by $\Psi = Q (\mathbf{E},\mathbf{H})$. 
\section{Comparison to previous works and quantum mechanics} 
\label{comparison}
Let us close this Section by making a comparison to earlier works and quantum mechanics. For simplicity, let us repeat the list of all three ingredients: 
\begin{enumerate}[(S1)]
	\item \emph{States} are represented by \emph{complex fields of non-negative frequency}, and they form the Hilbert space $\Hil = \Jphys \oplus \mathcal{G}$ composed of transversal fields of positive frequency $\Jphys$ and longitudinal gradient fields $\mathcal{G}$. The representative $\Psi = Q (\mathbf{E},\mathbf{H})$ is obtained by projecting the real field onto the non-negative frequencies. 
	\item The \emph{Maxwell operator} $M = W^{-1} \, \Rot \big \vert_{\omega \geq 0}$ plays the role of the Hamiltonian in the \emph{Schrödinger-type equation}
	\begin{align*}
		\ii \, \partial_t \Psi(t) = M \Psi(t) - \ii \, W^{-1} \, J(t)
		, 
		&&
		\Psi(0) = Q (\mathbf{E}_0,\mathbf{H}_0)
		. 
	\end{align*}
	\item The \emph{real, physical fields} are recovered by taking the real part, $\bigl ( \mathbf{E}(t) , \mathbf{H}(t) \bigr ) = 2 \Re \Psi(t)$. 
\end{enumerate}
The range of validity of the Schrödinger formalism (as covered in Section~\ref{Schroedinger:formalism}) has the exact same range of validity as equations~\eqref{Maxwell_equations:eqn:approximate_Maxwell_equations}: they hold for approximately monochromatic waves composed of frequencies $\omega \approx \omega_0$ so that $\widehat{W}(\omega)$ does not change appreciably. And while for non-gyrotropic media, our somewhat elaborate construction is unnecessary, it is compulsory for media with $W \neq \overline{W}$ in order to obtain the correct and physically meaningful equations. Given the role that gyrotropic media play in novel applications such as topologically non-trivial photonic crystals \cite{Raghu_Haldane:quantum_Hall_effect_photonic_crystals:2008,Wang_et_al:edge_modes_photonic_crystal:2008}, we expect that the Schrödinger formalism for Maxwell's equation we have derived here will prove useful when investigating quantum-light analogies.

\subsection{Quantum–wave analogies revisited} 
\label{comparison:quantum_mechanics}
Historically speaking, the first quantum-wave analogies were actually wave-quantum analogies: the founding fathers of modern quantum mechanics relied on the similarity to waves to find their bearing and arrive at a consistent interpretation of quantum theory. Nowadays, these are often “read in reverse”, where quantum phenomena serve as inspiration for finding novel effects in classical waves. Two prime examples are Yablovnovitch's proposal of photonics crystals \cite{Yablonovitch:photonic_band_gap:1993} and Haldane's seminal insight that topological effects are not inherently quantum, but \emph{bona fide} wave effects \cite{Raghu_Haldane:quantum_Hall_effect_photonic_crystals:2008}. Both of these extremely influential works have ended up creating whole subfields in multiple communities. 

Most works on quantum-light or, more broadly, quantum-\emph{wave} analogies, \emph{focus on similarities in the dynamical equation}. For example, Rechtsman et al \cite{Rechtsman_Zeuner_et_al:photonic_topological_insulators:2013} based their prediction of the existence of topological edge modes in periodic waveguide arrays on the similarity between the paraxial wave equation and the Schrödinger equation; in their setup the fussili-like twist of the waveguides leads to a vector potential that enters the “paraxial Hamiltonian” the same way a magnetic vector potential would in a quantum Hamiltonian. That explains how the twist leads to the breaking of time-reversal symmetry, and ultimately, to the presence of unidirectional edge modes. 

Compared to quantum systems the Schrödinger form of Maxwell's equations has two important differences, namely that (1)~Maxwell-type operators $M = W^{-1} \, D$ have a \emph{product structure} instead of being a \emph{sum}, and (2)~if $\mathbf{J}(t) \neq 0$ a source term is present, something which has no quantum analog. The product structure prevents us from applying many techniques from perturbation theory, linear response theory or scattering theory, because they rely on the sum structure of the Hamiltonian $H = \frac{1}{2m} (-\ii \hbar \nabla)^2 + V$. New methods tailored to Maxwell-type operators need to be developed. 
\medskip

\noindent
Establishing quantum-wave analogies systematically often requires going beyond a comparison of just the dynamics, also the nature of \emph{observables} and \emph{additional information} on typical regimes and configurations play a significant role (see Table~\ref{comparison:table:comparison_classical_waves_quanmtum_mechanics}). Quantum observables are defined in terms of selfadjoint operators like position, momentum or spin, and as a matter of principle, the quantum wavefunction itself is not observable. The situation in electromagnetism is very different: not only are the fields themselves measurable quantities, observables are \emph{functionals of the fields}. To give one specific example where this enters in crucial steps of the analysis, in our work on the derivation of ray optics equations in photonic crystals \cite{DeNittis_Lein:ray_optics_photonic_crystals:2014} the specific form of observables was essential when we wanted to find out \emph{in what sense} and \emph{quantify how well} ray optics approximates electrodynamics. Because we borrowed a semiclassical technique from quantum mechanics, Weyl quantization, we were only able to establish a ray optics limit for \emph{quadratic} observables which can be written as a “quantum expectation value”. For such observables, the quantum-light analogy lines up nicely: also such electromagnetic observables come in pairs, a ray optics observable on phase space and a functional on the fields — just like classical momentum is associated to the momentum operator, its quantization. Nevertheless, not in all cases do we recover what is naïvely expected: the ray optics observable associated to local averages of the Poynting vector $\mathcal{P}(\mathbf{E},\mathbf{H}) = \tfrac{1}{2} \, \Re \bigl ( \overline{\mathbf{E}} \times \mathbf{H} \bigr )$ is \emph{not} $\dot{r}$, but of a more complicated form \cite[Section~3.3.2]{DeNittis_Lein:ray_optics_photonic_crystals:2014}. 
\begin{table}
	\newcolumntype{A}{>{\centering\arraybackslash\normalsize} m{2.3cm} }
	\newcolumntype{B}{>{\centering\arraybackslash\normalsize} m{3.8cm} }
	\newcolumntype{C}{>{\centering\arraybackslash\normalsize} m{3.8cm} }
	\renewcommand{\arraystretch}{1.15}
	\begin{center}
		\begin{tabular}{A | B | B}
			 & \emph{Classical Waves} & \emph{Quantum Mechanics} \\ \hline \hline 
			Generator of dynamics & $M = W^{-1} \, D$ \linebreak (product structure) & \eg $H = \frac{1}{2m} (- \ii \hbar \nabla)^2 + V$ \linebreak (sum structure) \\ \hline 
			Hilbert space & $\Hil_+ \subset L^2_W(\R^d,\C^n)$ \linebreak (weighted, $\omega \geq 0$ states) & $L^2(\R^d,\C^n)$ \linebreak (unweighted) \\ \hline 
			Wave function & $\R$-valued & $\C$-valued \\ \hline 
			Conserved quantity $\norm{\psi}^2$ & \eg field energy \linebreak in electromagnetism & probability \\ \hline 
			Observables & Functionals of the fields & Selfadjoint operators \\
		\end{tabular}
	\end{center}
	\caption{Comparison of quantum mechanics and Schrödinger formalism of classical fields}
	\label{comparison:table:comparison_classical_waves_quanmtum_mechanics}
\end{table}

Fundamental differences in the physics, the form of typical states and regimes, can result in vastly different phenomenology despite relatively similar mathematics. One of our motivations to study the ray optics limit was to investigate whether it is possible to link the “transverse conductivity” in topological photonic crystals to a Chern number using “semiclassical” arguments (we discuss this point in detail in \cite[Section~5.1]{DeNittis_Lein:ray_optics_photonic_crystals:2014}). But there are two fundamental obstacles to that: first of all, as we mentioned above the ray optics observable associated to the Poynting vector is not $\dot{r}$. But even ignoring that for the moment, a second important assumption is not fulfilled: unlike in solid state physics completely filled (frequency) bands in photonics make no physical sense. Instead, experimentalists typically excite wave packets with well-defined momentum and frequency (\eg with laser light). 

\subsection{Comparison to the literature} 
\label{comparison:literature}
While the general idea of exploiting similarities between quantum mechanics and Maxwell's equations is as old as modern quantum mechanics itself, there are several different approaches on how to go beyond heuristics and systematically compare classical and quantum waves. Broadly speaking, these approaches fall into one of three categories: (1)~via the second-order equations, (2)~in the sense of a generalized eigenvalue problem and (3)~using the first-order formalism akin to the one presented here.

\subsubsection{Second-order formalism} 
\label{comparison:literature:second_order_formalism}
An introduction to the second-order formalism can be found \eg in \cite[pp.~10–16]{Joannopoulos_Johnson_Winn_Meade:photonic_crystals:2008}: in the absence of charges and currents, $\pmb{\rho} = 0$ and $\mathbf{J}(t) = 0$ and coupling between electric and magnetic field ($\chi = 0$), squaring Maxwell's equations yields two \emph{seemingly} uncoupled equations, 
\begin{align}
	\partial_t^2 \Psi + M^2 \Psi 
	= \left (
	\begin{matrix}
		\bigl ( \partial_t^2 + M^2_{EE} \bigr ) \psi^E \\
		\bigl ( \partial_t^2 + M^2_{HH} \bigr ) \psi^H \\
	\end{matrix}
	\right )
	= 0
	, 
	\label{Schroedinger:eqn:second_order_Maxwell_equations}
\end{align}
one for the electric and one for the magnetic field; we emphasize that they only decouple at first glance, because second-order equations not only require $\Psi(0)$ as initial condition, but also $\partial_t \Psi(0)$, and the latter connects $\mathbf{E}$ to $\mathbf{H}$ and vice versa. One of the two operators, \eg the magnetic one 
\begin{align}
	M^2_{HH} \psi^H = \mu^{-1} \, \nabla \times \bigl ( \eps^{-1} \, \nabla \times \psi^H \bigr )
	, 
	\label{Schroedinger:eqn:second_order_Maxwell_equation}
\end{align}
is then regarded as the analog of the quantum Hamiltonian. In many media where $\mu = 1$ holds to good approximation and $\eps = \overline{\eps}$ real, $M^2_{HH}$ is selfadjoint (hermitian) on the unweighted Hilbert space $L^2(\R^3,\C^3)$. 

Certainly the second-order formalism is a perfectly valid and equivalent way of studying Maxwell's equations, but it is \emph{not} ideally suited to transfer methods from quantum mechanics. That is because many methods — explicitly or implicitly — rely on the Schrödinger equation being \emph{first} order in time. For instance, the selfadjointness of $M_{HH}^2$ means that $\e^{- \ii t M_{HH}^2}$ is well-defined, but it does \emph{not} solve the dynamical equation~\eqref{Schroedinger:eqn:second_order_Maxwell_equation}; hence, it is not apt to call $M^2_{HH}$ the analog of the quantum Hamiltonian. Therefore, standard techniques such Duhamel arguments, which allow us to make perturbation expansions of the evolution group, do not readily apply. 

Another example which will be explored in a future work \cite{DeNittis_Lein:symmetries_electromagnetism:2017} is the topological classification of photonic crystals, adapting the Cartan-Altland-Zirnbauer scheme \cite{Altland_Zirnbauer:superconductors_symmetries:1997,Schnyder_Ryu_Furusaki_Ludwig:classification_topological_insulators:2008,Chiu_Teo_Schnyder_Ryu:classification_topological_insulators:2016}. Here, the distinction between commuting and anticommuting symmetries is crucial for a proper classification. But we can no longer distinguish between operators $U$ which commute or anticommute with $M$, because in either case $U \, M^2 \, U^{-1} = (\pm 1)^2 \, M^2 = M^2$ holds. This was the source of confusion for calling complex conjugation $C$ a \emph{time-reversal} symmetry in \cite{Raghu_Haldane:quantum_Hall_effect_photonic_crystals:2008} rather than an even particle-hole-type symmetry; the actual time-reversal symmetry that is broken for gyrotropic media is 
\begin{align}
	T : (\psi^E,\psi^H) \mapsto \bigl ( \overline{\psi^E} , - \overline{\psi^H} \bigr )
	. 
	\label{comparison:eqn:time_reversal_symmetry}
\end{align}
Related to the topic of topological phenomena are Chern numbers: as we will show in \cite{DeNittis_Lein:symmetries_electromagnetism:2017}, the only topological invariants which play a role for topological photonic crystals are (first and second) Chern numbers. First Chern numbers can be computed from the Berry curvature, and mathematically speaking, we have three distinct choices: the “electromagnetic” Chern number is given in terms of the Berry curvature computed from electromagnetic Bloch functions $\varphi_n(k) = \bigl ( \varphi_n^E(k) , \varphi_n^H(k) \bigr )$. But alternatively, we may use $\varphi_n^E(k)$ and $\varphi_n^H(k)$ to define an “electric” and a “magnetic” Chern number. We see no obvious reason why they should be related, \eg why the sum of electric and magnetic Chern number should yield the electromagnetic Chern number. This is evidently a problem if we want to obtain photonic bulk-boundary correspondences. This is significant, because we have found a number of publications which compute Chern numbers from the electric field alone (\eg \cite{Wang_et_al:edge_modes_photonic_crystal:2008,Jacobs_et_al:topological_photonic_crystal_tellegen:2005}). For qualitative predictions whether or not unidirectional, backscattering-free edge modes exist, finding a non-zero electric Chern number, for example, may suffice, because that necessitates breaking time-reversal symmetry~\eqref{comparison:eqn:time_reversal_symmetry}. 

Nevertheless, we reckon that the second-order formalism is suited to investigate \emph{wave}-wave analogies and when adapting techniques initially developed for other wave equations. 
\medskip

\noindent
However, there is also a deeper, more subtle issue, which to our knowledge have not been addressed in the literature before. As equation~\eqref{Schroedinger:eqn:second_order_Maxwell_equation} stands, it only holds for \emph{non-gyrotropic media}. Here, $U : (\psi^E,\psi^H) \mapsto (\psi^E,-\psi^H)$ is an anticommuting, unitary symmetry of $\Maux$, 
\begin{align*}
	U \, \Maux \, U^{-1} &= - \Maux
	,
\end{align*}
that maps complex fields composed of non-negative frequencies onto fields composed solely of non-positive frequencies. Therefore, $\Hil^H = \bigl \{ \psi^H \; \; \vert \; \; (\psi^E , \psi^H) \in \Hil \bigr \}$, by which we mean the magnetic component of the non-negative frequency Hilbert space $\Hil$ defined in Section~\ref{Schroedinger:formalism:states} and endowed with the scalar product $\scpro{\phi^H}{\mu \psi^H}$, coincides with $L^2_{\mu}(\R^3,\C^3) = \Hil^H$. Put succinctly, we may start with \eqref{Schroedinger:eqn:second_order_Maxwell_equation} acting on all of $L^2_{\mu}(\R^3,\C^3)$. 
So even though squaring the auxiliary Maxwell operator $\Maux_+$ destroys the information on the sign, we are still able to pick a sign of $\pm \omega$ by choosing a sign of $\pm \partial_t \psi^H = \mp \mu^{-1} \nabla \times \psi^E$. Again, the electric and the magnetic equation couple via the initial condition $\partial_t \bigl ( \psi^E(0) , \psi^H(0) \bigr )$. Heuristically, this is consistent with the heuristic argument that $L^2_{\mu}(\R^3,\C^3)$ and $L^2(\R^3,\R^6)$ are of “the same dimension” if we identify $\C$ with $\R^2$. 

However, when the weights $(\eps , \mu) \neq ( \overline{\eps} , \overline{\mu} )$ have a non-vanishing imaginary part, we start from \emph{two} sets of Maxwell equations and evolve waves composed of frequencies $\pm \omega \geq 0$ with \emph{different} material weights; this has nothing to do with the formalism one chooses, but is merely a consequence of the real-valuedness of electromagnetic fields (see the discussion in Section~\ref{Maxwell_equations:heuristic_monochromatic_approximation}). The insight that the information contained in the negative frequency contribution is redundant of course also applies here, and we can study \eqref{Schroedinger:eqn:second_order_Maxwell_equation} as an equation on the Hilbert space $\Hil^H \subseteq L^2_{\mu}(\R^3,\C^3)$. 

This restriction to $\Hil^H$ is now \emph{more subtle}, and necessitates the use of the first-order formalism. Making the weights complex breaks the symmetry $U : (\psi^E,\psi^H) \mapsto (\psi^E,-\psi^H)$, and it is not automatic that $\Hil^H$ coincides with all of $L^2_{\mu}(\R^3,\C^3)$. Therefore, if one \emph{starts} with \eqref{Schroedinger:eqn:second_order_Maxwell_equation} (or its electric counterpart) on all of $L^2_{\mu}(\R^3,\C^3)$, then if $\Hil^H \neq L^2_{\mu}(\R^3,\C^3)$ it is no longer possible to discard unphysical solutions. In case of photonic cyrstals where the weights are periodic, this would lead to artificial band crossings (where at least one of the bands involved is due to an unphysical Bloch function). 

\subsubsection{Generalized eigenvalue problem} 
\label{comparison:literature:generalized_eigenvalue_problem}
Generalized eigenvalue problems are of the form $D \Psi = \omega \, W \Psi$; this corresponds to starting with \eqref{Maxwell_equations:eqn:approximate_Maxwell_equations:dynamics} and is the approach adopted in \cite[pp.~3]{Raghu_Haldane:quantum_Hall_effect_photonic_crystals:2008} and many other works \cite{Silveirinha:Z2_topological_index_continuous_photonic_systems:2016}. The main drawback is that this no longer gives a dynamical description of the system. 

Nevertheless, more general types of media can be treated in this framework, including those with dispersion and with zero modes where $W$ is no longer invertible. However, for lossless positive index materials, and that is the case typically considered in the theory of photonic crystals, we can invert $W$ and use the Schrödinger formalism instead. 

\subsubsection{First-order formalism} 
\label{comparison:literature:first_order_formalism}
The use of the Schrödinger formalism in the sense it is discussed here, where (S1)–(S3) from Section~\ref{Schroedinger:abstract:classical_waves} are spelled out explicitly, is surprisingly rare in the physics literature. While quite a few write down the Maxwell operator (see \eg \cite[equations~(7)–(11)]{Ruiz_Dodin:geometrical_optics_polarization_effects:2015} or \cite[equation~(39)]{Proskurin_Ovchinnikov_Nosov_Kishine:optical_chirality_gyrotropic_media:2016}), very few supply also the scalar product and discuss selfadjointness $M = M^{\ast_W}$ of the Maxwell operator. Most exceptions are due to more mathematically minded theoretical physicists and mathematical physicists (see \eg \cite{Babin_Figotin:nonlinear_Maxwell_1:2001,Babin_Figotin:nonlinear_Maxwell_4:2005,Tip:linear_absorptive_dielectrics:1998,Tip_Moroz_Combes:band_structure_absorptive_photonic_crystals:2000,DeNittis_Lein:symmetries_Maxwell:2014,Gangaraj_Silveirinha_Hanson:Berry_connection_Maxwell:2017}). 

Among this group, very, very few discuss how the real-valuedness of the physical fields can be squared with the unavoidable complexification of the equations. One notable exception are the works by Dodin \cite{Dodin:geometric_view_noneikonal_waves:2014} as well as Ruiz and Dodin \cite{Ruiz_Dodin:geometrical_optics_polarization_effects:2015,Ruiz_Dodin:lagrangian_geometrical_optics_vector_waves_spin_particles:2015}: they not only explain how to represent real fields with complex waves of positive frequencies in \cite[Sections~2.2–2.3]{Dodin:geometric_view_noneikonal_waves:2014}, but flesh out how to systematically exploit that many classical wave equations admit a Schrödinger formalism. Of particular interest to them are the derivation of ray optics equations in adiabatically perturbed homogeneous isotropic media as well as applications to plasma physics \cite{Ruiz:geometric_theory_waves_and_plasma_physics:2017}. While the present work covers much more general situations (including non-homogeneous, non-isotropic media with potentially complex weights) and gives a mathematically rigorous definition of Maxwell-type operators, the overarching strategies are very similar. 

We furthermore would like to highlight the contributions by Babin and Figotin to better understand \emph{non-linear} photonic crystals. In \cite{Babin_Figotin:nonlinear_Maxwell_4:2005} they derived the non-linear Schrödinger equation from Maxwell's equations for non-linear, \emph{non}-gyrotropic media via a suitable scaling limit. Here, the even “particle-hole”-type symmetry $C$ is responsible for the point symmetry of the band spectrum (\cf \cite[Figure~3]{Babin_Figotin:nonlinear_Maxwell_4:2005}), which in turn ensures that it is always possible to satisfy the so-called frequency matching condition. The condition $\Psi_- = C \Psi_+$ further allows them to obtain a \emph{scalar} non-linear Schrödinger equation rather than a \emph{pair} of coupled non-linear Schrödinger equations (\ie going from equations~(98)–(99) to equation~(108) in \cite{Babin_Figotin:nonlinear_Maxwell_4:2005}). While Babin and Figotin always assumed to work with non-gyrotropic media where $W = \overline{W}$ \cite[equation~(6)]{Babin_Figotin:nonlinear_Maxwell_4:2005}, we expect that all of their results straightforwardly extend to the non-gyrotropic case if the Schrödinger formalism developed here is employed. 

Noteworthy are also the works by by Tip et al \cite{Tip:linear_absorptive_dielectrics:1998,Tip_Moroz_Combes:band_structure_absorptive_photonic_crystals:2000} who develop a Schrödinger formalism for \emph{dispersive}, absorptive media. This goes beyond what we do here as we have neglected dispersion. To deal with dispersion, they address the concerns discussed at the end of Section~\ref{Maxwell_equations:heuristic_monochromatic_approximation} by choosing a specific initial trajectory, namely those that are $0$ for $t < t_0$ and “spring to life” at the initial time $t_0$; this . The second main idea is to add fictitious electromagnetic fields to the equations so that the total field energy balance is restored; the absorbed field energy is now used to excite these fictitious fields. 

What distinguishes our work from these earlier publications, including our own \cite{DeNittis_Lein:adiabatic_periodic_Maxwell_PsiDO:2013,DeNittis_Lein:sapt_photonic_crystals:2013,DeNittis_Lein:symmetries_Maxwell:2014}, is that ours gives a physically meaningful description also of gyrotropic materials where $W \neq \overline{W}$. 
\section{Obtaining Conserved Quantities in the Schrödinger Formalism} 
\label{conserved}
We close the discussion of Maxwell's equations with an application of the Schrödinger formalism. Perhaps surprisingly, the topic of conserved quantities is still an area of active research in electromagnetism (see \eg \cite{Bliokh_Bekshaev_Nori:dual_electromagnetism:2013,Proskurin_Ovchinnikov_Nosov_Kishine:optical_chirality_gyrotropic_media:2016} and references therein), in particular in media. The default strategy is nicely explained in \cite{Bliokh_Bekshaev_Nori:dual_electromagnetism:2013} for the \emph{in vacuo} equations: given that Maxwell's equations are \emph{the} relativistic equations for a massless spin-$1$ particle, their symmetry group is the Poincaré group. To each continuous symmetry there exists a conserved quantity, \eg the rotational symmetry leads to conservation of total angular momentum and “dual symmetry”
\begin{align}
	\left (
	\begin{matrix}
		\mathbf{E} \\
		\mathbf{H} \\
	\end{matrix}
	\right ) \mapsto \left (
	\begin{matrix}
		\cos \vartheta & - \sin \vartheta \\
		\sin \vartheta & \cos \vartheta \\
	\end{matrix}
	\right ) \left (
	\begin{matrix}
		\mathbf{E} \\
		\mathbf{H} \\
	\end{matrix}
	\right )
	\label{conserved:eqn:continuous_dual_symmetry}
\end{align}
ensures conservation of helicity. Media break some, perhaps all of these symmetries. Therefore, the question arises that \emph{given} material weights, how can we systematically find conserved quantities? The standard approach is to reformulate Maxwell's equation in the \emph{Lagrangian formalism} and to exploit \emph{Noether's theorem}. Even in vacuum, this approach has unexpected complications, something that Bliokh, Beshaev and Nori explain very elegantly in  \cite{Bliokh_Bekshaev_Nori:dual_electromagnetism:2013}. We will briefly summarize the main points of their discussion. First of all, the Lagrangian does not have all the symmetries of the Maxwell equations. The standard Lagrangian density $\mathcal{L}(\mathbf{E},\mathbf{H}) = \mathbf{E}^2 - \mathbf{H}^2$ breaks the dual symmetry~\eqref{conserved:eqn:continuous_dual_symmetry}, so identifying symmetries of the physical equations directly from the Lagrangian is harder. Moreover, the conserved quantities obtained from this Lagrangian lack a clear physical interpretation. 

Bliokh et al propose a dual symmetric Lagrangian formulation of electromagnetism where the Lagrangian density involves two $\R^6$-valued vector fields $F$ and $G$ as independent variables, and then impose the relation $F = (\mathbf{E} , \mathbf{H})$ and $G = (\mathbf{H} , -\mathbf{E})$ afterwards. This corresponds to complexifying the fields and then restricting to real solutions afterwards. The complexification is necessary, because the generator for the continuous transformation~\eqref{conserved:eqn:continuous_dual_symmetry} is $J = \sigma_2 \otimes \id$ and its eigenvectors, the helicity basis, are \emph{complex}. Consequently, the conserved quantities they derive inherit the evident dual symmetry of the Lagrangian density, are more symmetric under exchange of $\mathbf{E}$ and $\mathbf{H}$ and are more “consistent” (Figure~1 and the discussion in Sections~2–3 of \cite{Bliokh_Bekshaev_Nori:dual_electromagnetism:2013} make this precise). 
\medskip

\noindent
While Bliokh et al's paper conclusively answers the question of systematically finding physically meaningful conserved electromagnetic observables \emph{in vacuum}, we believe the Schrödinger formalism offers additional insight into the topic. Chiefly, we are immediately able to extend the discussion to non-homogeneous electromagnetic media where the material weights $\eps$, $\mu$ and $\chi$ may depend on position. While we have no doubt that it is possible to find a version of the dual symmetric formalism for media, this requires no extra effort in the Schrödinger formalism. Secondly, as we shall see below, identifying conserved quantities and associated currents becomes straightforward.

\subsection{The nature of electromagnetic observables} 
\label{conserved:observables}
Electromagnetic observables are generically functionals $\mathcal{F} : L^2(\R^3,\R^6) \longrightarrow \R$ of the fields, \ie functions associate real numbers to electromagnetic fields. The two most obvious examples are the energy density 
\begin{align*}
	\mathcal{E}_x(\mathbf{E},\mathbf{H}) = \bigl ( Q_+(\mathbf{E},\mathbf{H}) \bigr )(x) \cdot W(x) \bigl ( Q_+(\mathbf{E},\mathbf{H}) \bigr )(x)
\end{align*}
or the volume integrals of components of the Poynting vector 
\begin{align}
	\mathcal{P}_j(\mathbf{E},\mathbf{H}) &= \int_{\R^3} \dd x \, \mathcal{P}_{x,j}(\mathbf{E},\mathbf{H})
	= \int_{\R^3} \dd x \, \tfrac{1}{2} \Re \bigl ( \overline{\psi^E} \times \psi^H \bigr )_j 
	\label{conserved:eqn:Poynting_vector}
	\\
	&= \Bscpro{Q_+ (\mathbf{E},\mathbf{H}) \,}{\, W^{-1} \, \bigl ( - \ii \sigma_2 \otimes e_j^{\times} \bigr ) \, Q_+ (\mathbf{E},\mathbf{H})}_W 
	\notag
\end{align}
Here, $\Psi = (\psi^E,\psi^H) = Q_+(\mathbf{E},\mathbf{H})$ is the complex positive frequency component of the real field $(\mathbf{E},\mathbf{H})$, $e_j$ is one of the three canonical basis vectors $e_1 = (1,0,0)$, $e_2 = (0,1,0)$ or $e_3 = (0,0,1)$ and $e_j^{\times} \mathbf{E} = e_j \times \mathbf{E}$ is a convenient notation to associate a matrix $e_j^{\times}$ to the vector $e_j$. We will list a few more observables in Section~\ref{conserved:other_observables} below. It is worth mentioning that the fields themselves are electromagnetic observables. This is \emph{fundamentally different} from quantum mechanics where observables are represented as selfadjoint operators on the relevant Hilbert space. 

Nevertheless, many interesting observables such as the Poynting vector, total angular momentum and helicity can be written as “quantum expectation values”, \ie $\mathcal{F}$ is defined in terms of a selfadjoint operator $F = F^{\ast_W}$ with 
\begin{align}
	\mathcal{F}(\mathbf{E},\mathbf{H}) = \bscpro{Q_+ (\mathbf{E},\mathbf{H}) \,}{\, F \, Q_+ (\mathbf{E},\mathbf{H})}_W
	. 
	\label{conserved:eqn:quadratic_observable}
\end{align}
We emphasize that the similarity to quantum expectation values is purely computational, and we are not ascribing some “quantum interpretation” to these electromagnetic observables. For such observables we were able to derive ray optics equations in photonic crystals where $W$ is periodic using semiclassical techniques \cite{DeNittis_Lein:ray_optics_photonic_crystals:2014}. Fortunately, the fundamental conserved quantities for Maxwell's equations are indeed of this form (see \eg \cite[equations~(3.40')–(3.45')]{Bliokh_Bekshaev_Nori:dual_electromagnetism:2013} for a list). 

\subsection{The conservation criterion for quadratic electromagnetic observables} 
\label{conservation:criterion}
A quantity $\mathcal{F}$ is conserved if 
\begin{align}
	\mathcal{F} \bigl ( \mathbf{E}(t) , \mathbf{H}(t) \bigr ) = \mathcal{F} \bigl ( \mathbf{E}(0) , \mathbf{H}(0) \bigr ) 
	\label{conserved:eqn:definition_conserved_quantity}
\end{align}
does not depend on time and the choice of initial state. 
\begin{proposition}[Conserved quantity]
	Suppose $\mathcal{F}$ is a quadratic electromagnetic observable of the form~\eqref{conserved:eqn:quadratic_observable} associated to some bounded selfadjoint operator $F = F^{\ast_W}$ on $\Hil_+$. Then $\mathcal{F}$ is a conserved quantity for all electromagnetic fields $(\mathbf{E},\mathbf{H}) \in L^2(\R^3,\R^6)$ if and only if $\bigl [ F \, , \, \e^{- \ii t M_+} \bigr ] = 0$ holds for all times.
\end{proposition}
The assumption that $F$ be bounded is non-essential and meant to avoid a non-essential, technical discussion on domains in the proof. In the same vein, the condition $\bigl [ F , \e^{- \ii t M_+} \bigr ] = 0$ can in practice be replaced by $[ F , M_+ ] = 0$ or $[ F , \Maux ] = 0$. 
\begin{proof}
	Assume $\bigl [ F \, , \, \e^{- \ii t M_+} \bigr ] = 0$ for all $t \in \R$. Hence, we can commute the time evolution operator with $F$ and exploit the $\scpro{\, \cdot \,}{\, \cdot \,}_W$-unitarity, 
	\begin{align*}
		\mathcal{F} \bigl ( \mathbf{E}(t) , \mathbf{H}(t) \bigr ) &= \Bscpro{\e^{- \ii t M_+} \, Q_+ \bigl ( \mathbf{E}(0) , \mathbf{H}(0) \bigr ) \,}{\, F \, \e^{- \ii t M_+} \, Q_+ \bigl ( \mathbf{E}(0) , \mathbf{H}(0) \bigr )}_W
		\\
		&
		= \Bscpro{\e^{- \ii t M_+} \, Q_+ \bigl ( \mathbf{E}(0) , \mathbf{H}(0) \bigr ) \,}{\, \e^{- \ii t M_+} \, F \, Q_+ \bigl ( \mathbf{E}(0) , \mathbf{H}(0) \bigr )}_W
		\\
		&
		= \Bscpro{Q_+ \bigl ( \mathbf{E}(0) , \mathbf{H}(0) \bigr ) \,}{\, F \, Q_+ \bigl ( \mathbf{E}(0) , \mathbf{H}(0) \bigr )}_W
		= \mathcal{F} \bigl ( \mathbf{E}(0) , \mathbf{H}(0) \bigr ) 
		. 
	\end{align*}
	Conversely, suppose $\mathcal{F}$ satisfies \eqref{conserved:eqn:definition_conserved_quantity} for all fields. This translates to the condition 
	\begin{align*}
		\e^{+ \ii t M_+} \, F \, \e^{- \ii t M_+} = F 
	\end{align*}
	for all $t \in \R$, \ie $F$ commutes with the evolution group. 
\end{proof}
One could also think to consider observables which are only conserved for a certain class of states, \eg electromagnetic fields with a particular symmetries. Technically speaking, these are still of the form \eqref{conserved:eqn:definition_conserved_quantity}, although $F = \Pi \, G \, \Pi$ is now flanked on both sides by a projection $\Pi$ onto the subspace spanned by the high symmetry states. 

\subsection{Conservation of helicity in media} 
\label{conserved:helicity}
Let us illustrate our ideas with helicity so that we can contrast and compare with a recent preprint on that subject \cite{Proskurin_Ovchinnikov_Nosov_Kishine:optical_chirality_gyrotropic_media:2016}. Suppose we are given a medium described by 
$W(x) = \left ( 
\begin{smallmatrix}
	\eps(x) & \chi(x) \\
	\chi(x)^* & \mu(x) \\
\end{smallmatrix}
\right )$ 
that is dual symmetric. That means $J = \sigma_2 \otimes \id$ has to commute with $M_+ = W \, \Rot \, \big \vert_{\omega \geq 0}$. Evidently, $J$ commutes with $\Rot = - \sigma_2 \otimes \nabla^{\times}$ and we need to check whether $J$ and $W$ commute. So writing the material weights 
\begin{align*}
	W = \sum_{j = 0}^3 \sigma_j \otimes w_j 
\end{align*}
as the sum of $\sigma_0 = \id$ and the Pauli matrices, we immediately find that $[W , J] = 0$ translates to the vanishing of $w_1 = 0 = w_3$. Such material weights are of the form 
\begin{align*}
	W(x) = \id \otimes \eps(x) + \sigma_2 \otimes \chi(x) 
	= \left (
	\begin{matrix}
		\eps(x) & - \ii \, \chi(x) \\
		+ \ii \, \chi(x) & \eps(x) \\
	\end{matrix}
	\right )
\end{align*}
where $\chi(x)^* = \chi(x)$ takes values in the hermitian matrices. 

For such media, helicity, also known as the \emph{Lipkin zilch} [cite Lipkin's paper] 
\begin{align*}
	\mathcal{C} \bigl ( \mathbf{E}(t) , \mathbf{H}(t) \bigr ) &= \frac{1}{2} \int_{\R^3} \dd x \, \bigl ( \mathbf{B}(t,x) \cdot \partial_t \mathbf{D}(t,x) - \mathbf{D}(t,x) \cdot \partial_t \mathbf{B}(t,x) \bigr )
\end{align*}
is conserved (it was initially referred to as zilch because its physical interpretation was initially mysterious); note that unlike many other authors, complex conjugation is implicitly contained in the dot product $\mathbf{E}(x) \cdot \mathbf{H}(x) = \sum_{j = 1}^3 \overline{E_j(x)} \, H_j(x)$. To rewrite it in the Schrödinger formalism in the form~\eqref{conserved:eqn:quadratic_observable}, we use the constitutive relations~\eqref{Maxwell_equations:eqn:linear_nondispersive_constitive_relations} to express the auxiliary fields $(\mathbf{D},\mathbf{B})$ in terms of $(\mathbf{E},\mathbf{H})$, remember that positive and negative frequency contributions evolve according to different equations and note that this difference of fields can be written as $\ii J = \ii \sigma_2 \otimes \id$: 
\begin{align}
	\mathcal{C} \bigl ( \mathbf{E}(t) , \mathbf{H}(t) \bigr ) &= \Bscpro{W \e^{- \ii t M_+} Q_+ \bigl ( \mathbf{E}(0) , \mathbf{H}(0) \bigr ) \, }{\, J \, M_+ \, \e^{- \ii t M_+} Q_+ \bigl ( \mathbf{E}(0) , \mathbf{H}(0) \bigr )}
	\notag \\
	&= \Bscpro{\e^{- \ii t M_+} Q_+ \bigl ( \mathbf{E}(0) , \mathbf{H}(0) \bigr ) \,}{\, J \, M_+ \, \e^{- \ii t M_+} Q_+ \bigl ( \mathbf{E}(0) , \mathbf{H}(0) \bigr )}_W 
	\label{conserved:eqn:modern_form_helicity}
\end{align}
The factor of $\nicefrac{1}{2}$ has disappeared, because the negative frequency contribution, which we have omitted due to symmetry, contributes the same amount. Since $[J , M_+] = 0$ the helicity is in fact a conserved quantity, $\mathcal{C} \bigl ( \mathbf{E}(t) , \mathbf{H}(t) \bigr ) = \mathcal{C} \bigl ( \mathbf{E}(0) , \mathbf{H}(0) \bigr )$. 

Proskurin et al have found expressions very similar to \eqref{conserved:eqn:modern_form_helicity}, namely equations~(24)–(26) in \cite{Proskurin_Ovchinnikov_Nosov_Kishine:optical_chirality_gyrotropic_media:2016} for vacuum and equation~(62). Apart from the non-essential difference that they compute these electromagnetic observables from the Fourier transformed fields, they miss the important fact that the negative frequency fields evolve according with the complex conjugate weights $\overline{W} = \id \otimes \overline{\eps} - \sigma_2 \otimes \overline{\chi}$. Lastly, we note that generalizations akin to \cite[equation~(27)]{Proskurin_Ovchinnikov_Nosov_Kishine:optical_chirality_gyrotropic_media:2016} are straight-forward: higher-order time derivatives become factors of $\bigl ( - \ii M_+ \bigr )^n$. 

\subsection{Other observables} 
\label{conserved:other_observables}
Helicity is not the only observable that has seen a lot of attention up until now. A second one is, perhaps surprisingly to the uninitiated, the discussion on what the physical momentum observable in electromagnetism is. There are two separate issues: the first one is the so-called Abraham-Minkowski controversy (see \eg \cite{Pfeifer_Nieminen_et_al:momentum_electromagnetic_wave:2007,Silveirinha:Abraham_Minkowski_revisited:2017}), which we will not get into, and circles whether the momentum vector should be defined in terms of $(\mathbf{E},\mathbf{H})$ (which yields the Poynting vector) or $(\mathbf{D},\mathbf{B})$ (resulting in the Minkowski vector). We will discuss the second one, though: the Poynting vector $\mathcal{P}_x$ arises from deriving the local energy density in time so as to satisfy the local energy conservation law 
\begin{align}
	\partial_t \mathcal{E}_x \bigl ( \mathbf{E}(t),\mathbf{H}(t) \bigr ) + \nabla \cdot \mathcal{P}_x \bigl ( \mathbf{E}(t) , \mathbf{H}(t) \bigr ) = 0 
	, 
	\label{conserved:eqn:local_energy_conservation}
\end{align}
and even in the present of position-dependent weights, we nevertheless obtain the usual form of $\mathcal{P}_x(\mathbf{E},\mathbf{H})$ as defined in equation~\eqref{conserved:eqn:Poynting_vector}. 

While it is very suggestive to call the Poynting vector the momentum vector, we have a “gauge freedom” here: if all we require for a momentum density to satisfy the local energy conservation law~\eqref{conserved:eqn:local_energy_conservation}, we can replace the Poynting vector $\mathcal{P}_x$ with 
\begin{align*}
	\mathcal{P}_x'(\mathbf{E},\mathbf{H}) = \mathcal{P}_x(\mathbf{E},\mathbf{H}) + \nabla \times \mathcal{F}(\mathbf{E},\mathbf{H})
	, 
\end{align*}
which differs from the Poynting vector by a divergence-free quantity. One physically relevant question is how to split the Poynting vector 
\begin{align*}
	\mathcal{P}_x(\mathbf{E},\mathbf{H}) = \mathcal{P}_{O,x}(\mathbf{E},\mathbf{H}) + \mathcal{P}_{S,x}(\mathbf{E},\mathbf{H}) 
\end{align*}
into an orbital and a spin component, that then allows us to define orbital and spin angular momentum. In vacuum, the orbital and the spin parts for waves of frequency $\omega$ \cite[Section~3.3]{Bliokh_Bekshaev_Nori:dual_electromagnetism:2013} are given by 
\begin{align*}
	\mathcal{P}_{O,x}(\mathbf{E},\mathbf{H}) &= \Psi(x) \cdot (- \ii \nabla \Psi)(x)
	\\
	\mathcal{P}_{S,j,x}(\mathbf{E},\mathbf{H}) &= \Psi(x) \cdot (\nabla \times S_j \Psi)(x)
\end{align*}
where the spin part is defined in terms of the three spin matrices whose entries
\begin{align*}
	(S_j)_{nk} = - \ii \epsilon_{jnk}
\end{align*}
are $-\ii$ times the Ricci tensor $\epsilon_{jnk}$. Note that due to the absence of sources, $\nabla \cdot \mathbf{E} = 0$ and $\nabla \cdot \mathbf{H} = 0$, the spin contribution is divergence free. 

The Schrödinger formalism again allows for a straight-forward generalization of the work of Bliokh et al \cite[Section~3.2.2 and 3.3]{Bliokh_Bekshaev_Nori:dual_electromagnetism:2013} for the \emph{in vacuo} equations to homogeneous media. So suppose $W$ is independent of $x$ and there are no sources, \ie $\Div \, W Q_+(\mathbf{E},\mathbf{H}) = 0$ holds. Then evidently, the weights and the Maxwell operator commute with the generator of translations $P = - \ii \nabla$, \ie $[W , P] = 0$ and $[M_+ , P] = 0$. Note that $P = - \ii \nabla$ should not be interpreted as an electromagnetic “momentum operator”, although it does give rise to a conserved quantity, 
\begin{align}
	\mathcal{P}_O(\mathbf{E},\mathbf{H}) &= \bscpro{Q_+(\mathbf{E},\mathbf{H}) \,}{\, P \, Q_+(\mathbf{E},\mathbf{H})}_W
	\label{conserved:eqn:orbital_momentum}
	\\
	&= \int_{\R^3} \dd x \, \mathcal{P}_{O,x}(\mathbf{E},\mathbf{H})
	= \int_{\R^3} \dd x \, \bigl ( Q_+(\mathbf{E},\mathbf{H}) \bigr )(x) \cdot W \bigl ( - \ii \nabla Q_+(\mathbf{E},\mathbf{H}) \bigr )(x)
	\notag 
	.  
\end{align}
In spirit this expression is a generalization of \cite[equation~(3.32)]{Bliokh_Bekshaev_Nori:dual_electromagnetism:2013} to homogeneous media. And while mathematics tells us $\mathcal{P}_O \bigl ( \mathbf{E}(t) , \mathbf{H}(t) \bigr ) = \mathcal{P}_O \bigl ( \mathbf{E}(0) , \mathbf{H}(0) \bigr )$ is conserved, it is not yet clear whether it is legitimate to interpret this as the orbital contribution to momentum. That is because we do not know whether the integrand of \eqref{conserved:eqn:orbital_momentum}, $\mathcal{P}_{O,x}$, satisfies local energy conservation~\eqref{conserved:eqn:local_energy_conservation}. Equivalently, we would have to verify whether the candidate for the spin density $\mathcal{P}_x - \mathcal{P}_{O,x}$ is divergence free. Experiments might shed light on the issue: depending on the experimental setup, it is possible to measure the \emph{orbital} contribution $\mathcal{P}_{x,O}$ instead of the momentum density $\mathcal{P}_x$, because the radiation force for a particle with an electric \emph{and} a magnetic dipole is proportional to $\mathcal{P}_{x,O}$ rather than $\mathcal{P}_x$ (\cf discussion in \cite[Section~5]{Bliokh_Bekshaev_Nori:dual_electromagnetism:2013} and references therein). We shall not tug further on this thread here and intend to revisit this question on the future. 

There are a number of other quadratic observables of the same mold such as \emph{spin} and \emph{helicity density} (see \eg \cite[Section~3.3]{Bliokh_Bekshaev_Nori:dual_electromagnetism:2013} and \cite[Section~II]{Proskurin_Ovchinnikov_Nosov_Kishine:optical_chirality_gyrotropic_media:2016}), 
\begin{align*}
	\mathcal{S}_{x,j}(\mathbf{E},\mathbf{H}) &= \Psi(x) \cdot \bigl ( \id \otimes S_j \bigr ) \Psi(x) 
	, 
	\\
	\mathcal{H}_x(\mathbf{E},\mathbf{H}) &= \sum_{j = 1}^3 \Psi(x) \cdot \Bigl ( \tfrac{P_j}{\abs{P}} \, (\id \otimes S_j) \Bigr ) \Psi(x) 
	. 
\end{align*}
%

\subsection{Going beyond continuous symmetries} 
\label{conserved:beyond_continuous_symmetries}
Not all symmetries of physical systems are continuous, \eg periodic systems have discrete symmetries, and these are beyond the reach of Noether's Theorem. This is a second advantage the Schrödinger formalism: Assume $W$ describes a photonic crystal where the weights are periodic. Then the above is no longer a conserved quantity, but it is compatible with the Bloch-Floquet transform in that we can consider the contributions to the orbital current density for fixed Bloch momentum. 

The orbital momentum could then be studied in the ray optics approximation scheme of \cite{DeNittis_Lein:ray_optics_photonic_crystals:2014}. In the parlance of \cite[Definition~3.4]{DeNittis_Lein:ray_optics_photonic_crystals:2014} the orbital momentum $\mathcal{P}_{O,x}$ is a \emph{scalar} observable, meaning it does not mix different components of the electromagnetic field. For scalar observables like the orbital momentum (and unlike for the Poynting vector!), the Berry curvature explicitly enters the ray optics equations (\cf \cite[Theorem~3.7~(i)]{DeNittis_Lein:ray_optics_photonic_crystals:2014}). Consequently, it seems as if the \emph{orbital} momentum which behaves analogously to the current in quantum condensed matter systems, a link which may prove useful when studying the Quantum Hall Effect of Light \cite{Raghu_Haldane:quantum_Hall_effect_photonic_crystals:2008}. 
\section{Schrödinger formalism for acoustic waves} 
\label{other_waves}
The Schrödinger formalism developed in the preceding sections for classical electromagnetism applies also to many other types of classical waves, because they share the same essential characteristics: they are (i)~linear, (ii)~first-order in time and (iii)~the classical fields are real-valued. While we are not the first to point this out (see \eg \cite[Appendix]{Wilcox:scattering_theory_classical_physics:1966} or \cite{Reed_Simon:scattering_theory_wave_equations:1977}), these works do not take the reality of waves into account. Other works (\eg \cite{Dodin:geometric_view_noneikonal_waves:2014,Ruiz:geometric_theory_waves_and_plasma_physics:2017}) lack mathematical rigor and impose additional, unnecessary restrictions on the weights. We first distill down our ideas in the abstract and then apply them to linearized magnetohydrodynamics, Alvfén waves and transverse acoustic waves.

\subsection{The abstract mathematical formalism} 
\label{other_waves:abstract}
For these examples, we need to slightly widen the scope as compared to Section~\ref{Schroedinger:abstract:classical_waves} and generalize the notion of 
\begin{definition}[Maxwell-type operator]\label{other_waves:defn:Maxwell_type_operator}
	A (generalized) Maxwell-type operator 
	\begin{align}
		M := W_L \, D \, W_R
		\label{other_waves:eqn:generalized_acoustic_operator}
	\end{align}
	is a linear operator with product structure on a complex Hilbert space $\Hil$ that has the following properties: 
	\begin{enumerate}[(a)]
		\item $D$ is a (possibly unbounded) selfadjoint operator on $\Hil$ with domain $\mathcal{D}_0$. 
		\item $W_L , W_R \in \mathcal{B}(\Hil)$ are bounded, selfadjoint, commuting operators with bounded inverses, \ie $[W_L , W_R] = 0$ and $W_L^{-1} , W_R^{-1} \in \mathcal{B}(\Hil)$. 
		\item The product $W_R \, W_L^{-1}$ is selfadjoint and bounded away from $0$, \ie there exists a constant $c > 0$ so that $W_R \, W_L^{-1} > c \, \id_{\Hil}$. 
		\item $M$ is endowed with the domain $\mathcal{D}_R := W_R^{-1} \, \mathcal{D}_0$. 
		\item $M$ anticommutes with complex conjugation $(C \Psi)(x) := \overline{\Psi(x)}$, \ie $C \, \mathcal{D}_R = \mathcal{D}_R$ is left invariant and $C \, M \, C = - M$ holds. 
	\end{enumerate}
\end{definition}
We will show that operators of this type fit exactly into the same framework we have developed previously. Note that several authors prefer to study the (electromagnetic) Maxwell operator $M = W^{- \nicefrac{1}{2}} \, \Rot \, W^{- \nicefrac{1}{2}}$ (see \eg \cite[Section~III]{Tip:linear_absorptive_dielectrics:1998} or \cite[Section~II.A.]{Gangaraj_Silveirinha_Hanson:Berry_connection_Maxwell:2017}), because $W_L = W_R$ means the operator $M$ is selfadjoint (hermitian) with respect to the usual, \emph{un}weighted scalar product (see Section~\ref{other_waves:abstract:change_of_representation} below).

\subsubsection{Endowing $\Hil$ with a suitably weighted scalar product and selfadjointness of $M$} 
\label{other_waves:abstract:selfadjointness}
The commutativity of $W_L$ and $W_R$ also implies $\bigl [ W_L^{\pm 1} , W_R^{\pm 1} \bigr ] = 0$ for all of the other three other sign combinations. Therefore, we can exchange the order of the two factors in the weight operator 
\begin{align*}
	W := W_L^{-1} \, W_R 
	= W_R \, W_L^{-1} 
\end{align*}
which enters the definition of a \emph{weighted} scalar product
\begin{align}
	\bscpro{\Phi}{\Psi}_W := \bscpro{\Phi}{W \, \Psi}_{\Hil} 
	. 
	\label{other_waves:eqn:weighted_scalar_product}
\end{align}
Thanks to the conditions imposed on $W_L$ and $W_R$ in Definition~\ref{other_waves:defn:Maxwell_type_operator}~(b)–(c), this sesquilinear map is strictly positive definite, and therefore indeed defines a scalar product. To distinguish $\Hil$ with its standard scalar product, we will write $\Hil_W$ for the Banach space $\Hil$ endowed with $\scpro{\, \cdot \,}{\, \cdot \,}_W$. 

Since the norm $\snorm{\Psi}_W := \scpro{\Psi}{\Psi}_W^{\nicefrac{1}{2}}$ of $\Hil_W$ is equivalent to the original norm $\Hil$ came with, 
\begin{align*}
	\norm{W^{- \nicefrac{1}{2}}}_{\mathcal{B}(\Hil)} \, \snorm{\psi}_W \leq \snorm{\psi} \leq \norm{W^{+ \nicefrac{1}{2}}}_{\mathcal{B}(\Hil)} \, \snorm{\psi}_W
	,
\end{align*}
$\Hil$ and $\Hil_W$ agree as Banach spaces, and we can think of both as the same vector spaces with differing notions of orthogonality. Consequently, we may consider operators to be acting on either space and properties such as boundedness or closedness transfer immediately. The advantage of changing scalar product comes from the fact that $M$ is in fact selfadjoint with respect to $\scpro{\, \cdot \,}{\, \cdot \,}_W$ (but \emph{not} the original scalar product $\scpro{\, \cdot \,}{\, \cdot \,}_{\Hil}$). 
\begin{proposition}\label{other_waves:prop:selfadjointness_Maxwell_type_operator}
	The generalized acoustic operator $M = W_L \, D \, W_R$ turns out to be selfadjoint on $\Hil_W$ with dense domain $\mathcal{D}_R := W_R^{-1} \, \mathcal{D}_0$.
\end{proposition}
The symmetry of $M$ follows from the same arguments as in equation~\eqref{Schroedinger:eqn:symmetry_Maxwell_operator}, and the interested reader may find the remaining technical details for proving selfadjointness in Appendix~\ref{appendix:generalized_Maxwell_type_operators}. 

\subsubsection{Change of representation: bringing $M$ to the form~\eqref{Schroedinger:eqn:Maxwell_type_operator}} 
\label{other_waves:abstract:change_of_representation}
Borrowing the terminology quantum mechanics, a change of representation can bring $M = W_L \, D \, W_R$ to the form $\widetilde{M} = W \, D$, the operator studied in Section~\ref{Schroedinger:abstract:classical_waves}. That is because 
\begin{align*}
	W_R : L^2_W(\R^3,\C^6) \longrightarrow L^2_{\widetilde{W}}(\R^3,\C^6) 
	,
\end{align*}
seen as a map between two suitably weighted Hilbert spaces, is in fact unitary: Setting the other weight to be $\widetilde{W} := W_L^{-1} \, W_R^{-1}$, a quick computation yields 
\begin{align*}
	\bscpro{W_R \Phi}{W_R \Psi}_{\widetilde{W}} &= \bscpro{W_R \Phi}{W_R \, W_L^{-1} \, W_R^{-1} \, \Psi}
	= \bscpro{\Phi}{W_L^{-1} \, W_R \, \Psi}
	= \bscpro{\Phi}{\Psi}_W 
	. 
\end{align*}
In this new representation, the Maxwell-type operator takes the form considered previously, 
\begin{align*}
	\widetilde{M} := W_R \, M \, W_R^{-1} 
	= W_L \, W_R \, D 
	, 
\end{align*}
and endowed with the domain $\mathcal{D}_0$ of $D$, it inherits the selfadjointness and the even particle-hole-type symmetry $C_R := W_R \, C \, W_R^{-1}$ from $M$. Consequently, under these hypotheses $M = W_L \, D \, W_R$ is unitarily equivalent to $\widetilde{M} = W \, D$, the type of operator considered in Section~\ref{Schroedinger:abstract:classical_waves}. 

\subsubsection{Reduction to $\omega \geq 0$} 
\label{other_waves:abstract:reduction_omega_geq_0}
Condition~(e) in Definition~\ref{other_waves:defn:Maxwell_type_operator} ensures that real solutions to the associated Schrödinger equation
\begin{align*}
	\ii \partial_t \Psi(t) = M \Psi(t) 
	, 
	&&
	\Psi(0) = \Phi \in \Hil
	,
\end{align*}
supports real solutions. Put another way, $\Hil$ contains a lot of unphysical, complex waves and not just real fields $u \in \Hil_{\R} := \Re \Hil$. Following the next step in the scheme developed for classical electromagnetism, Section~\ref{Schroedinger:formalism:reality}, we exploit that we can represent real fields as the real part of complex waves composed solely of non-negative frequencies: the spectral projection $Q := 1_{[0,\infty)}(M)$ implements this restriction, and show that $2 \Re$ is a left-inverse of $Q$ restricted to $\Hil_{\R}$, 
\begin{align*}
	2 \Re \, Q \, \big \vert_{\Hil_{\R}} = \id_{\Hil_{\R}} 
	. 
\end{align*}
To verify this, we can modify the arguments which prove Proposition~\ref{Schroedinger:prop:identification_real_complex_vector_spaces} in a straightforward fashion (Lemma~\ref{appendix:generalized_Maxwell_type_operators:lem:1_to_1_correspondence}). 

Therefore, the whole physics of real waves is contained in the restriction to $\omega \geq 0$, 
\begin{align*}
	M_+ := Q \, M \, Q \, \big \vert_{\Hil_+}
	, 
\end{align*}
where this operator now acts on the complex Hilbert space of non-negative frequency fields, 
\begin{align*}
	\Hil_+ := Q [\Hil] = \ran Q 
	, 
\end{align*}
that comes with the weighed scalar product~\eqref{other_waves:eqn:weighted_scalar_product}. 

\subsection{Linearized magnetohydrodynamics} 
\label{other_waves:MHD}
The first example is linearized magnetohydrodynamics (MHD). The set of equations~\eqref{other_waves:eqn:ideal_MHD} below describes the interaction of magnetic fields with electrically conducting fluids such as plasmas or liquid metals, and can be deduced phenomenologically by coupling the hydrodynamic equations of ordinary fluids to a magnetic field via Ampère's Law. The main approximation is to neglect the displacement current. In the standard non-relativistic form the MHD equations consist of the basic conservation laws of mass, momentum and energy together with the induction equation for the magnetic field.

\subsubsection{Ideal linearized MHD equations} 
\label{other_waves:MHD:idealized}
Consider a stationary plasma permeated by a stationary magnetic field. The \emph{equilibrium configuration} of the plasma is described by the \emph{density} $\rho_0$, the \emph{pressure} $\wp_0$, the \emph{background magnetic field} $\textbf{B}_0 = \bigl ( B_{0,1} , B_{0,2} , B_{0,3} \bigr )$ and the \emph{gravitational field} $\textbf{g} = \bigl ( g_1 , g_2 , g_3 \bigr )$. When the system is perturbed only slightly from its equilibrium state, the state variable can be written in the form
\begin{align*}
	\rho(x,t) &= \rho_0(x) + \rho_1(x,t)
	,
	\\
	\wp(x,t) &= \wp_0(x) + \wp_1(x,t)
	,
	\\
	\textbf{B}(x,t) &= \textbf{B}_0(x) + \textbf{B}_1(x,t)
	,
\end{align*}
and in this ansatz the quantities $\rho_1$, $\wp_1$ and $\textbf{B}_1$ can be considered “small” if compared with the unperturbed part. Since the stationary condition implies $\textbf{v}_0 = 0$ for the equilibrium \emph{velocity} of the plasma we will use the symbol $\textbf{v}(x,t) = \bigl ( v_1(x,t),v_2(x,t),v_3(x,t) \bigr )$ for the “small” velocity induced by the perturbation. According to \cite[Chapter~5, Section~2]{Lifshits:magnetohydrodynamics:1989} the \emph{linearized ideal} MHD equations are: 

\begin{subequations}\label{other_waves:eqn:ideal_MHD}
	\begin{align}
		\frac{\dd \rho_1}{\dd t} &= - \nabla \cdot \bigl ( {\rho}_0 \, \textbf{v} \bigr )
		,
		\label{other_waves:eqn:ideal_MHD:1}
		\\
		\rho_0 \, \frac{\dd \textbf{v}}{\dd t} &= - \nabla \wp_1 - \textbf{g} \, \rho_1 + \bigl ( \textbf{J}_0 \times \textbf{B}_1 + \textbf{J}_1 \times \textbf{B}_0 \bigr )
		,
		\label{other_waves:eqn:ideal_MHD:2}
		\\
		\frac{\dd \textbf{B}_1}{\dd t} &= \nabla \times \bigl ( \textbf{v} \times \textbf{B}_0 \bigr )
		,
		\label{other_waves:eqn:ideal_MHD:3}
		\\
		\wp_1 &= \Gamma(\rho_1,\textbf{v})
		. 
		\label{other_waves:eqn:ideal_MHD:4}
	\end{align}
\end{subequations}
The \emph{continuity equation} \eqref{other_waves:eqn:ideal_MHD:1} expresses mass conservation. Equation \eqref{other_waves:eqn:ideal_MHD:2}, usually called \emph{Euler equation}, is the equation of motion of an element of the fluid. The vectors $\textbf{J}_0$ and $\textbf{J}_1$ are the \emph{electric current densities} associated to the magnetic field $\textbf{B}_0$ and $\textbf{B}_1$ as computed from \emph{Ampère's Law}
\begin{align*}
	\textbf{J}_k = \frac{1}{\mu} \, \nabla \times \textbf{B}_k
	,
	&&
	k = 0 , 1 ,
\end{align*}
with $\mu$ being the \emph{magnetic permeability}. Equation~\eqref{other_waves:eqn:ideal_MHD:2} conceals a second equilibrium state constraint, $\nabla \wp_0 + \textbf{g} \, \rho_0 = \textbf{J}_0 \times \textbf{B}_0$. The \emph{induction equation}~\eqref{other_waves:eqn:ideal_MHD:3} is derived from the Maxwell's equations. The last equation \eqref{other_waves:eqn:ideal_MHD:4}, usually called \emph{energy equation} or \emph{adiabatic state equation}, establishes a linear functional constraint which may be derived from the thermodynamic equations for the plasma and allows us to remove the variable $\wp_1$ from equation~\eqref{other_waves:eqn:ideal_MHD:1}–\eqref{other_waves:eqn:ideal_MHD:3}. For example, if we impose the local \emph{isothermal condition} $\wp \, \rho^{-\gamma} = \mathrm{const.}$ for a homogeneous medium, in the linearized regime equation~\eqref{other_waves:eqn:ideal_MHD:4} assumes the form
\begin{align*}
	\wp_1 = \gamma \left ( \frac{\wp_0}{\rho_0} \right ) \, \rho_1 
	= \gamma \, \nu_{s}^2 \, \rho_1
\end{align*}
where the quantity $\nu_{s} := \sqrt{\nicefrac{\wp_0}{\rho_0}}$ is the local \emph{sound velocity}. 

The system of equations~\eqref{other_waves:eqn:ideal_MHD} provides the time evolution of the “small” perturbations $\rho_1$, $\textbf{v}$ and $\textbf{B}_1$ in terms of the stationary quantities $\rho_0$, $\wp_0$ and $\textbf{B}_0$. The linear approximation allows to rewrite this equation in a more compact form, 
\begin{align}
	\ii \frac{\dd}{\dd t} 
	\left (
	\begin{matrix}
		\rho_1 \\
		\mathbf{v} \\
		\mathbf{B}_1 \\
	\end{matrix}
	\right ) = M(\rho_0 , \wp_0 , \mathbf{B}_0) \, 
	\left (
	\begin{matrix}
		\rho_1 \\
		\mathbf{v} \\
		\mathbf{B}_1 \\
	\end{matrix}
	\right )
	,
	\label{other_waves:eqn:linearized_MHD}
\end{align}
where $M(\rho_0 , \wp_0 , \textbf{B}_0)$ is a differential operator of order 1 which depends from the stationary quantities $\rho_0$, $\wp_0$ and $\textbf{B}_0$. We will give the explicit form of $M(\rho_0 , \wp_0 , \textbf{B}_0)$ for the special case of Alfvén waves below. If one interprets $\Psi = (\rho_1,\textbf{v},\textbf{B}_1)$ as a wave function in $L^2(\R^3,\C^7)$, then the above equation can be written in the evocative form
\begin{align*}
	\ii \, \frac{\dd}{\dd t} \Psi = M \Psi
\end{align*}
where $M = M(\rho_0 , \wp_0 , \textbf{B}_0)$ plays the role, at least formally, of a Schrödinger operator acting in $L^2(\R^3,\C^7)$. Of course, the physical interpretations of the evolved quantities $\rho_1$, $\wp_1$ and $\textbf{B}_1$ requires us to introduce of the reality constraint in the Schrödinger formalism for electromagnetism.

\subsubsection{Alfvén waves} 
\label{other_waves:MHD:alfven_waves}
One situation which affords a lot of simplifications when we specialize the linearized MHD equations to a \emph{stratified} stationary plasma permeated by a stationary magnetic field (\eg the atmosphere \cite{Axelsson:three_wave_coupling_MHD_plasma:1998}). This is the typical setting for the production of the Alfvén waves. Mathematically, the adjective stratified translates to the assumption that all the physical quantities only depend on $x_1$. In this simplified 1d context the continuity equation \eqref{other_waves:eqn:ideal_MHD:1} reads
\begin{align}
	\frac{\dd \rho_1}{\dd t} = - \frac{\partial}{\partial x_1} \bigl ( \rho_0 \, v_1 \bigr )
	,
	\label{other_waves:eqn:continuity_equation_x}
\end{align}
the Euler equation \eqref{other_waves:eqn:ideal_MHD:2} can be rewritten componentwise as 
\begin{subequations}\label{other_waves:eqn:Euler}
	\begin{align}
		\rho_0 \, \frac{\dd v_1}{\dd t} &= -\frac{\partial \wp_1}{\partial x_1} - g \, \rho_1 - \frac{1}{\mu} \, \frac{\partial}{\partial x_1} \bigl ( B_{0,2} \, B_{1,2} + B_{0,3} \, B_{1,3} \bigr )
		,
		\label{other_waves:eqn:Euler:x}
		\\
		\rho_0 \, \frac{\dd v_2}{\dd t} &= \frac{1}{\mu} \, \left ( B_{1,1} \, \frac{\partial B_{0,2}}{\partial x_1} + B_{0,1} \, \frac{\partial B_{1,2}}{\partial x_1} \right )
		,
		\label{other_waves:eqn:Euler:y}
		\\
		\rho_0 \, \frac{\dd v_3}{\dd t} &= \frac{1}{\mu} \, \left( B_{1,1} \, \frac{\partial B_{0,3}}{\partial x_1} + B_{0,1} \, \frac{\partial B_{1,3}}{\partial x_1} \right )
		,
		\label{other_waves:eqn:Euler:z}
	\end{align}
\end{subequations}
and the induction equation takes the form
\begin{subequations}\label{other_waves:eqn:induction}
	\begin{align}
		\frac{\dd B_{1,1}}{\dd t} &= 0
		,
		\label{other_waves:eqn:induction:x}
		\\
		\frac{\dd B_{1,2}}{\dd t} &= \frac{\partial}{\partial x_1} \bigl ( v_2 \, B_{0,1} - v_1 \, B_{0,2} \bigr )
		, 
		\label{other_waves:eqn:induction:y}
		\\
		\frac{\dd B_{1,3}}{\dd t} &= \frac{\partial}{\partial x_1} \bigl ( v_3 \, B_{0,1} - v_1 \, B_{0,3} \bigr )
		.
		\label{other_waves:eqn:induction:z}
	\end{align}
\end{subequations}
Equation~\eqref{other_waves:eqn:induction:x} says that the component $B_{1,1}$ is constant in time. In particular one can assume that the “small” magnetic perturbation starts (and stays) parallel to the stratification, namely $B_{1,1} = 0$. And this results in further simplifications the equations~\eqref{other_waves:eqn:Euler:y} and \eqref{other_waves:eqn:Euler:z}. Under the further assumption $B_{0,2} = 0$, which forces the equilibrium magnetic field to lies in the $x_1 x_3$-plane, one can see that \eqref{other_waves:eqn:Euler:y} and \eqref{other_waves:eqn:induction:y} provides a system of dynamical equation for the pair $(v_2 , B_{1,2})$ that is decoupled from the other equations and can be written in the following form
\begin{align*}
	\ii \frac{\dd}{\dd t} \left (
	\begin{matrix}
		v_2 \\
		B_{1,2} \\
	\end{matrix}
	\right ) = 
	\left (
	\begin{matrix}
		\frac{B_{0,1}}{\rho_0 \, \mu} & 0 \\
		0 & 1 \\
	\end{matrix}
	\right )
	\left (
	\begin{matrix}
		0 & \ii \partial_{x_1} \\
		\ii \partial_{x_1} & 0 \\
	\end{matrix}
	\right ) \left (
	\begin{matrix}
		B_{0,1} & 0\\
		0 & 1 \\
	\end{matrix}
	\right )
	\, \left (
	\begin{matrix}
		v_2 \\
		B_{1,2} \\
	\end{matrix}
	\right )
	. 
\end{align*}
The above differential equation for the “wavefunction” $\Psi = (v_2 , B_{1,2})$ can be rewritten in the Schrödinger-type formalism in such a way the time evolution for the Alfvén waves assumes the form 
\begin{align*}
	\ii \frac{\dd}{\dd t} \Psi = M_{\mathrm{Alf}} \Psi
\end{align*}
with the \emph{Alfvén Hamiltonian} $M_{\mathrm{Alf}} = W_L \, D \, W_R$. 

\subsection{The linear acoustic equation} 
\label{other_waves:acoustic}
The behavior of acoustic waves is modeled a set of equations 
\begin{subequations}\label{other_waves:eqn:no_edyn_MHD}
	\begin{align}
		\frac{\dd \rho_1}{\dd t} &= - \nabla \cdot \bigl ( \rho_0 \textbf{v} \bigr )
		,
		\label{other_waves:eqn:no_edyn_MHD:1}
		\\
		\rho_0 \, \frac{\dd \textbf{v}}{\dd t} &= - \nabla \wp_1
		,
		\label{other_waves:eqn:no_edyn_MHD:2}
		\\
		\wp_1 &= \gamma \, \nu_{s}^2 \, \rho_1
		\label{other_waves:eqn:no_edyn_MHD:4}
	\end{align}
\end{subequations}
which are very similar to MHD. In fact, we may obtain these equations by disregarding the electromagnetic degrees of freedom and gravity in equations~\eqref{other_waves:eqn:ideal_MHD} and pick the (linearized) isothermal condition in the adiabatic state equation~\eqref{other_waves:eqn:no_edyn_MHD:4}. The above system can be rewritten in Schrödinger form as
\begin{align*}
	\ii \frac{\dd}{\dd t} 
	\left (
	\begin{matrix}
		\rho_1 \\
		\mathbf{v} \\
	\end{matrix}
	\right ) 
	&= 
	\left (
	\begin{matrix}
		1 & 0 \\
		0 & \rho_0^{-1} \, \id_{\R^3} \\
	\end{matrix}
	\right ) \left (
	\begin{matrix}
		0 & - \ii \nabla^T \\
		- \ii \nabla & 0_{\R^3} \\
	\end{matrix}
	\right )
	\left (
	\begin{matrix}
		\gamma \, \nu_{s}^2 & 0 \\
		0 & \rho_0 \, \id_{\R^3} \\
	\end{matrix}
	\right )
	\left (
	\begin{matrix}
		\rho_1 \\
		\mathbf{v} \\
	\end{matrix}
	\right ) 
	,
\end{align*}
where $\Psi = (\rho_1 , \mathbf{v})$ is the “wavefunction” that enters 
\begin{align*}
	\ii \frac{\dd}{\dd t}\Psi = M_{\mathrm{ac}} \Psi
\end{align*}
and the \emph{acoustic operator} $M_{\mathrm{ac}} = W_L \, D \, W_R$ again has the characteristic product structure. 
  
Mathematically, we can widen the scope a little and generalize the above operator $D$ to an arbitrary dimension $d$. The free acoustic operator 
\begin{align*}
	D := 
	\left (
	\begin{matrix}
		0 & \ii \nabla^T \\
		\ii \nabla & 0_{\C^d} \\
	\end{matrix}
	\right )
\end{align*}
turns out to be essentially selfadjoint if endowed with the domain $\mathcal{D}_0 := H^1(\R^d,\C^{d+1})$. The 1d incarnation entered the “Hamilton operator” for van Alfvén waves and for $d = 3$ it founds its way into the acoustic operator. 

Evidently, $D$ anticommutes with complex conjugation, $C \, D \, C = - D$, so it possess the fundamental symmetry which characterizes “Hamilton operators” for classical waves. As this is a differential operator, it is “diagonalized” with the help of the Fourier transform, and we obtain a family of hermitian $(d+1) \times (d+1)$ matrices 
\begin{align*}
	D(k) = 
	- \left (
	\begin{matrix}
		0 & k^T \\
		k & 0_{\C^d} \\
	\end{matrix}
	\right )
	.
\end{align*}
The spectrum of $D(k)$ is in this case just the set of eigenvalues, namely $0$ with multiplicity $d - 1$ and $\pm \abs{k}$ with multiplicity $1$. The spectrum of $D$ is the union of the eigenvalues of the $D(k)$, 
\begin{align*}
	\sigma(D) = \bigcup_{k \in \R^d} \sigma \bigl ( D(k) \bigr )
	= \sigma_{\mathrm{ac}}(D) \cup \sigma_{\mathrm{pp}}(D)
	= \R \cup \{ 0 \}
	,
\end{align*}
and we see that $D$ comes with an infinite-dimensional kernel. For example, for $j = 2 , \ldots , d$ functions of the form $\Psi = ( 0 , \partial_{x_j} f , 0 , \ldots , 0 , - \partial_{x_1} f , 0 , \ldots )$ are in the kernel, where $- \partial_{x_1} f$ is in the $(j+1)$th position. This is very reminiscent of the free Maxwell operator $\Rot$ whose kernel consists of gradient fields. 
\appendix

\section{Rigorous proof of the correspondence between real and complex fields} 
\label{appendix:1_to_1_correspondence}
For notational clarity, let us restore the index $\pm$, \ie the auxiliary Maxwell operators are denoted with $\Maux_{\pm} = W_{\pm}^{-1} \, \Rot$ where $W_+ = W = \overline{W_-}$ and give rise to the spectral projections  
\begin{align*}
	P_{\pm} &= 1_{(0,\infty)} \bigl ( \pm \Maux_{\pm} \bigr ) 
	, 
	\\
	P_{\pm,0} &= 1_{\{ 0 \}} \bigl ( \Maux_{\pm} \bigr ) 
	, 
	\\
	Q_{\pm} &= P_{\pm} + \tfrac{1}{2} \, P_{\pm,0} 
	. 
\end{align*}
These operators are naturally defined on the weighted $L^2_{W_{\pm}}(\R^3,\C^6)$. The two auxiliary Maxwell operators are related by complex conjugation, $C \, \Maux_{\pm} \, C = - \Maux_{\mp}$, and thus, this symmetry $C$ relates the projections and the maps $Q_{\pm}$, 
\begin{align*}
	C \, P_{\pm} \, C &= P_{\mp} 
	, 
	\\
	C \, P_{0,\pm} \, C &= P_{0,\mp} 
	, 
	\\
	C \, Q_{\pm} \, C &= Q_{\mp} 
	. 
\end{align*}
The relevant Hilbert spaces of physical states are $\Hil_{\pm} = \ran Q_{\pm}$, which inherit their scalar products $\scpro{\, \cdot \,}{\, \cdot \,}_{W_{\pm}}$ from $L^2_{W_{\pm}}(\R^3,\C^6)$. 

Before we prove the main result, Theorem~\ref{Schroedinger:thm:equivalence_frameworks}, which guarantees the equivalence of the two descriptions of Maxwell's equations, we first establish that we can systematically identify the \emph{real} Hilbert space $L^2(\R^3,\R^6)$ with the \emph{real} subspace 
\begin{align*}
	\Hil_{\R} := \Bigl \{ \bigl ( \Psi_+ , \overline{\Psi_+} \bigr ) \in \Hil_+ \oplus \Hil_- \; \; \big \vert \; \; \Psi_+ \in \Hil_+ \Bigr \} 
	\subseteq \Hil_+ \oplus \Hil_- 
\end{align*}
of the \emph{complex} Hilbert space $\Hil_+ \oplus \Hil_-$. The latter can be seen as the “real part” of $\Hil_+ \oplus \Hil_-$ with respect to the abstract real part operator $2 \Re_K = \id + K$, namely $\Hil_{\R} = 2 \Re_K \bigl [ \Hil_+ \oplus \Hil_- \bigr ]$, where 
\begin{align*}
	K = (\sigma_1 \otimes \id) \, C = \left (
	\begin{matrix}
		0 & C \\
		C & 0 \\
	\end{matrix}
	\right )
\end{align*}
takes the place of complex conjugation. Alternatively, $L^2(\R^3,\R^6)$ can be included into $\Hil_+$, thereby embedding a \emph{real} Hilbert space into a \emph{complex} Hilbert space. 
\begin{lemma}\label{appendix:1_to_1_correspondence:lem:1_to_1_correspondence}
	Suppose we are in the setting of Proposition~\ref{Schroedinger:prop:identification_real_complex_vector_spaces}. Then the following statements hold: 
	\begin{enumerate}[(1)]
		\item The map $Q_{\pm} : L^2(\R^3,\R^6) \longrightarrow \Hil_{\pm}$ is injective. 
		\item The map $2 \Re \, Q_{\pm} : L^2(\R^3,\R^6) \longrightarrow L^2(\R^3,\R^6)$ is injective. 
		\item The map $Q_+ \oplus Q_- : L^2(\R^3,\R^6) \longrightarrow \Hil_{\R}$ is injective. 
		\item The map $\Hil_{\R} \longrightarrow L^2(\R^3,\R^6)$, $(\Psi_+,\Psi_-) \mapsto \Psi_+ + \Psi_-$, is injective. 
	\end{enumerate}
\end{lemma}
Note that in case $W_+ \neq W_-$ we do not know whether $2 \Re \, Q_+ \, \Re = \Re$ holds true. That is because in the non-gyrotropic case $W_+ = W_-$ we can write 
\begin{align*}
	\id_{L^2(\R^3,\R^6)} &= \Bigl ( 1_{(-\infty,0)} \bigl ( \Maux_- \bigr ) + \tfrac{1}{2} \, 1_{\{ 0 \}} \bigl ( \Maux_- \bigr ) + \tfrac{1}{2} \, 1_{\{ 0 \}} \bigl ( \Maux_+ \bigr ) + 1_{(0,\infty)} \bigl ( \Maux_+ \bigr ) \Bigr ) \Big \vert_{L^2(\R^3,\R^6)} 
	\\
	&= \bigl ( Q_+ + Q_- \bigr ) \big \vert_{L^2(\R^3,\R^6)}
\end{align*}
as the sum of orthogonal projections of the \emph{same} operator $\Maux_+ = \Maux_-$ that acts on a \emph{single} Hilbert space $L^2_{W_+}(\R^3,\C^6) = L^2_{W_-}(\R^3,\C^6)$. Points~(1) and (2) of this Lemma can be rephrased in terms of the problem at hand: 
\begin{corollary}[Representing real fields as complex $\omega \geq 0$ wave]
	Suppose we are in the setting of Proposition~\ref{Schroedinger:prop:identification_real_complex_vector_spaces}. Then there are two (potentially distinct) one-to-one correspondences between real (physical) fields and complex $\omega \geq 0$ waves: 
	\begin{subequations}
		\begin{align}
			L^2(\R^3,\R^6) \ni (\mathbf{E},\mathbf{H}) &\mapsto Q_+ (\mathbf{E},\mathbf{H}) \in Q_+ \bigl [ L^2(\R^3,\R^6) \bigr ] 
			\\
			Q_+ \bigl [ L^2(\R^3,\R^6) \bigr ] \ni \Psi_+ &\mapsto 2 \Re \Psi_+ \in L^2(\R^3,\R^6)
		\end{align}
	\end{subequations}
\end{corollary}
\begin{proof}[Lemma~\ref{appendix:1_to_1_correspondence:lem:1_to_1_correspondence}]
	\begin{enumerate}[(1)]
		\item First of all, the fact that the material weights $c \, \id \leq W \leq C \, \id$ are bounded away from $0$ and $\infty$, implies that the ordinary $L^2(\R^3,\C^6)$ conincides with $L^2_{W_{\pm}}(\R^3,\C^6)$ as Banach spaces. Therefore, (un)bounded operators on $L^2_{W_{\pm}}(\R^3,\C^6)$ can also be considered as (un)bounded operators on $L^2(\R^3,\C^6)$ — and vice versa. The $\Maux_{\pm}$ are selfadjoint on the appropriately weighted Hilbert spaces $L^2_{W_{\pm}}(\R^3,\C^6)$. 
		
		For proving injectivity, namely that $Q_{\pm} (\mathbf{E},\mathbf{H}) = 0$ implies $(\mathbf{E},\mathbf{H}) = 0$, we exploit that any $(\mathbf{E},\mathbf{H}) \in L^2(\R^3,\R^6)$ can be seen as an element of $L^2_{W_{\pm}}(\R^3,\C^6)$, and therefore we can use two resolutions of the identity, 
		\begin{align*}
			(\mathbf{E},\mathbf{H}) &= P_{\pm} (\mathbf{E},\mathbf{H}) + P_{\pm,0} (\mathbf{E},\mathbf{H}) + 1_{(-\infty,0)} \bigl ( \pm \Maux_{\pm} \bigr ) (\mathbf{E},\mathbf{H}) 
			. 
		\end{align*}
		Suppose now that $Q_{\pm} (\mathbf{E},\mathbf{H}) = 0$ holds. Then the mutual orthogonality of the three contributions as well as $Q_{\pm} = P_{\pm} + \tfrac{1}{2} P_{\pm,0}$ implies only the last term survives, 
		\begin{align*}
			(\mathbf{E},\mathbf{H}) &= 1_{(-\infty,0)} \bigl ( \pm \Maux_{\pm} \bigr ) (\mathbf{E},\mathbf{H})
			, 
		\end{align*}
		and consequently, $(\mathbf{E},\mathbf{H}) \in \ran 1_{(-\infty,0)} \bigl ( + \Maux_+ \bigr ) \cap 1_{(-\infty,0)} \bigl ( - \Maux_- \bigr )$ holds. However, the same arguments in the proof of \cite[Lemma~2.5]{DeNittis_Lein:ray_optics_photonic_crystals:2014} that allow us to conclude $\ran P_+ \cap \ran P_- = \{ 0 \}$ also ensure that the intersection 
		\begin{align*}
			\ran 1_{(-\infty,0)} \bigl ( + \Maux_+ \bigr ) \cap 1_{(-\infty,0)} \bigl ( - \Maux_- \bigr ) = \{ 0 \}
		\end{align*}
		is zero. Hence, the vector $(\mathbf{E},\mathbf{H}) = 0$ is zero and we have shown injectivity of $Q_{\pm} : L^2(\R^3,\R^6) \longrightarrow \Hil_{\pm}$. 
		\item In the proof of \cite[Lemma~2.5]{DeNittis_Lein:ray_optics_photonic_crystals:2014} we have shown that $2 \Re$ is injective on $\ran P_{\pm}$ by arguing that the intersection $\ran P_+ \cap \ran P_- = \{ 0 \}$ is trivial. 
	
		Therefore, it remains to show that $2 \Re$ is injective on the set $P_{\pm,0} \, \Re \bigl [ L^2(\R^3,\C^6) \bigr ] = P_{\pm,0} \bigl [ L^2(\R^3,\R^6) \bigr ]$. Suppose this is false and there exists a non-zero real vector field $(\mathbf{E},\mathbf{H})$ so that $P_{\pm,0} (\mathbf{E},\mathbf{H}) \neq 0$ while $2 \Re P_{\pm,0} (\mathbf{E},\mathbf{H}) = 0$. 
	
		In that case $P_{\pm,0} (\mathbf{E},\mathbf{H})$ is purely imaginary, 
		\begin{align*}
			\Psi_{\pm,0} &= \ii \, \Im \Psi_{\pm,0} 
			= \tfrac{1}{2} \bigl ( P_{+,0} - P_{-,0} \bigr ) (\mathbf{E},\mathbf{H})
			, 
		\end{align*}
		and we deduce $P_{+,0} (\mathbf{E},\mathbf{H}) = - P_{-,0} (\mathbf{E},\mathbf{H})$. Now the fact that both projections have the same range, $\mathcal{G} = \ran P_{\pm,0}$, applying $P_{-,0}$ to left- and right-hand side yields the eigenvalue equation 
		\begin{align*}
			P_{-,0} \, \bigl ( P_{+,0} (\mathbf{E},\mathbf{H}) \bigr ) = - P_{0,-}^2 (\mathbf{E},\mathbf{H})
			= - P_{0,-} (\mathbf{E},\mathbf{H}) 
		\end{align*}
		of $P_{0,-}$ to $-1$. However, orthogonal projections can only have $0$ and $+1$ as eigenvalues, so $P_{+,0} (\mathbf{E},\mathbf{H}) = 0$. The reality of $(\mathbf{E},\mathbf{H}) = C (\mathbf{E},\mathbf{H})$ and $C \, P_{\pm,0} \, C = P_{\mp,0}$ also lead us to conclude $P_{-,0} (\mathbf{E},\mathbf{H}) = C \, P_{+,0} (\mathbf{E},\mathbf{H}) = 0$. Thus, $2 \Re : \Hil_{\pm} \longrightarrow L^2(\R^3,\R^6)$ is injective. 
		\item This follows from (1) and the observation that the real-valuedness of $(\mathbf{E},\mathbf{H})$ implies $\Psi_- = Q_- (\mathbf{E},\mathbf{H}) = C \, Q_+ (\mathbf{E},\mathbf{H}) = C \Psi_+$. 
		\item The injectivity follows from $2 \Re \Psi_+ = \Psi_+ + \Psi_-$ as well as (2). 
	\end{enumerate}
\end{proof}
We will need one more essential ingredient for the proof of Theorem~\ref{Schroedinger:thm:equivalence_frameworks}, namely that the spectral decomposition according to the selfadjoint operator $M$ for time-dependent fields and currents coincides with that obtained by the Fourier transform in time. 
\begin{lemma}\label{appendix:1_to_1_correspondence:lem:equivalence:Fourier_transform_time_spectral_decomposition}
	Suppose we are in the setting of Theorem~\ref{Schroedinger:thm:equivalence_frameworks}. Then we have: 
	\begin{enumerate}[(1)]
		\item $W^{-1} \, J_+(t)$ defined as in \eqref{Maxwell_equations:eqn:complex_positive_negative_frequency_waves} via the Fourier transform coincides with 
		\begin{align*}
			J(t) = Q \, W^{-1} \, \mathbf{J}(t) = W^{-1} \, J_+(t)
			. 
		\end{align*}
		\item The same two statements hold for the solution $\Psi_+(t)$ to Maxwell's equations~\eqref{Maxwell_equations:eqn:approximate_Maxwell_equations} and the solution $\Psi(t)$ to the Schrödinger-type equation~\eqref{Schroedinger:eqn:pm_Schroedinger_equation_em}. 
	\end{enumerate}
\end{lemma}
\begin{proof}
	\begin{enumerate}[(1)]
		\item First of all, the Fourier transform of $J(t)$ is well-defined: because we may view $L^2(\R^3,\R^6)$ a subset of $L^2_W(\R^3,\C^6)$, and the operators $Q$ and $W^{-1}$ are bounded (the latter by Assumption~\ref{Schroedinger:assumption:material_weights}), the current density $J(t) = Q \, W^{-1} \, \mathbf{J}(t)$ inherits the Bochner-integrability of $\mathbf{J} \in L^1 \bigl ( L^2(\R^3,\C^6) \bigr )$ in $t$. 
		
		Moreover, because $W^{-1} \, \mathbf{J}(t) \in L^2_W(\R^3,\C^6)$ can be uniquely represented by the complex wave $J(t) = Q \, W^{-1} \, \mathbf{J}(t) \in \Hil \subset L^2_W(\R^3,\C^6)$, we can insert the spectral resolution of the selfadjoint operator $M$ to verify that 
		\begin{align*}
			J(t) &= \int_0^{\infty} \dd 1_{\omega}(M) \, W^{-1} \, \mathbf{J}(t)
		\end{align*}
		is indeed only supported on the spectrum of $M$ which, by definition, is contained in the subset $[0,\infty)$. 
		
		Given that $W^{-1}$ is time-independent, the Fourier components $\Fourier^{-1} \bigl ( W^{-1} \, \mathbf{J} \bigr ) = W^{-1} \, \Fourier^{-1} \mathbf{J}$ are just $W^{-1}$ times the Fourier components of $\mathbf{J}$; a similar equality holds for $J_+$. Moreover, because the resolution of the identity associated to $\Hil$ exactly coincides with the the fibration $\Hil = \int_0^{\infty} \dd \omega \, \Hil(\omega)$ (Lemma~\ref{Schroedinger:lem:identification_dispersive_nondispersive_Hilbert_spaces}), Assumption~\ref{Maxwell_equations:assumption:current_density} ($\widehat{J}(\omega) \in W \, \Hil(\omega)$) translates to 
		\begin{align*}
			\widehat{\mathbf{J}}(\omega') = W \, \dd 1_{\omega'}(M) \, W^{-1} \, \mathbf{J}(t) 
			= \widehat{J}_+(\omega') 
		\end{align*}
		for $\omega \geq 0$ and a similar expression involving the complex conjugated weights $\overline{W}$ for $\omega < 0$. However, the $\omega < 0$ does not matter as the spectral projections only filter out the $\omega \geq 0$ contributions, and we may replace $\widehat{\mathbf{J}}(\omega)$ by $\widehat{J}_+(\omega) = 1_{[0,\infty)}(\omega) \; \widehat{\mathbf{J}}(\omega)$ in the expression below, 
		\begin{align*}
			J(t) &= \frac{1}{\sqrt{2\pi}} \int_{\R} \dd \omega' \int_0^{\infty} \e^{- \ii t \omega'} \, \dd 1_{\omega}(M) \, W^{-1} \, \widehat{\mathbf{J}}(\omega')
			\\
			&= \frac{1}{\sqrt{2\pi}} \int_{\R} \dd \omega' \int_0^{\infty} \e^{- \ii t \omega'} \, \dd 1_{\omega}(M) \; \dd 1_{\omega'}(M) \, W^{-1} \, \mathbf{J}(t)
			\\
			&= \frac{1}{2\pi} \int_0^{\infty} \e^{- \ii t \omega} \, \dd 1_{\omega}(M) \, W^{-1} \, J_+(t)
			\\
			&= W^{-1} \, J_+(t) 
			. 
		\end{align*}
		That proves the claim. 
		\item This follows exactly as above from the spectral decomposition. 
	\end{enumerate}
\end{proof}
Now let us proceed to the proof of the main result of this paper. 
\begin{proof}[Theorem~\ref{Schroedinger:thm:equivalence_frameworks}]
	\begin{enumerate}[(1)]
		\item Suppose $\Psi(t)$ solves the Schrödinger-type equation~\eqref{Schroedinger:eqn:pm_Schroedinger_equation_em} and we now verify that it satisfies Maxwell's equations~\eqref{Maxwell_equations:eqn:approximate_Maxwell_equations}. Note that by Lemma~\ref{Schroedinger:lem:identification_dispersive_nondispersive_Hilbert_spaces} both equations are defined on the same Hilbert space. 
		Our assumptions on the material weights guarantee that $W$ and its inverse $W^{-1}$ are bounded operators. Then not only are the currents related by multiplication with $W^{\pm 1}$ (Lemma~\ref{appendix:1_to_1_correspondence:lem:equivalence:Fourier_transform_time_spectral_decomposition}~(1)), but also the dynamical equations~\eqref{Maxwell_equations:eqn:approximate_Maxwell_equations:dynamics} and \eqref{Schroedinger:eqn:pm_Schroedinger_equation_em} are. 
		
		The solution of the Schrödinger equation also satisfies the constraint equation~\eqref{Maxwell_equations:eqn:approximate_Maxwell_equations:constraint}: we split the solution into transversal and longitudinal parts according the Helmholtz decomposition introduced in Section~\ref{Schroedinger:auxiliary_operators:helmholtz} and verify that this constraint propagates in time, \ie since it is initially satisfied at $t_0$, then it will also remain satisfied for $t > t_0$. 
		
		The constitutive relations~\eqref{Maxwell_equations:eqn:approximate_Maxwell_equations:constitutive_relations} are not involved in the Schrödinger formalism. 
		
		Lastly, by assumption $\pmb{\rho}(t)$ and $\mathbf{J}(t)$ satisfy the charge conservation law~\eqref{Maxwell_equations:eqn:approximate_Maxwell_equations:charge_conservation}. Writing charge conservation in the frequency domain and using $J_+(t) = W \, J(t)$ (again by Lemma~\ref{appendix:1_to_1_correspondence:lem:equivalence:Fourier_transform_time_spectral_decomposition}~(1)), we conclude that also the current density satisfies charge conservation. 
		\medskip
		
		\noindent
		Conversely, assume $\Psi_+(t)$ solves Maxwell's equations~\eqref{Maxwell_equations:eqn:approximate_Maxwell_equations}. By the same argument as before, the two dynamical equations are equivalent and $\Psi_+(t)$ is also a solution to the Schrödinger-type equation~\eqref{Schroedinger:eqn:pm_Schroedinger_equation_em}. 
		\item This is a consequence of the injectivity of $Q : L^2(\R^3,\R^6) \longrightarrow \Hil$ (Lemma~\ref{appendix:1_to_1_correspondence:lem:1_to_1_correspondence}), which means that if we restrict the target space of the map $Q$ to its range, it is invertible. Therefore, it makes sense to write $\bigl ( E(t) , H(t) \bigr ) = Q^{-1} \, \Psi(t) \in L^2(\R^3,\R^6)$. 
	\end{enumerate}
\end{proof}
%

\section{Generalized Maxwell-type operators} 
\label{appendix:generalized_Maxwell_type_operators}
We briefly collect a few facts about generalized Maxwell-type operators. First of all, they are closed operators. 
\begin{lemma}\label{appendix:generalized_Maxwell_type_operators:lem:product_structure}
	Let $M = W_L \, D \, W_R$ be a generalized Maxwell-type operator in the sense of Definition~\ref{other_waves:defn:Maxwell_type_operator} with domain $\mathcal{D}_R$. Then $M$ is a closed operator.
\end{lemma}
\begin{proof}
	The operator $M$ is initially well defined on $\mathcal{D}_R$ and this implies that $\mathcal{D}_R \subseteq \mathcal{D}(M)$ is contained in the closure of $\mathcal{D}(M) := \overline{\mathcal{D}_R}^{\sNorm{\cdot}}$ with respect to the graph norm $\sNorm{\varphi}^2 := \snorm{\varphi}^2_{\Hil} + \snorm{M \varphi}^2_{\Hil}$ of $M$. 
	
	Therefore $\mathcal{D}_R = \mathcal{D}(M)$ follows if we show the opposite inclusion $\mathcal{D}_R \supseteq \mathcal{D}(M)$. By assumption the domain $\mathcal{D}_0$ is closed with respect to the topology induced by the graph norm $\sNorm{\varphi}_0^2 := \snorm{\varphi}^2_{\Hil} + \snorm{D \varphi}^2_{\Hil}$. Let $\psi_n := W_R^{-1} \varphi_n$ be a sequence in $\mathcal{D}_R$ (with $\varphi_n$ the related sequence in $\mathcal{D}_0$) which converges to a vector $\psi \in \mathcal{D}(M)$ with respect to the graph norm $\sNorm{\cdot}$. The estimate
	\begin{align*}
		\bNorm{\varphi_n - W_R \psi}_0^2 = \bNorm{W_R ( \psi_n - \psi )}_0^2 \leq \max \Bigl \{ \bnorm{W_L^{-1}}_{\mathbb{B}(\Hil)}^2 \, , \, \bnorm{W_R}_{\mathbb{B}(\Hil)}^2 \Bigr \} \; \sNorm{\psi_n - \psi}^2 
	\end{align*}
	tells us that $W_R \psi \in \mathcal{D}_0$ holds, and consequently, $W_R \bigl [ \mathcal{D}(M) \bigr ] \subseteq \mathcal{D}_0$. Put differently, $\mathcal{D}(M) \subseteq \mathcal{D}_R$ is shown.
\end{proof}
The closedness of $M$ enters the proof of selfadjointness. 
\begin{proof}[Proposition~\ref{other_waves:prop:selfadjointness_Maxwell_type_operator}]
	For each pair of vectors $\psi , \varphi \in \mathcal{D}_R$ the following computation
	\begin{align*}
		\bscpro{\psi}{M \varphi}_W &= \bscpro{\psi}{W_R \, D \, W_R \varphi}
		= \bscpro{W_R \, D \, W_R \psi}{\varphi}
		\\
		&= \bscpro{W_L^{-1} \, W_L \, D \, W_R \psi}{W_R \varphi}
		= \bscpro{M \psi}{\varphi}_W 
	\end{align*}
	shows that $M$ is a symmetric operator on $\Hil_W$. Since $\Hil_W$ agrees with $\Hil$ as Banach spaces, Lemma~\ref{appendix:generalized_Maxwell_type_operators:lem:product_structure} tells us that $M$ is also closed if seen as an operator on $\Hil_W$. 
	
	To show selfadjointness, we need to consider now the $\scpro{\, \cdot \,}{\, \cdot \,}_W$-adjoint of $M$ which we will denote with $M^{\ast_W}$ in order to distinguish it from the $\scpro{\, \cdot \,}{\, \cdot \,}$-adjoint $M^*$. Let $\phi , \eta \in \Hil_W$ be a pair of vectors such that $\phi \in \mathcal{D} \bigl ( M^{\ast_W} \bigr )$ and $M^{\ast_W} \phi = \eta$. This conditions implies that $\bscpro{\phi}{M \psi}_W = \scpro{\eta}{\psi}_W$ for all $\psi \in \mathcal{D}_R$ and this equality can be rewritten as $\bscpro{W_R \phi}{D \, W_R \psi} = \bscpro{W_L^{-1} \eta}{W_R \psi}$. Since $D$ is selfadjoint (by assumption) and $W_R \psi \in \mathcal{D}_0$ for all $\psi \in \mathcal{D}_R$ (by definition) the last equality implies that $W_R \phi\in \mathcal{D}_0$ and $D^* \, W_R \phi = D \, W_R \phi = W_L^{-1} \eta$. Hence, $\phi \in W_R^{-1}[\mathcal{D}_0] = \mathcal{D}_R$ holds and consequently, the domain of the adjoint $M^{\ast_W}$ agrees with the domain $\mathcal{D}_R$ of $M$. The last fact along with the symmetry of $M$ assures that $M$ is a selfadjoint operator on the weighted Hilbert space $\Hil_W$.
\end{proof}
Lastly, the idea to be able to uniquely write real fields as complex waves composed solely of non-negative frequencies generalizes to Maxwell-type operators. 
\begin{lemma}\label{appendix:generalized_Maxwell_type_operators:lem:1_to_1_correspondence}
	For Maxwell-type operators (\cf Definition~\ref{other_waves:defn:Maxwell_type_operator}) the analogs to Proposition~\ref{Schroedinger:prop:identification_real_complex_vector_spaces} and Lemma~\ref{appendix:1_to_1_correspondence:lem:1_to_1_correspondence} holds true. 
\end{lemma}
\begin{proof}
	In principle, we can adapt the strategy of the proof of Lemma~\ref{appendix:1_to_1_correspondence:lem:1_to_1_correspondence}, with the added simplification that we need not distinguish between $\Maux_{\pm}$. Instead $M = \Maux_{\pm}$, the kernel is \emph{by definition} $\mathcal{G} := \ker M = W_R^{-1} \, [\ker D]$ and the positive frequency fields are again \emph{defined} through $\mathcal{J}_+ := \ran 1_{(0,\infty)}(M)$. The tricky bit \cite[Lemma~2.5]{DeNittis_Lein:ray_optics_photonic_crystals:2014} simplifies tremendously, because we can deduce 
	\begin{align*}
		1_{(0,\infty)}(M) \, 1_{(-\infty,0)}(M) = 0 = 1_{(-\infty,0)}(M) \, 1_{(0,\infty)}(M)
	\end{align*}
	directly from functional calculus instead of having to use an indirect argument. 
\end{proof}
\begin{remark}
	Note that the last step simplifies because we are restricting ourselves to what would be called the non-gyrotropic case in electromagnetism. This is an inessential restriction, which can be lifted by following the line of argumentation in the proof of Lemma~\ref{appendix:1_to_1_correspondence:lem:1_to_1_correspondence}. 
\end{remark}
%

\printbibliography

\end{document}